\documentclass[a4paper,12pt]{article}
\usepackage{geometry}
\usepackage{amsmath}
\usepackage{amsthm}
\usepackage{authblk}
\usepackage{hyperref}
\usepackage{bm}
\usepackage{cite}
\usepackage{mathrsfs}
\usepackage{amssymb}
\usepackage{amsfonts}
\usepackage{graphicx}
\usepackage{subcaption}
\usepackage{tikz}
\usepackage{mathtools}

\usetikzlibrary{intersections}
\allowdisplaybreaks

\theoremstyle{plain}
\newtheorem{prop}{Proposition}[section]

\newtheorem{theo}{Theorem}[section]
\newtheorem{lem}{Lemma}[section]

\newtheorem{assu}{Assumption}[section]
\newtheorem{rhp}{RH problem}[section]
\newtheorem{Dbar}{$\bar{\partial}$ problem}[section]
\newtheorem{rhp-Dbar}{RH-$\bar{\partial}$ problem}[section]

\def\rd{{\rm d}}
\def\ri{{\rm i}}

\allowdisplaybreaks
\newtheorem{remark}{Remark}[section]

\title{\bf Asymptotic stability of solutions to the good Boussinesq equation in dispersive wave region}

\author{Yingmin Yang\thanks{Corresponding author: yangym@mail.bnu.edu.cn}}

\affil{\small School of Mathematical Sciences,
	Beijing Normal University, Beijing 100875, China}

\date{}

\begin{document}
	\maketitle
	
		\begin{abstract}
		This work studies the asymptotic stability of solutions to the good Boussinesq equation in dispersive wave region when the reflection coefficients associated with the initial data belong to weighted Sobolev space. The $\bar{\partial}$-steepest descent method is applied to the Riemann-Hilbert problem and the long-time asymptotic expansion of the solution are obtained up to an optimal error of order $\mathcal{O}(t^{-3/4})$. Compared with previous results, we extend the initial data from the rapidly decaying Schwartz space to a weighted Sobolev space, and prove the asymptotic stability of the solution in dispersive wave region.\\
\par
\vspace{6pt}
		\end{abstract}
		
	  \section{Introduction}
	  \ \ \ \
	  Investigating the long-time asymptotic behaviors and asymptotic stability of solutions to integrable nonlinear evolution equations is a core and challenging topic \cite{CMP_2002,JPSJ_1979}. This research not only deeply reveals the intrinsic laws of nonlinear physical phenomena but also provides key examples for verifying and developing nonlinear mathematical theories. As a fundamental model describing a class of important physical processes, the long-time behavior analysis of solutions to the good Boussinesq equation \cite{PLA_1983,SIAM_JSSC_1984,MC_1991,JMP_1988} has long attracted significant attention. Traditional methods such as the Inverse Scattering Transform \cite{SIAM_1981,Studies_1974,Solitons_1980} often face bottlenecks when dealing with such problems due to the complexity of spectral issues and the difficulties in long-time asymptotic analysis. In recent years, the $\bar \partial$-nonlinear steepest descent method \cite{IMRP_2006,IMRN_2008} has emerged as a powerful and more universal analytical tool, demonstrating great potential in solving such problems. This paper will systematically study the long-time asymptotic behavior of solutions to the good Boussinesq equation using the $\bar \partial$-nonlinear steepest descent method.
	
	  The standard form of the good Boussinesq equation \cite{Indiana_2022,lenells-wang-good,Studies_gB_2024} is
	  \begin{equation}\label{gB_uxx}
	  	u_{tt}-u_{xx}+(u^2)_{xx}+u_{xxxx}=0,
	  \end{equation}
	  which holds profound physical significance in fluid dynamics \cite{fluid_CRM_2017} and nonlinear lattice dynamics. It was initially proposed as a ``well-posed" modification of the Boussinesq-type equations \cite{WaveMotion_2001, AMC_2007} for describing shallow water wave propagation in compressible fluids, avoiding the short-wave instability that may arise in its ``bad" Boussinesq counterpart . Furthermore, it also accurately describes nonlinear waves in nematic liquid crystals and long-wave vibrations in one-dimensional atomic chains.
	
	  This paper considers the following version of the good Boussinesq equation
	  \begin{equation}\label{gB}
	  	u_{tt}+\frac{4}{3}\left(u^2 \right)_{xx}+\frac{1}{3}u_{xxxx}=0,
	  \end{equation}
	  which can be obtained through a simple transformation $u\to u+\frac{1}{2}$ from the good Boussinesq \eqref{gB_uxx}. Equation \eqref{gB} can be rewritten as
	  \begin{equation}\label{regB}
	  	\left\lbrace
	  	\begin{aligned}
	  		&w_t+\frac{1}{3}u_{xxx}+\frac{4}{3}\left(u^2 \right)_x=0,\\
	  		&u_t=w_x,
	  	\end{aligned}\right.
	  \end{equation}
	  and possesses a $3\times3$ Lax pair. In recent years, significant progress has been made in the study of the equation \eqref{gB}. Researchers have utilized the inverse scattering transform \cite{CPAM_1982,Indiana_2022} to construct multi-soliton solutions \cite{ADE_2020} and rational solutions \cite{TMA_2017}, and have analyzed their stability \cite{CMP_1988}. Regarding asymptotic behavior, the classical Deift-Zhou method \cite{Annals_1993} has been applied by analyzing a regularized version of the RH problem and then transferring the results to the singular problem. This approach has provided the long-time asymptotic behavior for equation \eqref{gB} in \cite{lenells-wang-good} under Schwartz initial data with parameter $x/t$ within the compact subsets of $(0,\infty)$. However, this reduction process may sometimes obscure the essential structure of the problem and might not be applicable to more general boundary conditions or potential functions \cite{Advances_Tian_2022,CTP_2024,JDE_Geng_2024,AHPoincare_2022}.

	  The Deift-Zhou method relies on the ability to analytically extend the scattering matrix corresponding to the spectral problem to specific regions of the complex plane, thereby constructing a Riemann-Hilbert (RH) problem \cite{NoticesAMS_2003,SIAMJMA_1989,PRIMS_1984}. The $\bar \partial$-nonlinear steepest descent method \cite{Poincare_2018,Studies_2021,McLaughlin_2008}  emerged as a fundamental generalization of the Deift-Zhou framework. The core of this method lies in recognizing that the scattering data (or the extended matrix function) is not analytic over the entire complex plane but instead satisfies a $\bar \partial$-problem. By directly analyzing and controlling this $\bar \partial$-problem, we can bypass the strong constraint of analyticity, thereby enabling the treatment of a broader class of equations.
	
	  Although the $\bar \partial$-method offers significant advantages, its application, particularly in higher-order (e.g., $3\times3$) Lax pair problems \cite{JMP_2022,Tzitzeica_2024} are few. Most successful cases have been concentrated on $2\times2$ systems \cite{Advances_2022,CMP_2023,Nonlinearity_2017,Sigma_2022,Advances_Fan_2023,JDE_Tian_2023,Proceedings_2025,PhysicaD_2023}. For $3\times3$ systems \cite{Studies_Wang_2024,Studies_Fan_2016,Acta_2021}, the spectral structure is more complex, and the expressions for the $\bar \partial$-derivatives, along with the corresponding asymptotic analysis, become considerably more intricate. This constitutes a current cutting-edge challenge and difficulty in ongoing research.

	  This study constitutes a substantive endeavor to apply and advance the $\bar \partial$-method to higher-order Lax pairs. By comprehensively addressing the $3\times3$ $\bar \partial$-problem for the good Boussinesq equation, we will establish a standardized framework for tackling such complex systems. This framework will encompass the procedures for formulating the $\bar \partial$-problem, designing the steepest descent transformations, and performing precise error estimates, thereby providing a methodological reference for future research on more intricate $3\times3$ integrable systems.
	
	  Compared to the conclusions obtained in Ref. \cite{lenells-wang-good} using the Deift-Zhou nonlinear steepest descent method, this paper imposes significantly looser requirements on the initial data. It extends the initial data space from the Schwartz space to a weighted Sobolev space to which the scattering coefficients (obtained via the nonlinear Fourier transform of the initial data) belong. Consequently, this paper possesses stronger universality and can handle a broader range of initial data with slower decay. Furthermore, regarding the asymptotic expansion derived by the $\bar \partial$-method, the error term is obtained by analyzing the integral equation for the solution of the $\bar \partial$-problem. This approach focuses more on proving the existence and convergence of the solution, as well as describing its qualitative behavior. Finally, the conceptual framework of the $\bar \partial$-method is more readily generalizable to integrable systems with higher-order Lax pair, as exemplified in this paper by the good Boussinesq equation $\eqref{regB}$ which possesses a 3$\times$3 Lax pair.
	
	  The structure of this paper is organized as follows. The main results are presented in Section \ref{sec_2}. Section \ref{sec_dispersive} is devoted to a series of deformations of the original RH problem. Subsequently, the work in Section \ref{sec_4} focuses on constructing local parametrices at the three critical points, estimating the resulting small-norm RH problem, and deriving the asymptotic leading term of the solution to the RH problem. In Section \ref{sec_5}, error estimates are carried out for the two $\bar \partial$-problems. Finally, the last section provides the long-time asymptotic behavior of the solution to the good Boussinesq equation with initial data in a weighted Sobolev space, including detailed expressions for each asymptotic term and the corresponding error estimates.

	  \section{Main result}\label{sec_2}
	  \ \ \ \
	  Our results are almost entirely composed of two reflection coefficients
	  \begin{equation*}
	  	\left\lbrace
	  	\begin{aligned}
	  		&r_1(k)=\frac{s_{12}(k)}{s_{11}(k)}:\, (0,\infty)\to\mathbb{C},\\
	  		&r_2(k)=\frac{s^A_{12}(k)}{s^A_{11}(k)}:\, (-\infty,0)\to\mathbb{C},
	  	\end{aligned}
	  	\right.
	  \end{equation*}
	   which are determined by the initial values $u_0(x):=u(x,0)$ and $w_0(x):=w(x,0)$. Functions $r_1(k)$, $r_2(k)$, $s(k)$ and $s^A(k)$ mentioned above are fully defined in Ref. \cite{Indiana_2022}, and this paper also adopts the same construction. Define the weighted Sobolev space
	   {\begin{equation}\label{sobolev}
	  	H^{3,4}=\left\lbrace f(x)\in L^2\left| \right.\left(1+\left|x \right|^4  \right)\partial^j f(x)\in L^2, \text{for}\,\, j=1,2,3 \right\rbrace ,
	  \end{equation}
	  and in the following discussion, we assume that $r_1(k)\in H^{3,4}(0,\infty)$ and $r_2(k)\in H^{3,4}(-\infty
	  ,0)$, then they can be continuously extended to $k=0$. The selection of this space will be explained in the next subsection.}

	  \subsection{Riemann-Hilbert problem for $M(x,t,k)$ of the good Boussinesq equation} \ \ \ \
	  The RH problem corresponding to the good Boussinesq equation \eqref{gB} has already been fully constructed in Ref. \cite{Indiana_2022}. Here, we state some of the main content of the RH problem.
	
	  \begin{rhp}\label{rhp_M}
	  	Find a $3 \times 3$-matrix valued function $M(x, t, k)$ with the following properties:
	  	
	  	\begin{enumerate}
	  		\item The function $M(x, t, \cdot) : \mathbb{C} \setminus \Gamma \rightarrow \mathbb{C}^{3 \times 3}$ is analytic, where $\Gamma$ is defined by Fig. \ref{figGamma}.
	  		\item As $k$ approaches $\Gamma \setminus \{0\}$ from the left and right, the boundary values $M_+$ and $M_-$ of $M$ exist and are continuous on $\Gamma \setminus \{0\}$. Furthermore, they are related by
	  		\begin{equation*}
	  			M_+(x, t, k) = M_-(x, t, k) v(x, t, k), \quad k \in \Gamma,
	  		\end{equation*}
	  		where $v(x,t,k)=v_j(x, t, k)$ for $k \in \Gamma_j$,
	  		\begin{equation}\label{v}
	  			\begin{aligned}
	  				&v_1 = \begin{pmatrix}
	  					1 & -r_1(k){\rm{e}}^{-\theta_{21}} & 0 \\
	  					r_1^*(k){\rm{e}}^{\theta_{21}} & 1 - |r_1(k)|^2 & 0 \\
	  					0 & 0 & 1
	  				\end{pmatrix}, \quad
	  				v_2 = \begin{pmatrix}
	  					1 & 0 & 0 \\
	  					0 & 1 - r_2(\omega k)r_2^*(\omega k) & -r_2^*(\omega k){\rm{e}}^{-\theta_{32}} \\
	  					0 & r_2(\omega k){\rm{e}}^{\theta_{32}} & 1
	  				\end{pmatrix}, \\
	  				&v_3 = \begin{pmatrix}
	  					1 - r_1(\omega^2 k)r_1^*(\omega^2 k) & 0 & r_1^*(\omega^2 k){\rm{e}}^{-\theta_{31}} \\
	  					0 & 1 & 0 \\
	  					-r_1(\omega^2 k){\rm{e}}^{\theta_{31}} & 0 & 1
	  				\end{pmatrix}, \quad
	  				v_4 = \begin{pmatrix}
	  					1 - |r_2(k)|^2 & -r_2^*(k){\rm{e}}^{-\theta_{21}} & 0 \\
	  					r_2(k){\rm{e}}^{\theta_{21}} & 1 & 0 \\
	  					0 & 0 & 1
	  				\end{pmatrix}, \\
	  				&v_5 = \begin{pmatrix}
	  					1 & 0 & 0 \\
	  					0 & 1 & -r_1(\omega k){\rm{e}}^{-\theta_{32}} \\
	  					0 & r_1^*(\omega k){\rm{e}}^{\theta_{32}} & 1 - r_1(\omega k)r_1^*(\omega k)
	  				\end{pmatrix}, \quad
	  				v_6 = \begin{pmatrix}
	  					1 & 0 & r_2(\omega^2 k){\rm{e}}^{-\theta_{31}} \\
	  					0 & 1 & 0 \\
	  					-r_2^*(\omega^2 k){\rm{e}}^{\theta_{31}} & 0 & 1 - r_2(\omega^2 k)r_2^*(\omega^2 k)
	  				\end{pmatrix},
	  			\end{aligned}
	  		\end{equation}
	  		and $\theta_{ij}:=\theta_{ij}(x, t, k) = (l_i - l_j)x + (z_i - z_j)t$, $l_j(k) = \omega^j k$, $z_j(k) = \omega^{2j} k^2$, $\omega={\rm{e}}^{\frac{2\pi \ri}{3}}$.
	  		 {\item As $k \rightarrow \infty$, for
	  	 $k \in \mathbb{C}\setminus\Gamma$,
	  	 \begin{equation}\label{M_ktoinfty}
	  	 	M(x,t,k)=I+\frac{M^{(1)}(x,t)}{k}+\frac{M^{(2)}(x,t)}{k^2}+\mathcal{O}\left(\frac{1}{k^3} \right), 
	  	 \end{equation}
	  	 where the matrices $M^{(1)}$ and $M^{(2)}$ satisfy $M^{(1)}_{12}=M^{(1)}_{13}=M^{(2)}_{12}+M^{(2)}_{13}=0$.
	  	 	\item As $k\to 0$,  for $k \in \bar{D}_{1}$,
	  	 	\begin{equation}\label{M_kto0}
	  	 			M(x, t, k) = \sum_{l=-2}^{1} \mathcal{M}_{1}^{(l)}(x, t) k^{l} + \mathcal{O}(k^{2}),
	  	 	\end{equation}
	  	 where
	  			\begin{align*}
	  				&\mathcal{M}_1^{(-2)}(x,t)=\alpha(x,t)\begin{pmatrix}
	  					\omega & 0 & 0\\
	  					\omega & 0 & 0\\
	  					\omega & 0 & 0\\
	  				\end{pmatrix},\\
	  				&\mathcal{M}_1^{(-1)}(x,t)=\beta(x,t)\begin{pmatrix}
	  					\omega^2 & 0 & 0\\
	  					\omega^2 & 0 & 0\\
	  					\omega^2 & 0 & 0\\
	  				\end{pmatrix}+\gamma(x,t)\begin{pmatrix}
	  				\omega^2 & 0 & 0\\
	  				1 & 0 & 0\\
	  				\omega & 0 & 0\\
	  				\end{pmatrix}+\delta(x,t)\begin{pmatrix}
	  				0 & 1-\omega & 0\\
	  				0 & 1-\omega & 0\\
	  				0 & 1-\omega & 0\\
	  				\end{pmatrix},
	  			\end{align*}
	  			and the third column of $\mathcal{M}_1^{(0)}(x,t)$ is given by
	  			\begin{equation*}
	  				[\mathcal{M}_1^{(0)}(x,t)]_3=\epsilon(x,t)(1\,\, 1\,\, 1)^T,
	  			\end{equation*}
	  			holds for the scalar functions $\alpha$, $\beta$, $\gamma$,  $\delta$ and $\epsilon$ that are functions of $x$ and $t$ independent of the variable $k$.
	  			}
	  		\item $M$ satisfies the symmetries
	  		\begin{equation}\label{sym}
	  			M(x, t, k) = \mathcal{A} M(x, t, \omega k) \mathcal{A}^{-1} = \mathcal{B} \overline{M(x, t, \bar{k})} \mathcal{B}, \quad k \in \mathbb{C} \setminus \Gamma,
	  		\end{equation}
	  		where
	  		\begin{equation*}
	  			\mathcal{A} := \begin{pmatrix} 0 & 0 & 1 \\ 1 & 0 & 0 \\ 0 & 1 & 0 \end{pmatrix}, \quad  \quad \mathcal{B} := \begin{pmatrix} 0 & 1 & 0 \\ 1 & 0 & 0 \\ 0 & 0 & 1 \end{pmatrix}.
	  		\end{equation*}
	  	\end{enumerate}
	  \end{rhp}
	  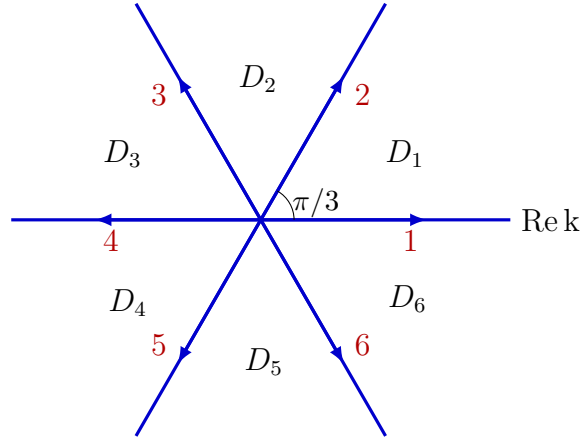
\begin{figure}[htbp]
	  	\centering
	  	\begin{tikzpicture}[scale=1.1]
	  		\draw[very thick, black!20!blue, -latex] (0,0) -- (2,0);
	  		\draw [very thick,black!20!blue](0,0) -- (3,0);
	  		\draw[very thick, black!20!blue, latex-] (-2,0) -- (0,0);
	  		\draw [very thick,black!20!blue](-3,0) -- (0,0);
	  		\draw[very thick, black!20!blue, -latex] (0,0) -- (1,1.732);
	  		\draw [very thick,black!20!blue](0,0) -- (1.5,1.5*1.732);
	  		\draw[very thick, black!20!blue, latex-] (-1,-1.732) -- (0,0);
	  		\draw [very thick,black!20!blue](-1.5,-1.5*1.732) -- (0,0);
	  		\draw[very thick, black!20!blue, -latex] (0,0) -- (-1,1.732);
	  		\draw [very thick,black!20!blue](0,0) -- (-1.5,1.732*1.5);
	  		\draw[very thick, black!20!blue, latex-] (1,-1.732) -- (0,0);
	  		\draw [very thick,black!20!blue](0,0) -- (1.5,-1.5*1.732);
	  		\node[red!70!black,below] at (1.8,0) {$1$};
	  		\node[red!70!black,right] at (1,1.5) {$2$};
	  		\node[red!70!black,left] at (-1,1.5) {$3$};
	  		\node[red!70!black,below] at (-1.8,0) {$4$};
	  		\node[red!70!black,left] at (-1,-1.5) {$5$};
	  		\node[red!70!black,right] at (1,-1.5) {$6$};
	  		\node[right] at (3,0) {${\rm{Re}\,k}$};
	  		\draw (0:0.4cm) arc (0:60:0.4cm);
	  		\node[right] at (0.25,0.2) {\small$\pi/3$};
	  		\node at (1.732,0.8) {$D_1$};
	  		\node[right] at (-0.4,1.7) {$D_2$};
	  		\node[left] at (-1.3,0.8) {$D_3$};
	  		\node[below] at (-1.6,-0.7) {$D_4$};
	  		\node[left] at (0.4,-1.7) {$D_5$};
	  		\node[right] at (1.4,-0.95) {$D_6$};
	  	\end{tikzpicture}
	  	\caption{The jump contour $\Gamma$ decomposes the complex $k$-plane into six parts.}
	  	\label{figGamma}
	  \end{figure}

	   {\begin{remark}
	  	Regarding the existence and uniqueness of solutions to RH problem \ref{rhp_M} and the selection of the space to which the reflection coefficients belong, we provide the following explanations:
	  	\begin{description}
	  		\item[{\footnotesize Existence and uniqueness of solutions:}] It has been discussed in Ref.~\cite{lenells-2018} that results concerning the existence of solutions to RH problems are relatively rare and often depend on specific symmetries, and a theory for $n \times n$ matrix RH problems has been developed, which the case $n=3$ in this paper still applies. Ref.~\cite{Indiana_2022} discusses the existence of solutions to the RH problem associated with the good Boussinesq equation and due to its different symmetry conditions compared to the bad Boussinesq equation, to the best of our knowledge, the vanishing lemma has not yet been established. Ref.~\cite{Bona-1988} proves that the good Boussinesq equation is globally well-posed on the real line for initial data $(u(x,0),u_t(x,0))\in H^s\times H^{s-2}$ with $s>\frac{5}{2}$. During the steepest descent analysis, the existence and uniqueness of the solution will be established by analyzing a small-norm RH problem. In Ref.~\cite{Indiana_2022}, it was proved that for the good Boussinesq equation with initial data in Schwartz space, if a solution to the corresponding RH problem exists, then it is unique. In this paper, for the given weighted Sobolev space in (\ref{sobolev}), the uniqueness of the solution to the RH problem is given in Proposition \ref{prop_solution_exist}, and a brief proof is completed.
	  		\item[{\footnotesize Selection of the functional space:}]In the first point, we have explained that the good Boussinesq equation is globally well-posed for initial data $(u(x,0),u_t(x,0))\in H^s\times H^{s-2}$ with $s>\frac{5}{2}$ in \cite{Bona-1988}. Furthermore, in Ref.~\cite{Indiana_2022}, during the direct scattering process for the initial value problem, it is necessary to analyze the behavior of the eigenfunctions as the spectral parameter $k\to\infty$ and $k\to0$, which involves integration by parts on Volterra integral equations. In Proposition \ref{prop_reconst} described below, we need to ensure that the series expansion \eqref{M_ktoinfty} converges uniformly with respect to the variables $x$ and $t$ as $k \to \infty$, and that the series expansion \eqref{M_kto0} converges uniformly with respect to $x$ and $t$ as $k \to 0$. Based on these comprehensive considerations, we specify that the reflection coefficients belong to the weighted Sobolev space $H^{3,4}$. In fact, in the subsequent treatment of the $\bar{\partial}$-problem, our requirements for the reflection coefficients $r_1$ and $r_2$ are merely $k\,r_1(k)\in L^2(0,\infty)$ and $k\,r_2(k)\in L^2(-\infty,0)$; the weighting introduced here is to ensure the existence of solutions to the good Boussinesq equation.
	  	\end{description}
	  	Nevertheless, assuming the existence of the solution to the RH problem remains indispensable in certain propositions and theorems. The main contribution of this paper is that, compared to Ref.~\cite{lenells-wang-good}, we make weaker assumptions on the initial data, improve the error estimates, and provide a systematic framework for the application of the $\bar{\partial}$-steepest descent method to $3 \times 3$ spectral problems. Future work will also focus on the bijection problem for the good Boussinesq equation, investigating whether the assumptions on the initial data can be further relaxed.
	  \end{remark}}

	   {\begin{remark}
	  		In RH problem \ref{rhp_M}, condition 4 clearly states that the matrix eigenfunction $M(x,t,k)$ has a second-order pole as the spectral parameter $k\to0$, and specifies the exact form of this singularity. To address the difficulties caused by this singularity, introduce the vector eigenfunction $N(x,t,k)=(\omega \,\, \omega^2 \,\, 1)M(x,t,k)$. It is straightforward to verify that $N(x,t,k)$ satisfies the RH problem \ref{rhp_n} described below; for further details, see Ref. \cite{Indiana_2022}. Subsequently, when applying the $\bar\partial$-steepest descent method to analyze the RH problem, we will continue to work with $N(x,t,k)$, which remains non-singular as $k\to0$.
	  \end{remark}}
  
   {\begin{rhp}\label{rhp_n}
  		Find a $1 \times 3$-row-vector valued function $N(x, t, k)$ with the following properties:
  		\begin{enumerate}
  			\item The function $N(x, t, \cdot) : \mathbb{C} \setminus \Gamma \rightarrow \mathbb{C}^{1 \times 3}$ is analytic.
  			\item As $k$ approaches $\Gamma \setminus \{0\}$ from the left and right, the boundary values $N_+$ and $N_-$ of $n$ exist and are continuous on $\Gamma \setminus \{0\}$. Furthermore, they are related by
  			\begin{equation*}
  				N_+(x, t, k) = N_-(x, t, k) v(x, t, k), \quad k \in \Gamma,
  			\end{equation*}
  		   where $v(x,t,k)$ is given by equation \eqref{v}.
  			\item As $k \rightarrow \infty$,
  			$N(x, t, k) = (\omega \,\, \omega^2 \,\, 1) + \mathcal{O}\left(k^{-1} \right) $ for $k \in \mathbb{C}\setminus\Gamma$.
  			\item As $k\to 0$, $
  				N(x, t, k) =\mathcal{O}(1)$.
  			\item $N$ satisfies the symmetries
  			\begin{equation*}
  				N(x, t, k) = \omega\, N(x, t, \omega k) \mathcal{A}^{-1} = \overline{N(x, t, \bar{k})} \mathcal{B}, \quad k \in \mathbb{C} \setminus \Gamma.
  			\end{equation*}
  		\end{enumerate}
  \end{rhp}}

	  \subsection{Statement of the main result}
	  \ \ \ \
	  Our work is primarily based on the weighted Sobolev space \eqref{sobolev}, considering $\zeta=x/t$ in compact subsets of $(0,\infty)$. And before presenting the most important results of this paper, we make the following assumptions.
	
	  \begin{assu}\label{assu_solitonless}
	  Assume that $s_{11}(k)$ and $s^A_{11}(k)$ are nonzero for $ k \in \bar{D}_1 \setminus \{0\}$ and $k \in \bar{D}_4 \setminus \{0\}$, respectively.
	  \end{assu}
	
	  \begin{assu}\label{assu_k=0}
	  	Assume that
	  	\begin{equation*}
	  		\lim_{k \to 0} k^2 s_{11}(k) \neq 0, \quad \lim_{k \to 0} k^2 s^A_{11}(k) \neq 0.
	  	\end{equation*}
	  \end{assu}
	
	   {The two assumptions above imply that no solitons are present and that $s_{11}(k)$ and $s^A_{11}(k)$ have double poles at $k=0$. Moreover, the corresponding functions $r_1(k)\in H^{3,4}(0,\infty)$ and $r_2(k)\in H^{3,4}(-\infty,0)$ satisfy Assumptions \ref{assu_solitonless} and \ref{assu_k=0}.}

	   {\begin{prop}\label{prop_solution_exist}
	  		Let the reflection coefficients $r_1(k)\in H^{3,4}(0,\infty)$ and  $r_2(k)\in H^{3,4}(-\infty,0)$ such that Assumptions \ref{assu_solitonless} and \ref{assu_k=0} hold, and that $\lim\limits_{k\to0^+}r_1(k)=\omega$ and $\lim\limits_{k\to0^-}r_2(k)=1$. Then, the solution of the RH problem \ref{rhp_M} is unique, if it exists.
	\end{prop}
	\begin{proof}
		In Appendix A of Ref. \cite{Indiana_2022}, the uniqueness of the solution to the RH problem corresponding to the initial value problem of the good Boussinesq equation \eqref{gB}, with initial data belonging to the Schwartz space, has been proved in detail. Here, we consider reflection coefficients $r_1$ and $r_2$, which have different properties from those in \cite{Indiana_2022}, and explain the differences in the proof of the uniqueness of the solution to the RH problem \ref{rhp_M}.\\
        \indent  Consider $r_1(k)\in H^{3,4}(0,\infty)$ and  $r_2(k)\in H^{3,4}(-\infty,0)$. From the direct scattering theory, the eigenfunction $M(x,t,k)$ involved in the RH problem \ref{rhp_M} possesses at least the following asymptotic expansion as $k \to \infty$:
        \begin{equation*}
        	\begin{aligned}
        		M(x,t,k)=&I+\frac{M^{(1)}_{33}(x,t)}{k}\begin{pmatrix}
        		\omega^2 & 0 & 0\\
        		0 & \omega & 0\\
        		0 & 0 & 1
        	\end{pmatrix}+\frac{M^{(2)}_{33}(x,t)}{k^2}\begin{pmatrix}
        		\omega & 0 & 0\\
        		0 & \omega^2 & 0\\
        		0 & 0 & 1
        	\end{pmatrix}\\
        	&-\frac{M^{(2)}_{13}(x,t)}{(1-\omega)k^2}\begin{pmatrix}
        		0 & 1 & -1\\
        		-\omega & 0 & \omega\\
        		\omega^2 & -\omega^2 & 0
        	\end{pmatrix}+\mathcal{O}(k^{-3}).
        	\end{aligned}
        \end{equation*}
      And $M$ has a unit determinant, according to the properties satisfied in RH problem \ref{rhp_M} as $k\to\infty$ and $k\to0$. \\
      \indent Assuming the existence of two solutions, $M$ and $N$, to the RH problem \ref{rhp_M}, it can be obtained from the above properties that $\det M$ and $\det N$ are identically equal to one. Furthermore, the asymptotic forms of $N^A:=(N^{-1})^T$ as $k\to0$ and $k\to\infty$ can be obtained as follows:
      \begin{align*}
      	N^A(x,t,k)=&\frac{1}{k^2}\begin{pmatrix}
      		0 & 0 & *\\
      		0 & 0 & *\\
      		0 & 0 & *\\
      	\end{pmatrix}+\frac{1}{k}\begin{pmatrix}
      	0 & * & *\\
      	0 & * & *\\
      	0 & * & *\\
      	\end{pmatrix}+\mathcal{O}(1),\quad k\in D_1, \,\, k\to0,\\
     	N^A(x,t,k)=&I-\frac{N^{(1)}_{33}(x,t)}{k}\begin{pmatrix}
      	\omega^2 & 0 & 0\\
      	0 & \omega & 0\\
      	0 & 0 & 1\\
      \end{pmatrix}+\frac{N^{(1)}_{33}(x,t)^2-N^{(2)}_{33}(x,t)}{k^2}\begin{pmatrix}
      \omega & 0 & 0\\
      0 & \omega^2 & 0\\
      0 & 0 & 1\\
      \end{pmatrix}\\
      &-\frac{N^{(2)}_{13}(x,t)}{k^2}\begin{pmatrix}
      	0 & \omega & -\omega^2\\
      	-1 & 0 & \omega^2\\
      	1 & -\omega & 0\\
      \end{pmatrix}+\mathcal{O}(k^{-3}),\quad k\in D_1, \,\, k\to\infty,
      \end{align*}
      where $*$ represents matrix elements that are irrelevant to the current argument.\\
      \indent According to the above expansion, it can be concluded that $MN^{-1}$ has at most a double pole at $k=0$ and is analytic for $k\in\mathbb{C}\setminus\{0\}$. It satisfies that it approaches the identity matrix as $k\to\infty$. Furthermore, by utilizing symmetries \eqref{sym} and considering the properties of $MN^{-1}$ as $k\to0$ and $k\to\infty$, we can obtain:
      \begin{equation*}
      	(MN^{-1})(x,t,k)=I+\frac{M^{(1)}_{33}(x,t)-N^{(1)}_{33}(x,t)}{k}\begin{pmatrix}
      		\omega^2 & 0 & 0\\
      		0 & \omega & 0\\
      		0 & 0 & 1\\
      	\end{pmatrix}.
         \end{equation*}
         On the left-hand side of the above equation, we have $\det MN^{-1}=1$, and since $\det M=\det N=1$, it follows that $M^{(1)}_{33}(x,t)=N^{(1)}_{33}(x,t)$. That is to say, we have proved that $M=N$. Therefore, if a solution to RH problem \ref{rhp_M} exists, then the solution is unique.
\end{proof}}

	   {\begin{prop}\label{prop_reconst}
	  		Let the reflection coefficient 
	  	coefficients $r_1(k)\in H^{3,4}(0,\infty)$ and $r_2(k)\in H^{3,4}(-\infty,0)$ such that Assumptions 
    \ref{assu_solitonless} and \ref{assu_k=0} hold. Then the solution of the RH problem \ref{rhp_M} is unique whenever it exists. If the initial data $u_0(x)$ and $w_0(x)$ of the good Boussinesq equation \eqref{gB} satisfy Assumption \ref{assu_k=0}, the solutions $u(x,t)$ and $w(x,t)$ for all $(x,t)$ in an open subset of $\mathbb{R} \times [0, \infty) $ can be reconstructed from the RH problem \ref{rhp_M} in the form:
	  \begin{equation}\label{reconst}
	  		\left\lbrace
	  		\begin{aligned}
	  			&u(x,t)=-\frac{3}{2}\frac{\partial}{\partial x}\lim_{k\to\infty}k(N_3(x,t,k)-1),\\
	  			&w(x,t)=-\frac{3}{2}\frac{\partial}{\partial t}\lim_{k\to\infty}k(N_3(x,t,k)-1),
	  		\end{aligned}\right.
	  	\end{equation}
	where $N(x,t,k):=(N_1(x,t,k),N_2(x,t,k),N_3(x,t,k))$.
	  \end{prop}}

	   {\begin{theo}\label{theo_asymptotics}
	  	For initial values $u_0(x)$ and $w_0(x)$ such that the reflection coefficients $r_1(k)\in H^{3,4}(0,\infty)$ and $r_2(k)\in H^{3,4}(-\infty,0)$, and assuming Assumptions \ref{assu_solitonless} and \ref{assu_k=0} hold. Then the following asymptotic formula holds uniformly for $\zeta=x/t$ in compact subsets of $(0,\infty)$ as $t\to \infty$
	  	\begin{align*}
	  		&\left| u(x,t)+\frac{3^{5/4}k_0\sqrt{\nu}}{\sqrt{2t}}\sin\left(\frac{19\pi}{12}+\nu\ln \left(6\sqrt{3}tk_0^2 \right) -\sqrt{3}tk_0^2-\arg r_1(k_0)  \right. \right.\\
	  		&\left.-\arg\Gamma\left(i\nu \right)+ \left.\frac{1}{\pi} \int_{k_0}^{\infty}\ln\left(\frac{\left|s-k_0 \right| }{\left| s-\omega k_0\right| } \right)\rd \ln\left(1-\left|r_1(s) \right|^2  \right)  \right)\right| \leq \mathcal{O}\left(t^{-3/4}\right),
	  	\end{align*}
	  	where $\Gamma$ denotes the Gamma function, $k_0=\frac{\zeta}{2}$, $\nu=-\frac{1}{2\pi}\ln (1-\left| r_1(k_0)\right|^2 )$.
	  \end{theo}}
	  The proof of this theorem is given in Sections \ref{sec_dispersive}-\ref{sec_u}. This provides the asymptotic stability of the solution to the good Boussinesq equation \eqref{gB} in dispersive shock wave region. 
	  
	   {\begin{remark}
	  	Obviously, the good Boussinesq equation \eqref{gB} contains a fourth-order partial derivative term with respect to $x$. Substituting the result on the asymptotic stability of the solution to the initial value problem for the good Boussinesq equation in Theorem \ref{theo_asymptotics} into equation \eqref{gB} involves the fourth-order derivatives of the reflection coefficients $r_1$ and $r_2$ with respect to their variables, which also justifies the reasonableness of our selection of the space to which the reflection coefficients belong.
	  \end{remark}}
	
	  {\begin{remark}
	(Asymptotic form in the left half-plane). Theorem \ref{theo_asymptotics} only presents the asymptotic expression of the solution $ u(x, t) $ for the good Boussinesq equation \eqref{gB} in a subsector of the right half-plane $ x > 0 $. Using the same method, a similar formula can be derived for the subsector in the symmetric region $ x < 0 $, where the critical point $ k_0 = \frac{x}{2t} < 0 $ will lie on the jump contour $\Gamma_4$ (see Fig. \ref{figGamma}). \\
	 \hspace*{2em}By employing the same opening procedure and error estimates, an analogous formula for the left half-plane subsector can be obtained. It is important to note that since the critical point $ k_0 $ lies on $\Gamma_4$, this manifests in the scattering data as follows: the reflection coefficient $ r_1(k) $ is incorporated into the error term in the subsequent analysis, while the reflection coefficient $ r_2(k) $ will replace the role of $ r_1(k) $ from the right region, becoming the primary constituent function for the leading asymptotic term of the solution to the initial-value problem in the left region. Consequently, the asymptotic form for the left half-plane can be derived directly from Theorem \ref{theo_asymptotics} by replacing reflection coefficient $r_1(k)$ with $r_2(k)$. Utilizing the invariance of the Boussinesq equation under space inversion, the solution is shown to satisfy the following estimate uniformly for $\zeta$ in compact subsets of $(-\infty,0)$ as $t\to \infty$:
	 \begin{align*}
	 	&\left| u(x,t)+\frac{3^{5/4}k_0\sqrt{\nu_1}}{\sqrt{2t}}\sin\left(\frac{19\pi}{12}+\nu_1\ln \left(6\sqrt{3}tk_0^2 \right) -\sqrt{3}tk_0^2-\arg r_2(k_0)  \right. \right.\\
	 	&\left.-\arg\Gamma\left(i\nu_1 \right)+ \left.\frac{1}{\pi} \int_{-\infty}^{k_0 }\ln\left(\frac{\left|s-k_0 \right| }{\left| s-\omega k_0\right| } \right)\rd \ln\left(1-\left|r_2(s) \right|^2  \right)  \right)\right| \leq \mathcal{O}\left(t^{-3/4}\right),
	 \end{align*}
	 where $\nu_1=-\frac{1}{2\pi}\ln (1-\left| r_2(k_0)\right|^2 )$.
	 \end{remark}}

	  \subsection{Notation}
	  \ \ \ \
	  We summarize some notation that will be used throughout the paper:
	  \begin{itemize}
	  	\item If $ A $ is an $ n \times m $ matrix, then $ |A| \geq 0 $ is defined by $ |A|^2 = \sum_{i,j} |A_{ij}|^2 $. Note that $ |A + B| \leq |A| + |B| $ and $ |AB| \leq |A||B| $.
	  	\item $C$ and $c$ will denote generic positive constants which may change within a computation.
	  	\item $\mathcal{I}$ is defined as any compact subset of $(0,\infty)$.
	  	\item $r^*(k)$ denotes the Schwartz reflection of $r(k)$.
	  	\item  {The notation $O$ appearing in the formulas of this paper is denoted as a $3\times 3$ zero matrix.}
	  \end{itemize}

		\section{Transformations of the RH problem}\label{sec_dispersive}
		\ \ \ \
		This paper employs the $\bar{\partial}$-steep descent method to analyze the long-time asymptotic properties of the good Boussinesq equation \eqref{gB}. During the process, the jump relations in the RH problem \ref{rhp_M} are decomposed, with some information being attenuated and useful information being transformed step by step into solvable problems. The long-time asymptotic solution is obtained, along with a more accurate error estimate. Furthermore, the asymptotic stability of the solution is proven.
		
		\subsection{Perform the first transformation}
		\ \ \ \
		First, consider the notation
		\begin{equation*}
			\theta_{ij}(\zeta,k):=t\Phi_{ij}(\zeta,k),
		\end{equation*}
	where $\zeta=\frac{x}{t}$ and
		\begin{align*}
			\Phi_{21}(\zeta, k) &= \omega(\omega - 1)(\zeta - k)k, \\
			\Phi_{31}(\zeta, k) &= (1 - \omega)(\zeta - \omega^2 k)k, \\
			\Phi_{32}(\zeta, k) &= (1 - \omega^2)(\zeta - \omega k)k,
		\end{align*}
		and $\partial\Phi_{ij}(\zeta, k)/\partial k=0$, $1\leq j<i\leq 3$ has the single zeros given by $
			 k_0=\zeta/2,\, \omega k_0,\, \omega^2k_0$.
		
		Since there are no singularities on the contours $\Gamma_2$, $\Gamma_4$, and $\Gamma_6$ in the jump relations satisfied by the RH problem \ref{rhp_M}, we first consider decomposing the corresponding three jump matrices as follows
		\begin{align*}
				&v_2(x,t,k)=\begin{pmatrix}
			 	1 & 0 & 0 \\
			 	0 & 1 & -r_2^*(\omega k){\rm{e}}^{-t\Phi_{32}} \\
			 	0 & 0 & 1
			 \end{pmatrix}
			 \begin{pmatrix}
			 	1 & 0 & 0 \\
			 	0 & 1  & 0 \\
			 	0 & r_2(\omega k){\rm{e}}^{t\Phi_{32}} & 1
			 \end{pmatrix},\\
				&v_4(x,t,k) =\begin{pmatrix}
				1 & -r_2^*(k){\rm{e}}^{-t\Phi_{21}} & 0 \\
				0 & 1 & 0 \\
				0 & 0 & 1
			\end{pmatrix}
			\begin{pmatrix}
				1  & 0 & 0 \\
				r_2(k){\rm{e}}^{t\Phi_{21}} & 1 & 0 \\
				0 & 0 & 1
			\end{pmatrix},\\
			&v_6(x,t,k) =\begin{pmatrix}
				1 & 0 & 0 \\
				0 & 1 & 0 \\
				-r_2^*(\omega^2 k){\rm{e}}^{t\Phi_{31}} & 0 & 1
			\end{pmatrix}
			\begin{pmatrix}
				1 & 0 & r_2(\omega^2 k){\rm{e}}^{-t\Phi_{31}} \\
				0 & 1 & 0 \\
				0 & 0 & 1
			\end{pmatrix}.
		\end{align*}
		
	The continuous extension of the reflection coefficient $r_2(k)$ is given by the following proposition.
	\begin{prop}\label{prop_G_exist}
		 {Assume $r_2\in H^{3,4}(-\infty,0)$, there exist continuous functions $G_j(k)$ for $k\in \bar D_j$ $j=3,4$ that satisfy the following boundary conditions:}
		\begin{align*}\label{Gj}
			 G_3(k)&=\left\lbrace\begin{aligned}
			 	&r_2^*( k), \quad &&k
			 	\in \Gamma_4,\\
			 	&g_3=r_2^*(0), \quad && k\in \Gamma_3,
			 \end{aligned}
			 \right.\\
			 G_4(k)&=\left\lbrace\begin{aligned}
			 	&r_2( k), \quad &&k
			 	\in \Gamma_4,\\
			 	&g_4=r_2(0), \quad && k\in \Gamma_5.
			 \end{aligned}
			 \right.
		\end{align*}
		And $G_j$ for $j=3,4$ satisfies the following estimates:
		\begin{equation}\label{dbarGj}
			\begin{aligned}
				&\left|\bar \partial G_j(k) \right| \leq c \left|k \right| ^{-1/2}+c \left|r'_2({\rm{Re}}\,k) \right|,\\
				&\left| G_j(k) \right| \leq c\sin^2(\arg k)+c [1+({\rm{Re}}\,k)^2]^{-1/4},
			\end{aligned}
		\end{equation}
		for positive real number  $c$.
	\end{prop}
	\begin{proof}
	Only prove the case for $G_4$. Here, define $k=\rho \text{e}^{\ri \psi}$ and $\bar k=\rho \text{e}^{-\ri \psi}$, then the following equation can be obtained:
		\begin{equation*}
			\bar \partial=\frac{\partial \rho}{\partial \bar k}\frac{\partial}{\partial \rho}+\frac{\partial \psi}{\partial \bar k}\frac{\partial}{\partial \psi}=\frac{1}{2}\text{e}^{\ri \psi}(\partial_{\rho}+\ri \rho^{-1}\partial_{\psi}).
 		\end{equation*}
		Construct the function
		\begin{equation*}
			G_4(k)=\cos(2\psi)r_2({\rm{Re}}\,k)+[1-\cos(2\psi)]g_4(k),
		\end{equation*}
		where $\rho\leq0$, $\pi\leq \psi\leq \frac{4\pi}{3}$, $\rho=\left|k \right| $ and ${\rm{Re}}\, k=\rho\cos \psi$.
		The function $G_4$ constructed above satisfies $G_4(k)=r_2(k)$ for $k \in\Gamma_4$, and $G_4(k)=g_4(k)$ for $k\in \Gamma_5$. And since $g_4$ is an analytic function, it follows that
		\begin{align*}
			\bar{\partial}G_4&=\bar{\partial}\left\lbrace\cos(2\psi)r_2({\rm{Re}}\,k)+[1-\cos(2\psi)]g_4(k) \right\rbrace \\
			&=\frac{1}{2}\text{e}^{\ri \psi}(\partial_{\rho}+\ri \rho^{-1}\partial_{\psi})\left\lbrace\cos(2\psi)r_2({\rm{Re}}\,k)+[1-\cos(2\psi)]g_4(k) \right\rbrace \\
			&=-\ri(r_2({\rm{Re}}\,k)-g_4(k))\text{e}^{\ri\psi}\rho^{-1}\sin(2\psi)+\frac{1}{2}\cos(2\psi)r'_2({\rm{Re}}\,k).
		\end{align*}
		Then
		\begin{equation*}
			\left| \bar \partial G_4\right|\leq \frac{c}{\left| k\right| }\left| r_2({\rm{Re}}\,k)-r_2(0)\right|+c\left|r'_2({\rm{Re}}\,k) \right|,
		\end{equation*}
		and
		\begin{align*}
			\left| r_2({\rm{Re}}\,k)-r_2(0)\right|&=\left| \int_{{\rm{Re}}\,k}^{0}r'_2(s)\text{d} s\right| \leq \int_{{\rm{Re}}\,k}^{0} \left| r'_2(s)\right| \text{d} s\\
			&\leq\left\| r'_2\right\|_{L^2({\rm{Re}}\,k,0)}\left\| 1\right\|_{L^2({\rm{Re}}\,k,0)} \leq c \left|k \right| ^{1/2}.
		\end{align*}
		Therefore, based on the above results, we obtain
		\begin{equation*}
			\left| \bar \partial G_4\right|\leq c \left|k \right| ^{-1/2}+c \left|r'_2({\rm{Re}}\,k) \right|.
		\end{equation*}		
		 {Furthermore, since $r_2\in H^{3,4}(-\infty,0)$, we can deduce  $\left| r_2({\rm{Re}}\,k)\right| \leq c[1+({\rm{Re}}\,k)^2]^{-1/4} $. Therefore}
		\begin{align*}
			\left|G_4\right|&\leq 2\sin^2\psi\left| g_1(k)\right| +\left| \cos(2\psi)\right| \left| r_2({\rm{Re}}\,k)\right| \\
			&\leq c\sin^2\psi+c [1+({\rm{Re}}\,k)^2]^{-1/4}.
		\end{align*}
	\end{proof}
	
	The functions $G_j$, $j=1,2,5,6$ can be obtained using symmetry and they satisfy the following conditions
		\begin{align*}
		G_j(k)=G_{j+2}(\omega  k),\quad & \left| G_j(k) \right| \leq c\sin^2(\arg  k)+c [1+({\rm{Re}}\, k)^2]^{-1/4},\\
		& \left|\bar \partial G_j(k) \right| \leq c \left|k \right| ^{-1/2}+c \left|r'_2({\rm{Re}}\,\omega k) \right|,\quad j=1,2,
			\end{align*}
			and
	\begin{align*}
		G_j(k)=G_{j-2}(\omega^2 k),\quad & \left| G_j(k) \right| \leq c\sin^2(\arg  k)+c [1+({\rm{Re}}\,k)^2]^{-1/4},\\
		&\left|\bar \partial G_j(k) \right| \leq c \left|k \right| ^{-1/2}+c \left|r'_2({\rm{Re}}\,\omega^2 k) \right|,\quad j=5,6.
	\end{align*}

	After providing the definitions of the functions $G_j$, $j=1,2,\cdots,6$, we define the following matrices in preparation for the first transformation:
	\begin{align*}
		&v_{2,U}=\begin{pmatrix}
			1 & 0 & 0 \\
			0 & 1 & -G_1(k){\rm{e}}^{-t\Phi_{32}} \\
			0 & 0 & 1
		\end{pmatrix},\quad
		&&v_{2,L}=\begin{pmatrix}
			1 & 0 & 0 \\
			0 & 1  & 0 \\
			0 & G_2(k){\rm{e}}^{t\Phi_{32}} & 1
		\end{pmatrix},\\
		&v_{4,U}=\begin{pmatrix}
			1 & -G_3(k){\rm{e}}^{-t\Phi_{21}} & 0 \\
			0 & 1 & 0 \\
			0 & 0 & 1
		\end{pmatrix},\quad
		&&v_{4,L}=\begin{pmatrix}
			1  & 0 & 0 \\
			G_4(k){\rm{e}}^{t\Phi_{21}} & 1 & 0 \\
			0 & 0 & 1
		\end{pmatrix},\\
		&v_{6,L}=\begin{pmatrix}
			1 & 0 & 0 \\
			0 & 1 & 0 \\
			-G_5(k){\rm{e}}^{t\Phi_{31}} & 0 & 1
		\end{pmatrix},
		&&v_{6,U}=\begin{pmatrix}
			1 & 0 & G_6(k){\rm{e}}^{-t\Phi_{31}} \\
			0 & 1 & 0 \\
			0 & 0 & 1
		\end{pmatrix}.
	\end{align*}
	Then define the transformation
	\begin{equation}\label{trans1}
		M^{(1)}(x,t,k)=M(x,t,k)G^{(1)}(x,t,k),
	\end{equation}
	 where $G^{(1)}$ is defined by
	\begin{equation*}
		G^{(1)}(x,t,k)=\left\lbrace \begin{aligned}
			&v_{2,U},\quad &&k\in D_1,\\
			&v_{2,L}^{-1},\quad &&k\in D_2,\\
			&v_{4,U},\quad &&k\in D_3,\\
			&v_{4,L}^{-1},\quad &&k\in D_4,\\
			&v_{6,L},\quad &&k\in D_5,\\
			&v_{6,U}^{-1},\quad &&k\in D_6.
			\end{aligned}
		 \right.
	\end{equation*}
		
	\begin{lem}\label{lem_G_bounded}
		As $\zeta\in\mathcal{I}$ and $t>1$, $G^{(1)}(x,t,k)$ and $(G^{(1)}(x,t,k))^{-1}$ are uniformly bounded for $k\in\mathbb{C}\setminus \Gamma$. More importantly, $\partial_x^l\lim\limits_{k\to\infty}k \left[ \left( G^{(1)}\right) ^{\pm1}-I\right]=O$ for $l=0,1$.
	\end{lem}
	 {\begin{proof}
		For $k = u + \ri v\in D_1$, where $0 < u < \infty$ and $0 < v < \sqrt{3}u$. The calculation yields
		\begin{equation*}
			t{\rm{Re}}\Phi_{32}(\zeta,u,v)=-\frac{\sqrt{3}t}{2}(u+\sqrt{3}v+2k_0)(v-\sqrt{3}u)>0, 
		\end{equation*}
	 and utilizing Proposition \ref{prop_G_exist}, we have
		\begin{align*}
			\left| G^{(1)}_1(k)-I\right|&= \left| v_{2,U}(k)-I\right|=\left| G_1(k){\rm{e}}^{-t\Phi_{32}}\right|\\
			&\leq c\left\lbrace  \sin^2(\arg k)+ [1+({\rm{Re}}\,k)^2]^{-1/4} \right\rbrace {\rm{e}}^{-t\,{\rm{Re}}\,\Phi_{32}}.
		\end{align*}
		Thus, 
	\begin{align*}
		\lim_{k\to\infty}\left| G^{(1)}_1(k)-I\right|\leq \lim_{u,v\to\infty}c\left[   1+ (1+u^2)^{-1/4} \right]  {\rm{e}}^{-t\,{\rm{Re}}\,\Phi_{32}(\zeta,u,v)}=0.
	\end{align*}
	Hence, it follows that  $\lim\limits_{k\to\infty}G^{(1)}_1(k)=I$, $k\in D_1$. Similarly, $\lim\limits_{k\to\infty}\left( G^{(1)}_1(k)\right) ^{-1}=I$. 
	According to the form of $G^{(1)}_1(k)-I$, only the $(2,3)$-position is non-zero:
		\begin{align*}
		   \lim_{k\to\infty}k\left(  G_1(k){\rm{e}}^{-t\Phi_{32}}\right)
			&\leq c\lim_{k\to\infty} k\left\lbrace  \sin^2(\arg \omega k)+ [1+({\rm{Re}}\,\omega k)^2]^{-1/4} \right\rbrace {\rm{e}}^{-t\Phi_{32}}   \\
			&\leq c\lim_{u,v\to\infty} \left(u+\ri v+\frac{u+\ri v}{\sqrt{u^2+v^2} } \right) {\rm{e}}^{-t\,{\rm{Re}}\,\Phi_{32}(\zeta,u,v)}=0.
		\end{align*}
		The proof for matrix $\left( G^{(1)}(k)\right) ^{-1}$ follows the same procedure as above. This is because the only difference between $G^{(1)}_1(k)$ and $\left( G^{(1)}(k)\right) ^{-1}$ is that the (2,3) entry is the opposite in sign, yet both decay rapidly as $k\to\infty$. The case for $l = 1$ will not be elaborated further and the proofs for $k$ belonging to the other five regions can be similarly obtained.
		In summary, $\left( G^{(1)}\right) ^{\pm1}=I+o(k^{-1})$.
	\end{proof}}
	 {\begin{lem}\label{lemDbarG1}
		$G^{(1)}(x,t,k)$ is a continuous but non-analytic function within its domain, and it satisfies the relation
		\begin{equation}\label{G_inverse_DbarG}
			\left( G^{(1)}(k)\right) ^{-1}\bar{\partial}G^{(1)}(k)=\bar{\partial}G^{(1)}(k),
		\end{equation}
		where the forms of $\bar{\partial}G^{(1)}_j(k)$ for $k\in D_j$ are as follows
		\begin{equation}\label{DbarG1}
			\begin{aligned}
				&\bar{\partial}G^{(1)}_1(k)=\begin{pmatrix}
					0 & 0 & 0 \\
					0 & 0 & -\bar{\partial}G_1(k){\rm{e}}^{-t\Phi_{32}} \\
					0 & 0 & 0
				\end{pmatrix},\quad
				&&\bar{\partial}G^{(1)}_2(k)=\begin{pmatrix}
					0 & 0 & 0 \\
					0 & 0  & 0 \\
					0 & -\bar{\partial}G_2(k){\rm{e}}^{t\Phi_{32}} & 0
				\end{pmatrix},\\
				&\bar{\partial}G^{(1)}_3(k)=\begin{pmatrix}
					0 & -\bar{\partial}G_3(k){\rm{e}}^{-t\Phi_{21}} & 0 \\
					0 & 0 & 0 \\
					0 & 0 & 0
				\end{pmatrix},\quad
				&&\bar{\partial}G^{(1)}_4(k)=\begin{pmatrix}
					0  & 0 & 0 \\
					-\bar{\partial}G_4(k){\rm{e}}^{t\Phi_{21}} & 0 & 0 \\
					0 & 0 & 0
				\end{pmatrix},\\
				&\bar{\partial}G^{(1)}_5(k)=\begin{pmatrix}
					0 & 0 & 0 \\
					0 & 0 & 0 \\
					-\bar{\partial}G_5(k){\rm{e}}^{t\Phi_{31}} & 0 & 0
				\end{pmatrix},\quad
				&&\bar{\partial}G^{(1)}_6(k)=\begin{pmatrix}
					0 & 0 & -\bar{\partial}G_6(k){\rm{e}}^{-t\Phi_{31}} \\
					0 & 0 & 0 \\
					0 & 0 & 0
				\end{pmatrix}.
			\end{aligned}
		\end{equation}
	\end{lem}}
 {\begin{remark}\label{remark_3.1}
		From the proof of Proposition \ref{prop_G_exist}, it can be seen that the constructions of functions  $G_j(k)$, $j=1,2,\cdots,6,$ satisfy the definition of continuous functions, but they are not analytic functions. Consequently, the matrix $G_j^{(1)}(k)$ defined by $G_j(k)$ must be continuous but non-analytic. The lemma above also provides the specific form of $\bar\partial G^{(1)}_j$ for $k\in D_j$.
		Equation \eqref{G_inverse_DbarG} can be easily proved by noting that only one element among the nine elements of matrix $\bar{\partial}G^{(1)}(k)$ in \eqref{dbarGj} is nonzero. We can directly compute $\left( G^{(1)}(k)\right) ^{-1}$ and verify the result by performing matrix multiplication.
\end{remark}}
	
	 {Under the above transformation \eqref{trans1}, the new vector eigenfunction $N^{(1)}$$(x,t,k)$$=$$(\omega\,\,\omega^2\,\,1)$$M^{(1)}$ $(x,t,k)$ obtained satisfies the jump relations $N^{(1)}_+(x,t,k)$$=$$N^{(1)}_-$$(x,t,k)v_j^{(1)}(x,t,k)$ for $k\in\Gamma_j^{(1)}$, $j=1,2,3$ (see Fig. \ref{figGamma1}).} The specific forms of $v_j^{(1)}$, $j=1,2,3,$ are given as follows:
	\begin{align}\label{v1}
		&v_1^{(1)}(x,t,k)=\begin{pmatrix}
			1 & -r_1(k){\rm{e}}^{-t \Phi_{21}} & \left[ r_1(k)g_1(k)+g_6(k)\right]{\rm{e}}^{-t \Phi_{31}}\\
			r_1^*(k){\rm{e}}^{t \Phi_{21}} & 1-\left|r_1(k) \right|^2 & -g_1(k)(1-\left| r_1(k)\right|^2) {\rm{e}}^{-t \Phi_{32}}\\
			0 & 0 & 1
		\end{pmatrix},\nonumber\\
		&v_2^{(1)}(x,t,k)=\begin{pmatrix}
			1-\left|r_1(\omega^2 k) \right|^2 & -g_3(k)( 1-\left|r_1(\omega^2 k) \right|^2) {\rm{e}}^{-t \Phi_{21}} & r_1^*(\omega^2 k){\rm{e}}^{-t \Phi_{31}}\\
			0 & 1 & 0\\
			-r_1(\omega^2 k){\rm{e}}^{t \Phi_{31}} & \left[r_1(\omega^2 k)g_3(k)+g_2(k)\right]{\rm{e}}^{t \Phi_{32}} & 1
		\end{pmatrix},\\
		&v_3^{(1)}(x,t,k)=\begin{pmatrix}
			1 & 0 & 0\\
			[r_1(\omega k)g_5(k)+g_4(k)]{\rm{e}}^{t \Phi_{21}} & 1 & -r_1(\omega k){\rm{e}}^{-t \Phi_{32}}\\
			-g_5(k) (1-\left|r_1(\omega k) \right|^2) {\rm{e}}^{t \Phi_{31}} & r_1^*(\omega k) {\rm{e}}^{t \Phi_{32}} & 1-\left|r_1(\omega k) \right|^2
		\end{pmatrix},\nonumber
	\end{align}
	and they satisfy the following symmetry relationships:
	\begin{align*}
		v^{(1)}_2(k)=\mathcal{A}^{-1}v_1^{(1)}(\omega^2k)\mathcal{A},\quad
		v^{(1)}_3(k)=\mathcal{A}^{-2}v_1^{(1)}(\omega k)\mathcal{A}^2.
	\end{align*}
	
		\begin{figure}[htbp]
		\centering
		\begin{tikzpicture}[scale=1.2]
			
			\draw[very thick, black!20!blue, -latex] (0,0) -- (2,0);
			\draw[very thick, black!20!blue, latex-] (-1,2*1.732/2) -- (0,0);
			\draw[very thick, black!20!blue, latex-] (-1,-2*1.732/2) -- (0,0);
			
			\draw [very thick,black!20!blue](0,0) -- (3,0);
			\draw [very thick,black!20!blue](-1.5,1.5*1.732) -- (0,0);
			\draw [very thick,black!20!blue](0,0) -- (-1.5,-1.5*1.732);
			
			\draw [dashed,very thick,white!50!blue](0,0) -- (-3,0);
			\draw [dashed,very thick,white!50!blue](1.5,-1.5*1.732) -- (0,0);
			\draw [dashed,very thick,white!50!blue](0,0) -- (1.5,1.5*1.732);

			\node[red!70!black,below] at (1.85,-0.1) {$1$};
			\node[red!70!black,right] at (-0.9,1*1.65) {$2$};
			\node[red!70!black,right] at (-0.9,-1.6) {$3$};

			\node[below] at (1,0) {$k_0$};
			\node[left] at (-0.5,1.732/2) {$\omega k_0$};
			\node[left] at (-0.5,-1.732/2) {$\omega^2 k_0$};
			
			\fill (0,0) circle (1pt);
			\fill (1,0) circle (1.5pt);
			\fill (-1/2,1.732/2) circle (1.5pt);
			\fill (-1/2,-1.732/2) circle (1.5pt);
			
			\node[right] at (3,0) {${\rm{Re}\,k}$};
			\draw (0:0.4cm) arc (0:60:0.4cm);
			\node[right] at (0.25,0.2) {{\small  $\pi/3$}};
		\end{tikzpicture}
		\caption{The jump contour $\Gamma^{(1)}$ in the complex $k$-plane.}
		\label{figGamma1}
	\end{figure}
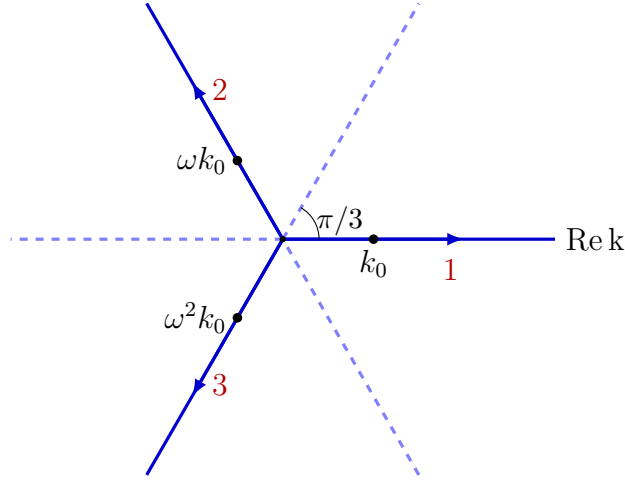
	
	 {The eigenfunction $N^{(1)}(x,t,k)$ satisfies the following mixed problem.}
	 {\begin{rhp-Dbar}\label{rhp-Dbar1}
		Find a $1 \times 3$-row-vector valued function $N^{(1)}(x, t, k)$ with the following properties:
		\begin{enumerate}
			\item $N^{(1)}(x, t, \cdot)$ is continuous for $k\in \mathbb{C}\setminus\Gamma^{(1)}$ with continuous boundary values $N^{(1)}_{\pm}$.
			\item The jump relation
			\begin{equation*}
				N^{(1)}_+(x, t, k) = N^{(1)}_-(x, t, k) v^{(1)}(x, t, k), \quad k \in \Gamma^{(1)},
			\end{equation*}
			where $v^{(1)}(x,t,k)=v^{(1)}_j(x,t,k)$ for $k\in \Gamma^{(1)}_j$, $j=1,2,3$ are given by equation \eqref{v1}, and the contour of  $\Gamma^{(1)}$ is given by Fig. \ref{figGamma1}.
			\item $N^{(1)}(x, t, k)\rightarrow (\omega \,\, \omega^2 \,\, 1)$ as $k\to \infty$ in $\mathbb{C}\setminus\Gamma^{(1)}$.
			\item As $k\to 0$, $
			N^{(1)}(x, t, k) =\mathcal{O}(1)$.
			\item For $k\in\mathbb{C}\setminus\Gamma^{(1)}$, there exists the $\bar{\partial}$ problem $\bar{\partial}N^{(1)}( k)=N^{(1)}(k)\bar{\partial}G^{(1)}(k) $.
		\end{enumerate}
	\end{rhp-Dbar}}
	
	We decompose $M^{(1)}(x,t,k)$ in equation \eqref{trans1} as follows
	\begin{equation}\label{M1factorisation}
		M^{(1)}(x,t,k)=E^G(x,t,k)M^{RH}(x,t,k),
	\end{equation}
	 {where $E^{G}(x,t,k)$ is the solution to the following pure $\bar{\partial}$ problem \ref{DbarEG}, and $M^{RH}(x,t,k)$ is the solution to a pure matrix RH problem. The vector eigenfunction $N^{RH}(x,t,k)=(\omega\,\,\omega^2\,\,1)M^{RH}(x,t,k)$, which has no singularity as $k \to 0$ as given by $M^{RH}$, is the solution to the following pure RH problem \ref{rhpMRH}.}
	 {\begin{rhp}\label{rhpMRH}
		Find a $1 \times 3$-row-vector valued function $N^{RH}(x, t, k)$ with the following properties:
		\begin{enumerate}
			\item  $N^{RH}(x, t, \cdot) : \mathbb{C} \setminus \Gamma^{(1)} \rightarrow \mathbb{C}^{1 \times 3}$ is analytic, that is, $\bar{\partial}N^{RH}(k)=O_{1\times 3}$.
			\item The jump relation
			\begin{equation*}
				N^{RH}_+(x, t, k) = N^{RH}_-(x, t, k) v^{(1)}(x, t, k), \quad k \in \Gamma^{(1)},
			\end{equation*}
			where $v^{(1)}(x,t,k)=v^{(1)}_j(x,t,k)$ for $k\in \Gamma^{(1)}_j$, $j=1,2,3$ are given by equation \eqref{v1}.
			\item $N^{RH}(x, t, k)\to (\omega \,\, \omega^2 \,\, 1)$ as $k\to \infty$ in $\mathbb{C}\setminus\Gamma^{(1)}$.
			\item As $k\to 0$, $
			N^{RH}(x, t, k) =\mathcal{O}(1)$.
		\end{enumerate}
	\end{rhp}}
	
	\begin{Dbar}\label{DbarEG}
		Find a $3 \times 3$-matrix valued function $E^{G}(x, t, k)$ with the following properties:
		\begin{enumerate}
			\item $E^{G}$ is continuous in the complex plane $\mathbb{C}$.
			\item $E^{G}(x, t, k)\to I$ as $k\to \infty$ in $\mathbb{C}$.
			\item $\bar{\partial}E^{G}(k)=E^{G}(k)Y(k)$,
			where
			\begin{equation*}
				Y(k)=M^{RH}(k)\bar{\partial}G^{(1)}_j(k)\left(M^{RH}(k) \right)^{-1},\quad k\in D_j,\quad j=1,2,\cdots,6,
			\end{equation*}
			and $\bar{\partial}G^{(1)}_j(k)$, $j=1,2,\cdots,6$, are given by equation \eqref{DbarG1}.
			\end{enumerate}
	\end{Dbar}
	\begin{remark}
		We note that $M^{RH}(x,t,k)$ is analytic and $M^{(1)}(x,t,k)$ is continuous in $D_j$, $j=1,2,\cdots,6$, hence
	\begin{align*}
			\bar{\partial} E^{G}(k)=&\bar{\partial}\left[M^{(1)}(k)\left(M^{RH}(k) \right)^{-1}   \right]\\
			=&\bar{\partial}M^{(1)}(k)\left(M^{RH}(k) \right)^{-1}\\
			=&M^{(1)}(k)\bar{\partial}G^{(1)}(k)\left(M^{RH}(k) \right)^{-1}\\
			=&E^G(k)\left[ M^{RH}(k)\bar{\partial}G^{(1)}(k)\left(M^{RH}(k) \right)^{-1}\right].
		\end{align*}
		In other words, $Y(k)=M^{RH}(k)\bar{\partial}G^{(1)}(k)\left(M^{RH}(k) \right)^{-1}$.
	\end{remark}

	\subsection{Perform the second transformation }
	\ \ \ \
	The scalar RH problem is first introduced as follows:
	\begin{align*}
		&\delta_{1+}(\zeta,k)=\left\lbrace\begin{aligned}
			&\delta_{1-}(\zeta,k)\left( 1-\left|r_1(k) \right|^2\right) ,\quad &&k\in \left[ k_0,\infty\right), \\
			&\delta_{1-}(\zeta,k), && k\in \left( -\infty,k_0\right),
		\end{aligned}
		\right.\\
		&\delta_1(\zeta,k)\to 1,\quad k\to\infty.
	\end{align*}
	Based on the relationship $k_0={\zeta}/{2}$, we obtain that there exists $\epsilon>0$ such that the following inequality holds for all $k\in \left[ k_0,\infty\right)$ and $\zeta\in \mathcal{I}$
	\begin{equation*}
		\left| r_1(k_0)\right| \leq 1-\epsilon.
	\end{equation*}
	
	By Plemeli formulas, we obtain
	\begin{equation}\label{delta1}
		\delta_1(\zeta,k)={\rm{exp}}\left\lbrace \frac{1}{2 \pi \ri} \int_{\left[ k_0,\infty\right)} \frac{\ln \left(1-\left|r_1(s)\right|^2\right)}{s-k} {\rm{d}}s \right\rbrace,\quad k\in \mathbb{C}\setminus \left[ k_0,\infty\right).
	\end{equation}

	\begin{prop}\label{prop_delta}
		 {Assuming $\zeta\in \mathcal{I}$, $r_1\in H^{3,4}(0,\infty)$, and based on the properties of the reflection coefficient $\left\| r_1\right\|_{L^\infty(0,\infty)}\leq \sigma<1 $, then the following properties hold:}
		\begin{enumerate}
			\item $\delta_1(k)$ is analytic on $\mathbb{C}\setminus\left[ k_0,\infty\right) $ and can be written as
			\begin{equation*}
				\delta_1(\zeta,k)={\rm{e}}^{-\ri\nu \ln_0(k-k_0)}{\rm{e}}^{-\chi(\zeta,k)},
			\end{equation*}
			where
			\begin{align*}
				&\nu=-\frac{1}{2\pi}\ln(1-\left|r_1(k_0) \right|^2 ),\\
				&\chi(\zeta,k)=\frac{1}{2\pi \ri}\int_{k_0}^{\infty}\ln_0(k-s){\rm{d}}\ln(1-\left|r_1(s)\right|^2 ).
			\end{align*}
			\item $\delta_1(k)\overline{\delta_1(\bar k)}=1$, $\left\|\delta_{1\pm} -1\right\|_{L^2}\leq \frac{c\left\|r_1 \right\|_{L^2} }{1-\sigma} $.
			\item $(1-\sigma^2)^{1/2}\leq \left| \delta_1(k)\right| \leq (1-\sigma^2)^{-1/2} $.
			\item  As $k\rightarrow k_0$ along a path which is nontangential to  $(k_0, \infty)$,  we have
			\begin{align*}
				&|\chi_1(\zeta, k) - \chi_1(\zeta, k_0)| \leq C\left| k - k_0\right|^{1/2} , \\
				&|\partial_x(\chi_1(\zeta, k) - \chi_1(\zeta, k_0))| \leq \frac{C}{t}(1 + |\ln|k - k_0||),
			\end{align*}
				where $C$ is independent of $\zeta \in \mathcal{I}$.  Furthermore,
				\begin{equation*}
					|\partial_x \chi_1(\zeta, k_0)| = \frac{1}{t}\left|\partial_u \chi_1(u, v)|_{(u,v)=(\zeta,k_0)} + \frac{1}{2}\partial_v \chi_1(u, v)|_{(u,v)=(\zeta,k_0)}\right| \leq \frac{C}{t}
				\end{equation*}
				and
				\begin{equation*}
					\partial_x(\delta_1(\zeta, k)^{\pm 1}) = \frac{\pm \ri\nu}{2t(k - k_0)}\delta_1(\zeta, k)^{\pm 1}.
				\end{equation*}
		\end{enumerate}
	\end{prop}
	\begin{proof}
		The proof can be completed by directly estimating based on equation \eqref{delta1}.
	\end{proof}

	By utilizing the symmetry, we can still provide the definitions of the functions $\delta_2$ and $\delta_3$ as follows
	\begin{align*}
		&\delta_2(\zeta,k)=\delta_1(\zeta,\omega^2k),\quad k\in \mathbb{C}\setminus \omega\left[ k_0,\infty\right),\\
		&\delta_3(\zeta,k)=\delta_1(\zeta,\omega k), \quad \,\, k\in \mathbb{C}\setminus \omega^2\left[ k_0,\infty\right).
	\end{align*}
	More importantly,
	\begin{align*}
		&\delta_2(\zeta,k)={\rm{exp}}\left\lbrace \frac{1}{2 \pi \ri} \int_{\omega\left[ k_0,\infty\right)} \frac{\ln \left(1-\left|r_1(s)\right|^2\right)}{s-k} {\rm{d}}s \right\rbrace,\quad &&k\in \mathbb{C}\setminus \omega\left[ k_0,\infty\right),\\
		&\delta_3(\zeta,k)={\rm{exp}}\left\lbrace \frac{1}{2 \pi \ri} \int_{\omega^2\left[ k_0,\infty\right)} \frac{\ln \left(1-\left|r_1(s)\right|^2\right)}{s-k} {\rm{d}}s \right\rbrace, &&k\in \mathbb{C}\setminus \omega^2\left[ k_0,\infty\right).
	\end{align*}

Construct the transformation $\Delta(\zeta,k)$ to obtain the new eigenfunction
\begin{equation*}\label{trans2}
	M^{(2)}(x,t,k)=M^{RH}(x,t,k)\Delta(\zeta,k),
\end{equation*}
 where
\begin{equation*}
	\Delta(\zeta,k)=
	\begin{pmatrix}
		\frac{\delta_1(\zeta,k)}{\delta_2(\zeta,k)} & 0 & 0\\
		0 & \frac{\delta_3(\zeta,k)}{\delta_1(\zeta,k)} & 0\\
		0 & 0 &\frac{\delta_2(\zeta,k)}{\delta_3(\zeta,k)}
	\end{pmatrix}.
\end{equation*}
	Based on Proposition \ref{prop_delta}, it can be concluded that $\Delta$ and $\Delta^{-1}$ are uniformly bounded for $\zeta\in\mathcal{I}$ and $k\in \mathbb{C}\setminus\Gamma^{(2)}$ (see Fig. \ref{figGamma2}), and satisfy $\Delta(\zeta,k)=I+\mathcal{O}(k^{-1})$ as $k\to \infty$.
	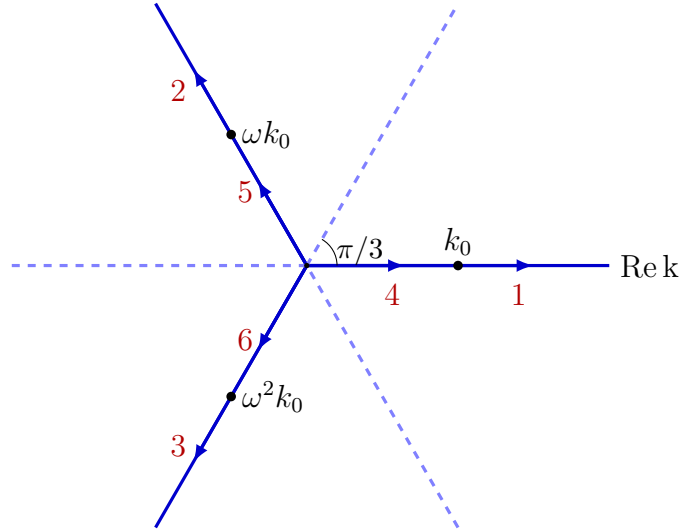
\begin{figure}[htbp]
		\centering
		\begin{tikzpicture}[scale=1]
			\draw[very thick, black!20!blue, -latex] (0,0) -- (1.3,0);
			\draw[very thick, black!20!blue, latex-] (-1.3/2,1.3*1.732/2) -- (0,0);
			\draw[very thick, black!20!blue, latex-] (-1.3/2,-1.3*1.732/2) -- (0,0);
			
			\draw[very thick, black!20!blue, -latex] (0,0) -- (3,0);
			\draw[very thick, black!20!blue, latex-] (-3/2,3*1.732/2) -- (0,0);
			\draw[very thick, black!20!blue, latex-] (-3/2,-3*1.732/2) -- (0,0);

			\draw [very thick,black!20!blue](0,0) -- (4,0);
			\draw [very thick,black!20!blue](-2,2*1.732) -- (0,0);
			\draw [very thick,black!20!blue](0,0) -- (-2,-2*1.732);
			
			\draw [dashed,very thick,white!50!blue](0,0) -- (-4,0);
			\draw [dashed,very thick,white!50!blue](2,-2*1.732) -- (0,0);
			\draw [dashed,very thick,white!50!blue](0,0) -- (2,2*1.732);

			\node[red!70!black,below] at (2.8,-0.1) {$1$};
			\node[red!70!black,left] at (-1.45,1.4*1.65) {$2$};
			\node[red!70!black,left] at (-1.45,-1.4*1.7) {$3$};
			\node[red!70!black,below] at (1.13,-0.1) {$4$};
			\node[red!70!black,left] at (-1.13/2,1.13*1.732/2) {$5$};
			\node[red!70!black,left] at (-1.13/2,-1.13*1.732/2) {$6$};
			
			\node[above] at (2,0) {$k_0$};
			\node[right] at (-1,1.732) {$\omega k_0$};
			\node[right] at (-1,-1.732) {$\omega^2 k_0$};
			
			\fill (2,0) circle (1.8pt);
			\fill (-1,1.732) circle (1.8pt);
			\fill (-1,-1.732) circle (1.8pt);
			\fill (0,0) circle (1pt);
			
			\node[right] at (4,0) {${\rm{Re}\,k}$};
			\draw (0:0.4cm) arc (0:60:0.4cm);
			\node[right] at (0.25,0.2) {{\small  $\pi/3$}};
		\end{tikzpicture}
		\caption{The jump contour $\Gamma^{(2)}$ in the complex $k$-plane.}
		\label{figGamma2}
	\end{figure}
	
	The jump relationship $M^{(2)}_+(\zeta,k)=M^{(2)}_-(\zeta,k)v^{(2)}(\zeta,k)$ for $k\in\Gamma^{(2)}$ as shown in the Fig. \ref{figGamma2} satisfied by $M^{(2)}$, with the corresponding jump matrix $v^{(2)}=\Delta_-^{-1}v^{(1)}\Delta_+$, are given in the following form
	\begin{align*}
		v_1^{(2)}
			&=\begin{pmatrix}
				1-\left|r_1(k) \right|^2 & -\frac{\delta_2\delta_3}{ \delta_{1-}^2} \frac{r_1(k)}{1-\left|r_1(k) \right|^2}{\rm{e}}^{-t \Phi_{21}} & \frac{\delta_2^2}{\delta_{1-}\delta_3}[r_1(k)g_1(k)+g_6(k)]{\rm{e}}^{-t \Phi_{31}}\\
				\frac{\delta_{1+}^2}{\delta_2\delta_3}\frac{r_1^*(k)}{1-\left|r_1(k) \right|^2}{\rm{e}}^{t \Phi_{21}} & 1 & -\frac{\delta_{1+}\delta_2}{\delta_3^2}g_1(k){\rm{e}}^{-t \Phi_{32}}\\
				0 & 0 & 1
			\end{pmatrix},\\
			v_2^{(2)}
			&=\begin{pmatrix}
				1 & -\frac{\delta_{2+}\delta_3}{\delta_1^2}g_3(k){\rm{e}}^{-t \Phi_{21}} & \frac{\delta_{2+}^2}{\delta_1\delta_3}\frac{r_1^*(\omega^2k)}{1-\left| r_1(\omega^2 k)\right|^2 }{\rm{e}}^{-t \Phi_{31}}\\
				0 & 1 & 0\\
				-\frac{\delta_1\delta_3}{\delta_{2-}^2}\frac{r_1(\omega^2k)}{1-\left| r_1(\omega^2 k)\right|^2 }{\rm{e}}^{t \Phi_{31}} & \frac{\delta_3^2}{\delta_1\delta_{2-}}[r_1(\omega^2k)g_3(k)+g_2(k)]{\rm{e}}^{t \Phi_{32}} & 1-\left| r_1(\omega^2 k)\right|^2
			\end{pmatrix},\\
			v_3^{(2)}
			&=\begin{pmatrix}
				1 & 0 & 0\\
				\frac{\delta_1^2}{\delta_2\delta_{3-}}[r_1(\omega k)g_5(k)+g_4(k)] {\rm{e}}^{t \Phi_{21}} & 1-\left|r_1(\omega k) \right|^2 & -\frac{\delta_1\delta_2}{\delta_{3-}^2} \frac{r_1(\omega k)}{1-\left|r_1(\omega k) \right|^2}{\rm{e}}^{-t \Phi_{32}}\\
				-\frac{\delta_1\delta_{3+}}{\delta_2^2}g_5(k){\rm{e}}^{t \Phi_{31}} & \frac{\delta_{3+}^2}{\delta_1\delta_2} \frac{r_1^*(\omega k)}{1-\left|r_1(\omega k) \right|^2}{\rm{e}}^{t \Phi_{32}} & 1
			\end{pmatrix},\\
			v_4^{(2)}&=\begin{pmatrix}
				1 & -\frac{\delta_2\delta_3}{\delta_1^2}r_1(k) {\rm{e}}^{-t \Phi_{21}} & \frac{\delta_2^2}{\delta_1\delta_3}[r_1(k)g_1(k) + g_6(k)]{\rm{e}}^{-t \Phi_{31}} \\
				\frac{\delta_1^2 }{\delta_2\delta_3}r_1^*(k) {\rm{e}}^{t \Phi_{21}} & 1-\left| r_1(k)\right|^2  & -\frac{\delta_1\delta_2}{\delta_3^2}g_1(k) (1-\left| r_1(k)\right|^2){\rm{e}}^{-t \Phi_{32}} \\
				0 & 0 & 1
			\end{pmatrix},\\
			v_5^{(2)}&=\begin{pmatrix}
				1 - \left|r_1(\omega^2 k) \right|^2  & -\frac{\delta_2\delta_3}{\delta_1^2}g_3(k)(1 - \left|r_1(\omega^2 k) \right|^2) {\rm{e}}^{-t\Phi_{21}}& \frac{\delta_2^2 }{\delta_1\delta_3}r_1^*(\omega^2 k) {\rm{e}}^{-t\Phi_{31}} \\
				0 & 1 & 0 \\
				-\frac{\delta_1\delta_3}{\delta_2^2}r_1(\omega^2 k){\rm{e}}^{t\Phi_{31}} & \frac{\delta_3^2 }{\delta_1 \delta_2}[r_1(\omega^2 k)g_3(k) + g_2(k)] {\rm{e}}^{t\Phi_{32}} & 1
			\end{pmatrix},\\
			v_6^{(2)}&=\begin{pmatrix}
				1 & 0 & 0 \\
				\frac{\delta_1^2 }{\delta_2\delta_3}[ r_1(\omega k)g_5(k)+g_4(k) ] {\rm{e}}^{t \Phi_{21}} & 1 & -\frac{\delta_1\delta_2}{\delta_3^2} r_1(\omega k){\rm{e}}^{-t \Phi_{32}} \\
				-\frac{\delta_1\delta_3}{\delta_2^2}(1-\left|r_1(\omega k) \right|^2 ) g_5(k){\rm{e}}^{t \Phi_{31}} & \frac{\delta_3^2}{\delta_1\delta_2}r_1^*(\omega k){\rm{e}}^{t \Phi_{32}} & 1-\left|r_1(\omega k) \right|^2
			\end{pmatrix}.
	\end{align*}

	\subsection{Perform the third transformation}	
	\ \ \ \
	To simplify the notation of the reflection coefficient in the jump matrix $v^{(2)}(x,t,k)$, we define
	\begin{equation*}
	   \rho_1(k)=\frac{r_1(k)}{1-\left|r_1(k) \right|^2 },\quad k\in \mathbb{R}_+.
	\end{equation*}
	And we provide the following triangular decomposition of the jump matrices:
	
	\begin{align*}
		v_1^{(2)}(\zeta,k)
		=&\begin{pmatrix}
			1 & -\frac{\delta_2\delta_3}{ \delta_{1-}^2} \rho_1(k){\rm{e}}^{-t \Phi_{21}} & \frac{\delta_2^2}{\delta_{1-}\delta_3}r_2(0){\rm{e}}^{-t \Phi_{31}}\\
			0 & 1 & 0\\
			0 & 0 & 1
		\end{pmatrix}
		\begin{pmatrix}
		1 & 0 & 0\\
		\frac{\delta_{1+}^2}{\delta_2\delta_3}\rho_1^*(k){\rm{e}}^{t \Phi_{21}} & 1 & -\frac{\delta_{1+}\delta_2}{\delta_3^2}r_2^*(0){\rm{e}}^{-t \Phi_{32}}\\
		0 & 0 & 1
		\end{pmatrix}\\
		:=&v_{1,below}^{(2)}(\zeta,k)v_{1,above}^{(2)}(\zeta,k),\\
				v_4^{(2)}(\zeta,k)
			=&\begin{pmatrix}
				1 & 0 & 0 \\
				\frac{\delta_1^2 }{\delta_2\delta_3}r_1^*(k) {\rm{e}}^{t \Phi_{21}} & 1 & -\frac{\delta_1\delta_2}{\delta_3^2}[r_2^*(0)+r_1^*(k)r_2(0)]{\rm{e}}^{-t \Phi_{32}} \\
				0 & 0 & 1
			\end{pmatrix}\\
			&\begin{pmatrix}
			1 & -\frac{\delta_2\delta_3}{\delta_1^2}r_1(k) {\rm{e}}^{-t \Phi_{21}} & \frac{\delta_2^2}{\delta_1\delta_3}[r_1(k)r_2^*(0) + r_2(0)]{\rm{e}}^{-t \Phi_{31}} \\
		0 & 1 & 0 \\
			0 & 0 & 1
			\end{pmatrix}\\
			=&v_{4,below}^{(2)}(\zeta,k)v_{4,above}^{(2)}(\zeta,k).
		\end{align*}
		The remaining decompositions are given based on symmetries:
		\begin{align*}
			&v_2^{(2)}(\zeta,k)=\mathcal{A}^{-1}v_1^{(2)}(\zeta,\omega^2 k)\mathcal{A}=[\mathcal{A}^{-1}v_{1,below}^{(2)}(\zeta,\omega^2 k)\mathcal{A}][\mathcal{A}^{-1}v_{1,above}^{(2)}(\zeta,\omega^2 k)\mathcal{A}],\\
		    &v_3^{(2)}(\zeta,k)=\mathcal{A}^{-2}v_1^{(2)}(\zeta,\omega k)\mathcal{A}^{2}=[\mathcal{A}^{-2}v_{1,below}^{(2)}(\zeta,\omega k)\mathcal{A}^{2}][\mathcal{A}^{-2}v_{1,above}^{(2)}(\zeta,\omega k)\mathcal{A}^{2}],\\
			&v_5^{(2)}(\zeta,k)=\mathcal{A}^{-1}v_4^{(2)}(\zeta,\omega^2 k)\mathcal{A}=[\mathcal{A}^{-1}v_{4,below}^{(2)}(\zeta,\omega^2 k)\mathcal{A}][\mathcal{A}^{-1}v_{4,above}^{(2)}(\zeta,\omega^2 k)\mathcal{A}],\\
		    &v_6^{(2)}(\zeta,k)=\mathcal{A}^{-2}v_4^{(2)}(\zeta,\omega k)\mathcal{A}^{2}=[\mathcal{A}^{-2}v_{4,below}^{(2)}(\zeta,\omega k)\mathcal{A}^{2}][\mathcal{A}^{-2}v_{4,above}^{(2)}(\zeta,\omega k)\mathcal{A}^{2}].
		\end{align*}
		
		The continuous extension properties of the reflection coefficient $r_1(k)$ are given below.
			\begin{prop}\label{R_exist}
			  {Assume $r_1\in H^{3,4}(0,\infty)$, there exist functions $R_j(k)$ for $k\in \bar V_j$ $j=1,\cdots,8$ (see Fig. \ref{figGamma3}) that satisfy the following boundary conditions:}
			\begin{align*}\label{Gj}
				R_1(k)&=\left\lbrace\begin{aligned}
					&\rho_1^*(k), \qquad k
					\in (k_0,\infty),\\
					&f_1=\rho_1^*(k_0)(k-k_0)^{-2\ri\nu}{\rm{e}}^{-2\chi{(\zeta,k_0)}}\delta_2^{-1}(\zeta,k_0)\delta_3^{-1}(\zeta,k_0)\frac{\delta_2(k)\delta_3(k)}{\delta_{1+}^2(k)}, \qquad  k\in \Gamma_1^{(3)},
				\end{aligned}
				\right.\\
				R_2(k)&=\left\lbrace\begin{aligned}
					&r_1(k), \qquad k
					\in \left( \frac{k_0(\sqrt{3}-1)}{2},k_0\right) ,\\
					&f_2=r_1(k_0)(k-k_0)^{2\ri\nu}{\rm{e}}^{2\chi{(\zeta,k_0)}}\delta_2(\zeta,k_0)\delta_3(\zeta,k_0)\frac{\delta_{1}^2(k)}{\delta_2(k)\delta_3(k)}, \qquad  k\in \Gamma_2^{(3)},
				\end{aligned}
				\right.\\
				R_3(k)&=\left\lbrace\begin{aligned}
					&r_1^*( k), \qquad k
					\in \left( \frac{k_0(\sqrt{3}-1)}{2},k_0\right),\\
					&f_3=r_1^*(k_0)(k-k_0)^{-2\ri\nu}{\rm{e}}^{-2\chi{(\zeta,k_0)}}\delta_2^{-1}(\zeta,k_0)\delta_3^{-1}(\zeta,k_0)\frac{\delta_2(k)\delta_3(k)}{\delta_{1}^2(k)}, \qquad k\in \Gamma_3^{(3)},
				\end{aligned}
				\right.\\
				R_4(k)&=\left\lbrace\begin{aligned}
					&\rho_1( k), \qquad k
					\in (k_0,\infty),\\
					&f_4=\rho_1(k_0)(k-k_0)^{2\ri\nu}{\rm{e}}^{2\chi{(\zeta,k_0)}}\delta_2(\zeta,k_0)\delta_3(\zeta,k_0)\frac{\delta_{1-}^2(k)}{\delta_2(k)\delta_3(k)}, \qquad  k\in \Gamma_4^{(3)},
				\end{aligned}
				\right.\\
				R_5(k)&=\left\lbrace\begin{aligned}
					&r_1( k), \qquad k
					\in \left(0, \frac{k_0(\sqrt{3}-1)}{2}\right),\\
					&f_5=r_1(0), \qquad  k\in \Gamma_{19}^{(3)},
				\end{aligned}
				\right.\\
				R_6(k)&=\left\lbrace\begin{aligned}
					&r_1^*( k), \qquad k
					\in \left(0, \frac{k_0(\sqrt{3}-1)}{2}\right),\\
					&f_6=r_1^*(0), \qquad  k\in \Gamma_{21}^{(3)},
				\end{aligned}
				\right.\\
				R_j(k)&=I,\qquad k\in V_j,\quad j=7,8.
			\end{align*}
			Moreover, $R_j$ satisfies the following estimates
			\begin{equation*}
				\left| R_j(k)\right|\leq c\left( 1+\left|{\rm{Re}}\,k \right|^2 \right)^{-1/4}, \quad j=1,\cdots,8,
			\end{equation*}
			and
			\begin{align*}
				&\left|\bar \partial R_1(k) \right| \leq c\left| k-k_0\right|^{-1/2} +c\left| {\rho_1}'({\rm{Re}}\,k)\right| ,\quad &&\left|\bar \partial R_2(k) \right| \leq c\left| k-k_0\right|^{-1/2} +c\left| r_1'({\rm{Re}}\,k)\right| ,\\
				&\left|\bar \partial R_3(k) \right| \leq c\left| k-k_0\right|^{-1/2} +c\left| {r_1}'({\rm{Re}}\,k)\right| ,\quad &&\left|\bar \partial R_4(k) \right| \leq c\left| k-k_0\right|^{-1/2} +c\left| \rho_1'({\rm{Re}}\,k)\right| ,\\
				&\left|\bar \partial R_5(k) \right| \leq c\left| k\right|^{-1/2} +c\left| r_1'({\rm{Re}}\,k)\right| ,\quad &&\left|\bar \partial R_6(k) \right| \leq c\left| k\right|^{-1/2} +c\left| {r_1}'({\rm{Re}}\,k)\right|,
			\end{align*}
			for positive real number $c$.
		\end{prop}
		
		\begin{proof}
			We only prove the proposition for $R_1$, $R_2$ and $R_5$, the proofs of $R_3$, $R_4$ and $R_6$ are be obtained by the relations
		\begin{equation*}
				R_3(k)={R_2( k)},\quad R_4(k)={R_1( k)},\quad R_6(k)={R_5(k)}.
			\end{equation*}
			The function $R_1(k)$ can be defined as
			\begin{equation}\label{R1}
				R_1(k)=\cos(2\phi_1)\rho_1^*({\rm{Re}}\,k)+[1-\cos(2\phi_1)]f_1(k),
			\end{equation}
			where $k=k_0+\rho{\rm{e}}^{\ri\phi_1}$, $\rho\geq0$, $0\leq \phi_1\leq \pi/4$. It can be readily verified that this function satisfies the boundary conditions.
			Noting that $f_1(k)$ is an analytic function on $\bar V_1$,  the calculation yields
			\begin{align*}
				\bar \partial R_1&=\left[ \rho_1^*({\rm{Re}}\,k)-f_1(k)\right] \bar \partial \cos(2\phi_1)+\cos(2\phi_1) \bar \partial\rho_1^*({\rm{Re}}\,k)\\
				&=-\frac{\ri{\rm{e}}^{\ri\phi_1}}{\rho}\sin(2\phi_1)[\rho_1^*({\rm{Re}}\,k)-f_1(k)]+\frac{1}{2}\cos(2\phi_1)\rho_1^{*'}({\rm{Re}}\,k).
			\end{align*}
			Then,
			\begin{align*}
				\left|\bar \partial R_1 \right| \leq \frac{c}{\left| k-k_0\right| }\left[ \left|\rho_1^*({\rm{Re}}\,k)-\rho_1^*(k_0) \right| +\left|\rho_1^*(k_0)-f_1(k) \right| \right] +c\left| \rho_1^{'}({\rm{Re}}\,k)\right|.
			\end{align*}
			First, we have
			\begin{align*}
				\left|\rho_1^*({\rm{Re}}\,k)-\rho_1^*(k_0) \right|&=\left| \int_{k_0}^{{\rm{Re}}\,k}\rho_1^{*'}(s)\rd s\right| \leq \int_{k_0}^{{\rm{Re}}\,k}\left| \rho_1^{*'}(s)\right| \rd s\\
				&\leq\left\| \rho_1^{*'}\right\|_{L^2(k_0,{\rm{Re}}\,k)} \left\| 1\right\|_{L^2(k_0,{\rm{Re}}\,k)} \leq c \left| k-k_0\right|^{1/2}.
			\end{align*}
			Secondly,
			\begin{align*}
				f_1(k)&=\rho_1^*(k_0)(k-k_0)^{-2i\nu}{\rm{e}}^{-2\chi{(\zeta,k_0)}}\delta_2^{-1}(\zeta,k_0)\delta_3^{-1}(\zeta,k_0)\frac{\delta_2(k)\delta_3(k)}{\delta_{1+}^2(k)}\\
				&=\rho_1^*(k_0){\rm{e}}^{2\chi(\zeta,k)-2\chi{(\zeta,k_0)}}\frac{\delta_1(\omega^2k)}{\delta_1(\zeta,\omega^2k_0)}\frac{\delta_1(\omega k)}{\delta_1(\zeta,\omega k_0)},
			\end{align*}
			is uniformly bounded for $k\in \Gamma_1^{(3)}$ and $\zeta\in\mathcal{I}$, using Proposition \ref{prop_delta}.
			Within $V_1$, we have
			\begin{align*}
				\left|\rho_1^*(k_0)-f_1(k) \right|=&\left| \rho_1^*(k_0)-\rho_1^*(k_0){\rm{e}}^{2\chi(\zeta,k)-2\chi{(\zeta,k_0)}}\frac{\delta_1(\omega^2k)}{\delta_1(\zeta,\omega^2k_0)}\frac{\delta_1(\omega k)}{\delta_1(\zeta,\omega k_0)}\right|\\
				=&\left|\rho_1^*(k_0) \right|\left|1-{\rm{exp}}\left\lbrace2\chi(\zeta,k)-2\chi{(\zeta,k_0)}-\ri\nu \ln_0\left( \frac{\omega^2 k-k_0}{\omega^2 k_0-k_0}\right)\right.   \right.\\
				&\left.\left.  -[\chi(\zeta,\omega^2k)-\chi(\zeta,\omega^2k_0)]-\ri\nu \ln_0\left( \frac{\omega k-k_0}{\omega k_0-k_0}\right) -[\chi(\zeta,\omega k)-\chi(\zeta,\omega k_0)]  \right\rbrace \right| \\
				=&\left|\rho_1^*(k_0) \right|\left|1-{\rm{exp}}\left\lbrace  \mathcal{O}\left( \left|k-k_0 \right|^{1/2} \right)  \right\rbrace\right| \\
				=&\left|\rho_1^*(k_0) \right| \mathcal{O}\left( \left|k-k_0 \right|^{1/2} \right).
			\end{align*}
			Based on the above estimates, we finally obtain
			\begin{equation*}
				\left|\bar \partial R_1 \right| \leq c\left| k-k_0\right|^{-1/2} +c\left| \rho_1^{'}({\rm{Re}}\,k)\right| .
			\end{equation*}
			 {Using equation \eqref{R1} and $r_1\in H^{3,4}(0,\infty)$, we have}
			\begin{align*}
				\left| R_1\right| &\leq \sin^2(\phi_1) \left|2 f_1(k) \right| +\left|\cos(2\phi_1)\right| \left| \rho_1^*({\rm{Re}}\,k) \right|\\
				&\leq c \sin^2(\phi_1)+c[1+({\rm{Re}}\, k)^2]^{-1/4}\\
				&\leq c\left( 1+\left|{\rm{Re}}\,k \right|^2 \right)^{-1/4}.
			\end{align*}
			
			For the case of $R_2$, we provide a construction that is somewhat different from $R_1$. The function $R_2(k)$ can be defined as
			\begin{equation}\label{R2}
				R_2(k)=f_2(k)+[r_1({\rm{Re}}\,k)-f_2(k)]\mathcal{K}_2(\phi_2),
			\end{equation}
			where $k=k_0+\rho{\rm{e}}^{\ri\phi_2}$, $\rho\geq0$, $\phi_2=\arg_0 (k-k_0)\in \left[\frac{3\pi}{4},\pi \right] $, $\mathcal{K}_2$ is a smooth function on $\left( \frac{3\pi}{4},\pi\right) $ and
			\begin{equation*}
				\mathcal{K}_2(\phi_2)=\left\lbrace
				\begin{aligned}
					&1, \quad &&\phi_2\in \left[ \pi-\arctan\left(  \frac{3+\sqrt{3}}{9}\right) ,\pi\right] ,\\
					&0, \quad &&\phi_2\in \left[ \frac{3\pi}{4},\pi-\arctan\left(  \frac{3+\sqrt{3}}{6}\right)\right] .\\
				\end{aligned}
				\right.
			\end{equation*}
			 It can be readily verified that the constructed function $R_2$ satisfies the boundary conditions. In addition,
			 \begin{align*}
			 	\bar{\partial}R_2&=[r_1({\rm{Re}}\,k)-f_2(k)]\bar{\partial}\mathcal{K}_2(\phi_2)+\mathcal{K}_2(\phi_2)\bar{\partial}r_1({\rm{Re}}\,k)\\
			 	&=-\frac{i{\rm{e}}^{i\phi_2}}{\rho}\mathcal{K}_2'(\phi_2)[r_1({\rm{Re}}\,k)-f_2(k)]+\frac{1}{2}\mathcal{K}_2(\phi_2)r_1'({\rm{Re}}\,k),
			 \end{align*}
			 and similarly to the proof for $R_1$, we obtain the following estimate
			 \begin{align*}
			 	\left|\bar{\partial}R_2 \right| \leq c\left| k-k_0\right|^{-1/2} +c\left| r_1'({\rm{Re}}\,k)\right| .
			 \end{align*}
			Based on the construction of $R_2$ in equation \eqref{R2}, we have
				\begin{align*}
				\left| R_2\right| \leq  c\left( \left|f_2(k) \right| + \left| r_1({\rm{Re}}\,k) \right|\right) \leq c\left( 1+\left|{\rm{Re}}\,k \right|^2 \right)^{-1/4}.
			\end{align*}
			
		Define
			 \begin{equation}\label{R5}
			 	R_5(k)=f_5(k)+[r_1({\rm{Re}}\,k)-f_5(k)]\mathcal{K}_5(\phi_5),
			 \end{equation}
			 where $k=\rho{\rm{e}}^{\ri\phi_5}$, $\rho\geq0$, $\phi_5=\arg_0 k\in \left[0,\frac{\pi}{3} \right] $, $\mathcal{K}_5$ is a smooth function on $\left( 0,\frac{\pi}{3}\right) $ and
			 \begin{equation*}
			 	\mathcal{K}_5(\phi_5)=\left\lbrace
			 	\begin{aligned}
			 		&1, \quad &&\phi_5\in \left[ 0,\arctan\left(  \frac{1+\sqrt{3}}{3}\right)\right] ,\\
			 		&0, \quad &&\phi_2\in \left[ \arctan\left(  \frac{1+\sqrt{3}}{2}\right),\frac{\pi}{3}\right] .\\
			 	\end{aligned}
			 	\right.
			 \end{equation*}
			 The proof for $R_5$ follows from Proposition \ref{prop_G_exist}, yielding
			 \begin{equation*}
			 	\left|\bar{\partial}R_5 \right| \leq c\left| k\right|^{-1/2} +c\left| r_1'({\rm{Re}}\,k)\right|,
			 \end{equation*}
			 and from equation \eqref{R5},
			 \begin{align*}
			 	\left| R_5\right| &\leq  c\left( \left|f_5(k) \right| + \left| r_1({\rm{Re}}\,k) \right|\right) \\
			 	&\leq c\left( 1+\left|{\rm{Re}}\,k \right|^2 \right)^{-1/4}.
			 \end{align*}

		\end{proof}
		
		\begin{figure}[htbp]
			\centering
			\begin{tikzpicture}[scale=1.6]
				\draw [dashed,very thick,white!50!blue](3,0) -- (-3,0);
				
				\fill (2,0) circle (1.5pt);
                
				\draw [very thick,black!20!blue](2/2.732,2*1.732/2.732) -- (2.8,-0.8);

				\draw [very thick,black!20!blue, latex-](1.3,0.7) -- (2,0);
				\draw [very thick,black!20!blue, -latex](2,0) -- (2.5,-0.5);
				
				\draw [very thick,black!20!blue, latex-](1.3,-0.7) -- (2,0);
				\draw [very thick,black!20!blue, -latex](2,0) -- (2.5,0.5);
				
				\draw [very thick,black!20!blue](2/2.732,2*1.732/2.732) -- (2/2.732,-2*1.732/2.732);
				\draw [very thick,black!20!blue](2/2.732,2*1.732/2.732) -- (2/2.732,-2*1.732/2.732);
				
				\draw [very thick,black!20!blue](2/2.732,-2*1.732/2.732) -- (2.8,0.8);
				
				\draw [very thick,black!20!blue, -latex](2/2.732,0) -- (2/2.732,-0.7);
				\draw [very thick,black!20!blue, latex-](2/2.732,0.7) -- (2/2.732,0);
				
				\fill (2/2.732,0) circle (1pt);
				
				\node[red!70!black,above] at (2.3,0.4) {$1$};
				\node[red!70!black,above] at (1.4,0.65) {$2$};
				\node[red!70!black,below] at (1.4,-0.65) {$3$};
				\node[red!70!black,below] at (2.3,-0.4) {$4$};
				\node[gray,below] at (0.25,0.7) {$19$};
				\node[red!70!black,right] at (2/2.732,0.6) {$5$};
				\node[red!70!black,right] at (2/2.732,-0.6) {$6$};
				
				\coordinate (center) at (0,0);
				\begin{scope}[rotate around={120:(center)}]
					\draw [dashed,very thick,white!50!blue](3,0) -- (-3,0);
					
					\fill (2,0) circle (1.5pt);
					
					\draw [very thick,black!20!blue](2/2.732,2*1.732/2.732) -- (2.8,-0.8);
					
					\draw [very thick,black!20!blue, latex-](1.3,0.7) -- (2,0);
					\draw [very thick,black!20!blue, -latex](2,0) -- (2.5,-0.5);
					
					\draw [very thick,black!20!blue, latex-](1.3,-0.7) -- (2,0);
					\draw [very thick,black!20!blue, -latex](2,0) -- (2.5,0.5);
					
					\draw [very thick,black!20!blue](2/2.732,2*1.732/2.732) -- (2/2.732,-2*1.732/2.732);
					\draw [very thick,black!20!blue](2/2.732,2*1.732/2.732) -- (2/2.732,-2*1.732/2.732);

					\draw [very thick,black!20!blue](2/2.732,-2*1.732/2.732) -- (2.8,0.8);
					
					\draw [very thick,black!20!blue, -latex](2/2.732,0) -- (2/2.732,-0.7);
					\draw [very thick,black!20!blue, latex-](2/2.732,0.7) -- (2/2.732,0);
					
					\fill (2/2.732,0) circle (1pt);
					
					\node[red!70!black,below] at (2.35,0.45) {$7$};
					\node[red!70!black,above] at (1.35,0.85) {$8$};
					\node[red!70!black,below] at (1.4,-0.6) {$9$};
					\node[red!70!black,right] at (2.4,-0.4) {$10$};
					\node[gray,below] at (0.45,0.45) {$20$};
					\node[red!70!black,right] at (2.7/2.732,0.7) {$11$};
					\node[red!70!black,right] at (2.7/2.6,-0.4) {$12$};
				\end{scope}
				
				\coordinate (center) at (0,0);
				\begin{scope}[rotate around={240:(center)}]
					\draw [dashed,very thick,white!50!blue](3,0) -- (-3,0);
					
					\fill (2,0) circle (1.5pt);
					
					\draw [very thick,black!20!blue](2/2.732,2*1.732/2.732) -- (2.8,-0.8);
					
					\draw [very thick,black!20!blue, latex-](1.3,0.7) -- (2,0);
					\draw [very thick,black!20!blue, -latex](2,0) -- (2.5,-0.5);
					
					\draw [very thick,black!20!blue, latex-](1.3,-0.7) -- (2,0);
					\draw [very thick,black!20!blue, -latex](2,0) -- (2.5,0.5);
					
					\draw [very thick,black!20!blue](2/2.732,2*1.732/2.732) -- (2/2.732,-2*1.732/2.732);
					\draw [very thick,black!20!blue](2/2.732,2*1.732/2.732) -- (2/2.732,-2*1.732/2.732);

					\draw [very thick,black!20!blue](2/2.732,-2*1.732/2.732) -- (2.8,0.8);
					
					\draw [very thick,black!20!blue, -latex](2/2.732,0) -- (2/2.732,-0.7);
					\draw [very thick,black!20!blue, latex-](2/2.732,0.7) -- (2/2.732,0);
					
					\fill (2/2.732,0) circle (1pt);
					
					\node[red!70!black,above] at (2.5,0.7) {$13$};
					\node[red!70!black,above] at (1.35,0.6) {$14$};
					\node[red!70!black,below] at (1.35,-0.9) {$15$};
					\node[red!70!black,above] at (2.4,-0.45) {$16$};
					\node[gray,below] at (0.2,0.45) {$21$};
					\node[red!70!black,right] at (2.7/2.732,0.5) {$17$};
					\node[red!70!black,right] at (2.7/2.732,-0.7) {$18$};
				\end{scope}
				
				\node[above] at (2,0.1) {$k_0$};
				\node[right] at (-0.9,1.9) {$\omega k_0$};
				\node[right] at (-0.9,-1.9) {$\omega^2 k_0$};
				
				\node[above] at (2.7,0.1) {$V_1$};
				\node[above] at (1.25,0.1) {$V_2$};
				\node[below] at (1.25,-0.2) {$V_3$};
				\node[below] at (2.7,-0.2) {$V_4$};
				\node[above] at (0.5,0.1) {$V_5$};
				\node[below] at (0.5,-0.1) {$V_6$};
				\node[below] at (2,1.3) {$V_7$};
				\node[below] at (2,-1.3) {$V_8$};
				
				\node[right] at (3,0) {${\rm{Re}\,k}$};
				
				\draw [very thick,dashed] (2+0.2/1.414,0.2/1.414) -- (2+0.2*1.414,0);
				\draw [very thick,dashed] (2+0.2*1.414,0) -- (2+0.2/1.414,-0.2/1.414);
				
				\node[red!70!black,above] at (2.3,0.4) {$1$};
				\node[red!70!black,above] at (1.4,0.65) {$2$};
				\node[red!70!black,below] at (1.4,-0.65) {$3$};
				\node[red!70!black,below] at (2.3,-0.4) {$4$};
				\node[gray,below] at (0.25,0.7) {$19$};
				\node[red!70!black,right] at (2/2.732,0.6) {$5$};
				\node[red!70!black,right] at (2/2.732,-0.6) {$6$};
				
				\fill (0,0) circle (1pt);
			\end{tikzpicture}
			\caption{The jump contour $\Gamma^{(3)}$ in the complex $k$-plane.}
			\label{figGamma3}
		\end{figure}
		
		Noticing that combinations of some values of the reflection coefficients $r_j$, $j=1,2$, appear in the jump matrix $v^{(2)}$, we define the notation for the sake of simplifying the writing
			\begin{equation*}
			\alpha_j(k)=r_2(0)+R_j(k)r_2^*(0),\quad k\in \Gamma^{(3)}_j,\,j=1,2,\cdots,6.
		\end{equation*}

		Next, we provide six matrix functions in the following form
		\begin{align*}
			&v_{1,R}=\begin{pmatrix}
				1 & 0 & 0\\
				\frac{\delta_{1+}^2}{\delta_2\delta_3}R_1(k){\rm{e}}^{t \Phi_{21}} & 1 & -\frac{\delta_{1+}\delta_2}{\delta_3^2}r_2^*(0){\rm{e}}^{-t \Phi_{32}}\\
				0 & 0 & 1
			\end{pmatrix}, \,\, v_{2,R}=\begin{pmatrix}
			1 & -\frac{\delta_2\delta_3}{\delta_1^2}R_2(k) {\rm{e}}^{-t \Phi_{21}} & \frac{\delta_2^2}{\delta_1\delta_3}\alpha_2(k){\rm{e}}^{-t \Phi_{31}} \\
			0 & 1 & 0 \\
			0 & 0 & 1
			\end{pmatrix},\\
			&  v_{3,R}=\begin{pmatrix}
				1 & 0 & 0 \\
				\frac{\delta_1^2 }{\delta_2\delta_3}R_3(k) {\rm{e}}^{t \Phi_{21}} & 1 & -\frac{\delta_1\delta_2}{\delta_3^2}\alpha_3^*(k){\rm{e}}^{-t \Phi_{32}} \\
				0 & 0 & 1
			\end{pmatrix},\,\,\,  v_{4,R}=\begin{pmatrix}
			1 & -\frac{\delta_2\delta_3}{ \delta_{1-}^2} R_4(k){\rm{e}}^{-t \Phi_{21}} & \frac{\delta_2^2}{\delta_{1-}\delta_3}r_2(0){\rm{e}}^{-t \Phi_{31}}\\
			0 & 1 & 0\\
			0 & 0 & 1
			\end{pmatrix},\\
				& v_{5,R}=\begin{pmatrix}
				1 & -\frac{\delta_2\delta_3}{\delta_1^2}R_5(k) {\rm{e}}^{-t \Phi_{21}} & \frac{\delta_2^2}{\delta_1\delta_3}\alpha_5(k){\rm{e}}^{-t \Phi_{31}} \\
				0 & 1 & 0 \\
				0 & 0 & 1
			\end{pmatrix},\,v_{6,R}=\begin{pmatrix}
				1 & 0 & 0 \\
				\frac{\delta_1^2 }{\delta_2\delta_3}R_6(k) {\rm{e}}^{t \Phi_{21}} & 1 & -\frac{\delta_1\delta_2}{\delta_3^2}\alpha_6^*(k){\rm{e}}^{-t \Phi_{32}} \\
				0 & 0 & 1
			\end{pmatrix}.
		\end{align*}
		Define the matrix-valued function
		\begin{equation*}
			R^{(3)}_{k_0}(\zeta,k)=\left\lbrace \begin{aligned}
				& (v_{1,R}(k))^{-1},\quad &&k\in V_1,\\
				& (v_{2,R}(k))^{-1},\quad &&k\in V_2,\\
				& v_{3,R}(k),\quad &&k\in V_3,\\
				& v_{4,R}(k),\quad &&k\in V_4,\\
				& (v_{5,R}(k))^{-1},\quad &&k\in V_5,\\
				& v_{6,R}(k),\quad &&k\in V_6,\\
				& I, \quad && k\in V_7\cup V_8,
			\end{aligned}
			\right.
		\end{equation*}
		and
		\begin{align*}
			R^{(3)}_{\omega k_0}(\zeta,k)=\mathcal{A}^{-1}R^{(3)}_{k_0}(\zeta,\omega^2 k)\mathcal{A},\qquad R^{(3)}_{\omega^2 k_0}(\zeta,k)=\mathcal{A}^{-2}R^{(3)}_{k_0}(\zeta,\omega k)\mathcal{A}^{2}.
		\end{align*}
		According to the above definitions, we construct the following matrix function:
		\begin{equation}\label{R3}
			R^{(3)}(\zeta,k)=\bigcup_{j=0}^2R^{(3)}_{\omega^jk_0}(\zeta,k).
		\end{equation}
		 {Here, some explanations are needed. It can be noted that $R^{(3)}_{\omega^jk_0}(\zeta,k)$ are all defined in the sectors with an angle of ${2\pi}/{3}$ centered at the critical points $\omega^jk_0$ , $j=0,1,2,$ in Fig. \ref{figGamma3}, respectively. The meaning of the union in the above expression is that the piecewise-defined functions $R^{(3)}_{\omega^jk_0}$ in these three sectors together combine to form the function $R^{(3)}$ defined on $\mathbb{C}\setminus\Gamma^{(3)}$.}
		
		 {\begin{lem}\label{R_bounded}
			$R^{(3)}(x,t,k)$ and $(R^{(3)}(x,t,k))^{-1}$ are uniformly bounded for $k\in\mathbb{C}\setminus \Gamma^{(3)}$, $t>0$.
			More importantly, $\partial_x^l\lim\limits_{k\to\infty}k \left[ \left( R^{(3)}(k)\right) ^{\pm1}-I\right]=O$ for $k\in V_j$, $j=1,4,7,8$ and $l=0,1$.
		\end{lem}}
		\begin{proof} {
		Since $R^{(3)}_{\omega^jk_0}(\zeta,k)=I$ for $k\in \omega^j V_m$, $j=0,1,2$, $m=7,8$, the conclusion holds automatically. The proof process is similar to the proof of Lemma \ref{lem_G_bounded}, and we provide a brief explanation for region $V_1$. For $k\in V_1$, following a calculation similar to Lemma \ref{lem_G_bounded}, we have ${\rm{e}}^{t{\rm{Re}}\,\Phi_{21}}=\mathcal{O}({\rm{e}}^{-ct|k|^2})$ and  ${\rm{e}}^{-t{\rm{Re}}\,\Phi_{32}}=\mathcal{O}({\rm{e}}^{-ct|k|^2})$ as $k\to\infty$. Calculate the following expression:
			\begin{align*}
				\left| R^{(3)}_1(k)-I\right|&= \left| \left( v_{1,R}(k)\right) ^{-1}-I\right|=\left| \frac{\delta_{1+}^2}{\delta_2\delta_3}R_1(k){\rm{e}}^{t \Phi_{21}}\right|+\left|\frac{\delta_{1+}\delta_2}{\delta_3^2}r_2^*(0){\rm{e}}^{-t \Phi_{32}} \right| \\
				&\leq c\left( 1+\left|{\rm{Re}}\,k \right|^2 \right)^{-1/4}  {\rm{e}}^{t\,{\rm{Re}}\,\Phi_{21}} +c{\rm{e}}^{-t\,{\rm{Re}}\,\Phi_{32}}.
			\end{align*}
			Then, we have  $\lim\limits_{k\to\infty}R^{(3)}_1(k)=I$, as $k\in V_1$. Similarly, $\lim\limits_{k\to\infty}\left( R^{(3)}_1(k)\right) ^{-1}=I$ 
			and calculate $\lim\limits_{k\to\infty}k\left(  R^{(3)}_1(k)-I\right)$, with only the elements at positions (2,1) and (2,3) are non-zero:
			\begin{equation*}
				\lim_{k\to\infty}\left| k \frac{\delta_{1+}^2}{\delta_2\delta_3}R_1(k){\rm{e}}^{t \Phi_{21}}\right|
				\leq c\lim_{k\to\infty} \left| k \left( 1+\left|{\rm{Re}}\,k \right|^2 \right)^{-1/4} \right|  {\rm{e}}^{t\, {\rm{Re}}\,\Phi_{21}} =0,
			\end{equation*}
			and
			\begin{equation*}
				\lim_{k\to\infty}\left| k \frac{\delta_{1+}\delta_2}{\delta_3^2}r_2^*(0){\rm{e}}^{-t \Phi_{32}}\right|
				\leq c\lim_{k\to\infty} \left| k\right|  {\rm{e}}^{-t\, {\rm{Re}}\,\Phi_{32}}  =0.
			\end{equation*}
			The above equation also holds for $\left( R^{(3)}(k)\right) ^{-1}$ and the case for $l=1$.
			 The proofs for other regions can be similarly obtained, and thus will not be elaborated upon here.
			In summary, $\left( R^{(3)}\right) ^{\pm1}=I+o(k^{-1})$.}
		\end{proof}
		
			\begin{lem}\label{lemDbarR3}
			$R^{(3)}(x,t,k)$ is a continuous but non-analytic function within its domain, and it satisfies the relation
			\begin{equation*}
				\left( R^{(3)}(k)\right) ^{-1}\bar{\partial}R^{(3)}(k)=\bar{\partial}R^{(3)}(k),
			\end{equation*}
			where
			\begin{align}\label{DbarR3}
				\bar{\partial}R^{(3)}_1(k)&=\begin{pmatrix}
					0 & 0 & 0\\
					-\bar{\partial}R_1(k)\frac{\delta_{1+}^2}{\delta_2\delta_3}{\rm{e}}^{t \Phi_{21}} & 0 & 0\\
					0 & 0 & 0
				\end{pmatrix},\qquad && k\in V_1,\nonumber\\
				\bar{\partial}R^{(3)}_2(k)&=\begin{pmatrix}
					0 & \bar{\partial}R_2(k)\frac{\delta_2\delta_3}{\delta_1^2} {\rm{e}}^{-t \Phi_{21}} & -\bar{\partial}\alpha_2(k)\frac{\delta_2^2}{\delta_1\delta_3}{\rm{e}}^{-t \Phi_{31}} \\
					0 & 0 & 0 \\
					0 & 0 & 0
				\end{pmatrix}, && k\in V_2,\nonumber\\
				\bar{\partial}R^{(3)}_3(k)&=\begin{pmatrix}
					0 & 0 & 0 \\
					\bar{\partial}R_3(k)\frac{\delta_1^2 }{\delta_2\delta_3} {\rm{e}}^{t \Phi_{21}} & 0 & -\bar{\partial}\alpha_3^*(k)\frac{\delta_1\delta_2}{\delta_3^2}{\rm{e}}^{-t \Phi_{32}} \\
					0 & 0 & 0
				\end{pmatrix}, && k\in V_3,\nonumber\\
				\bar{\partial}R^{(3)}_4(k)&=\begin{pmatrix}
					0 & -\bar{\partial}R_4(k)\frac{\delta_2\delta_3}{ \delta_{1-}^2} {\rm{e}}^{-t \Phi_{21}} & 0\\
					0 & 0 & 0\\
					0 & 0 & 0
				\end{pmatrix}, && k\in V_4,\\
				\bar{\partial}R^{(3)}_5(k)&=\begin{pmatrix}
					0 & \bar{\partial}R_5(k)\frac{\delta_2\delta_3}{\delta_1^2} {\rm{e}}^{-t \Phi_{21}} & -\bar{\partial}\alpha_5(k)\frac{\delta_2^2}{\delta_1\delta_3}{\rm{e}}^{-t \Phi_{31}} \\
					0 & 0 & 0 \\
					0 & 0 & 0
				\end{pmatrix}, && k\in V_5,\nonumber\\
				\bar{\partial}R^{(3)}_6(k)&=\begin{pmatrix}
					0 & 0 & 0 \\
					\bar{\partial}R_6(k)\frac{\delta_1^2 }{\delta_2\delta_3} {\rm{e}}^{t \Phi_{21}} & 0 & -\bar{\partial}\alpha_6^*(k)\frac{\delta_1\delta_2}{\delta_3^2}{\rm{e}}^{-t \Phi_{32}} \\
					0 & 0 & 0
				\end{pmatrix}, && k\in V_6,\nonumber\\
				\bar{\partial}R^{(3)}_j(k)&=O, && k\in V_j,\, j=7,8,\nonumber
			\end{align}
		and
		 \begin{equation*}
		 	\bar{\partial}R^{(3)}_j(k)=\mathcal{A}^{-1}\bar{\partial}R^{(3)}_{j-8}(\omega^2k)\mathcal{A},\quad j=9,10,\cdots,24.
		 \end{equation*}
		\end{lem}
		 {The situation here is similar to Remark \ref{remark_3.1}, so no further details are provided.}

		Construct the third transformation
		\begin{equation*}\label{trans3}
			M^{(3)}(\zeta,k)=M^{(2)}(\zeta,k)R^{(3)}(\zeta,k).
		\end{equation*}
		The eigenfunction $M^{(3)}$ obtained under this transformation satisfies the jump relation $M^{(3)}_+(\zeta,k)=M^{(3)}_-(\zeta,k)v^{(3)}_j(\zeta,k)$ for $k\in \Gamma^{(3)}_j$, $j=1,2,\cdots,18$, (see Fig. \ref{figGamma3}). In addition, $V_j$, $j=1,\cdots,8$, have been marked in Fig. \ref{figGamma3}, and $V_j$ for $j=9,\cdots,24$, are given according to the symmetry relation $V_{j}=\omega V_{j-8}$, which is not marked in Fig. \ref{figGamma3} for the aesthetic appearance of the figure. More importantly, $v^{(3)}_j(\zeta,k)$, $j=1,2,\cdots,18,$ are given by
		\begin{align}\label{v3_1}
			&v^{(3)}_1(k)=v_{1,R}(k)=\begin{pmatrix}
				1 & 0 & 0\\
				\frac{\delta_{1+}^2(k)}{\delta_2(k)\delta_3(k)}R_1(k){\rm{e}}^{t \Phi_{21}} & 1 & -\frac{\delta_{1+}(k)\delta_2(k)}{\delta_3^2(k)}r_2^*(0){\rm{e}}^{-t \Phi_{32}}\\
				0 & 0 & 1
			\end{pmatrix},\nonumber\\
			 &v_2^{(3)}(k)=v_{2,R}^{-1}(k)=\begin{pmatrix}
			1 & \frac{\delta_2(k)\delta_3(k)}{\delta_1^2(k)}R_2(k) {\rm{e}}^{-t \Phi_{21}} & -\frac{\delta_2^2(k)}{\delta_1(k)\delta_3(k)}\alpha_2(k){\rm{e}}^{-t \Phi_{31}} \\
			0 & 1 & 0 \\
			0 & 0 & 1
			\end{pmatrix},\nonumber\\ &v^{(3)}_3(k)=v_{3,R}^{-1}(k)=\begin{pmatrix}
			1 & 0 & 0 \\
			-\frac{\delta_1^2(k) }{\delta_2(k)\delta_3(k)}R_3(k) {\rm{e}}^{t \Phi_{21}} & 1 & \frac{\delta_1(k)\delta_2(k)}{\delta_3^2(k)}\alpha_3^*(k){\rm{e}}^{-t \Phi_{32}} \\
			0 & 0 & 1
			\end{pmatrix},\\
			 &v^{(3)}_4(k)=v_{4,R}(k)=\begin{pmatrix}
			 	1 & -\frac{\delta_2(k)\delta_3(k)}{ \delta_{1-}^2(k)} R_4(k){\rm{e}}^{-t \Phi_{21}} & \frac{\delta_2^2(k)}{\delta_{1-}(k)\delta_3(k)}r_2(0){\rm{e}}^{-t \Phi_{31}}\\
			 	0 & 1 & 0\\
			 	0 & 0 & 1
			 \end{pmatrix},\nonumber\\
			 &v^{(3)}_5(k)=v_{2,R}(k)v_{5,R}^{-1}(k)=\begin{pmatrix}
			 	1 & \frac{\delta_2(k)\delta_3(k)}{\delta_1^2(k)}[R_5(k)\!-\!R_2(k)]{\rm{e}}^{-t\Phi_{21}} & -\frac{\delta_2^2(k)}{\delta_1(k)\delta_3(k)}[\alpha_5(k)\!-\!\alpha_2(k)]{\rm{e}}^{-t\Phi_{31}}\\
			 	0 & 1 & 0\\
			 	0 & 0 & 1
			 \end{pmatrix},\nonumber\\ &v^{(3)}_6(k)=v_{6,R}^{-1}(k)v_{3,R}=\begin{pmatrix}
			 	1 & 0 & 0\\
			 	-\frac{\delta_1^2(k)}{\delta_2(k)\delta_3(k)}[R_3(k)+R_6(k)]{\rm{e}}^{t\Phi_{21}} & 1 & \frac{\delta_1(k)\delta_2(k)}{\delta_3^2(k)}[\alpha_3^*(k)+\alpha_6^*(k)]{\rm{e}}^{-t\Phi_{32}}\\
			 	0 & 0 & 1
			 \end{pmatrix},\nonumber
		\end{align}
		and
		\begin{equation}\label{v3_2}
			v^{(3)}_j(k)=\mathcal{A}^{-1}v^{(3)}_{j-6}(\omega^2k)\mathcal{A},\quad j=7,\cdots,18.
		\end{equation}

	\begin{remark}
		In Fig. \ref{figGamma3}, there are no jumps on lines $\Gamma^{(3)}_{19}$, $\Gamma^{(3)}_{20}$ and $\Gamma^{(3)}_{21}$. Taking $\Gamma^{(3)}_{19}$ as an example for illustration. For $k\in\Gamma^{(3)}_{19}$,
		\begin{align*}
			v^{(3)}_{19}(k)&=v_{5,R}(k)\mathcal{A}^{-1}v_{6,R}(\omega^2k)\mathcal{A}\\
			&=\begin{pmatrix}
				1 & -\frac{\delta_2(k)\delta_3(k)}{\delta_1^2(k)}[\alpha_5^*(k)+R_5(k)]{\rm{e}}^{-t\Phi_{21}} & \frac{\delta_2^2(k)}{\delta_1(k)\delta_3(k)}[\alpha_6^*(\omega^2k)+R_6(\omega^2k)]{\rm{e}}^{-t\Phi_{31}}\\
				0 & 1 & 0\\
				0 & 0 & 1
			\end{pmatrix},
		\end{align*}
	and $\alpha_5^*(k)+R_5(k)=r_2^*(0)+r_1^*(0)r_2(0)+r_1(0)=1+\bar{\omega}+\omega=0$, $\alpha_6^*(\omega^2k)+R_6(\omega^2k)=r_2^*(0)+r_1(0)r_2(0)+r_1^*(0)=1+\omega+\bar{\omega}=0$. In other words, $v^{(3)}_{19}=I$.
	\end{remark}
	\begin{remark}\label{rem_v56}
		Noting the definitions of $R_j(k)$ given by equations \eqref{R2} and \eqref{R5} in regions $V_j$ $(j=2,3,5,6)$, it involves two new jumps on $\Gamma^{(3)}_5$ and $\Gamma^{(3)}_6$ in Fig. \ref{figGamma3}, and the jump matrices $v_5^{(3)}$ and $v_6^{(3)}$ can be further expressed in the following forms
		\begin{equation*}
			v_5^{(3)}=\left\lbrace
			\begin{aligned}
					&I, \quad \qquad \qquad \qquad \qquad \qquad \qquad k\in \left( \frac{k_0(\sqrt{3}-1)}{2}, \frac{k_0(\sqrt{3}-1)}{2}+\frac{\ri k_0}{3} \right), \\
				&\begin{pmatrix}
					1 & \frac{\delta_2(k)\delta_3(k)}{\delta_1^2(k)}[R_5(k)-R_2(k)]{\rm{e}}^{-t\Phi_{21}} & -\frac{\delta_2^2(k)}{\delta_1(k)\delta_3(k)}[\alpha_5(k)-\alpha_2(k)]{\rm{e}}^{-t\Phi_{31}}\\
					0 & 1 & 0\\
					0 & 0 & 1
				\end{pmatrix}, \\
				&\qquad\qquad \qquad \qquad \qquad\qquad\qquad k\in \left( \frac{k_0(\sqrt{3}-1)}{2}+\frac{\ri k_0}{3}, \frac{k_0(\sqrt{3}-1)}{2}+{\ri k_0} \right)  ,
			\end{aligned}
			\right.
		\end{equation*}
			\begin{equation*}
			v_6^{(3)}=\left\lbrace
			\begin{aligned}
				&I, \quad \qquad \qquad \qquad \qquad \qquad \qquad k\in \left(\frac{k_0(\sqrt{3}-1)}{2}-\frac{\ri k_0}{3}, \frac{k_0(\sqrt{3}-1)}{2} \right), \\
				&\begin{pmatrix}
					1 & 0 & 0\\
					-\frac{\delta_1^2(k)}{\delta_2(k)\delta_3(k)}[R_3(k)+R_6(k)]{\rm{e}}^{t\Phi_{21}} & 1 & \frac{\delta_1(k)\delta_2(k)}{\delta_3^2(k)}[\alpha_3^*(k)+\alpha_6^*(k)]{\rm{e}}^{-t\Phi_{32}}\\
					0 & 0 & 1
				\end{pmatrix}, \\
				&\qquad\qquad \qquad \qquad \qquad\qquad\qquad k\in \left(\frac{k_0(\sqrt{3}-1)}{2}-{\ri k_0}, \frac{k_0(\sqrt{3}-1)}{2}-\frac{\ri k_0}{3} \right) .
			\end{aligned}
			\right.
		\end{equation*}
		Therefore, the jump matrices $v^{(3)}_j$, $j=5,6,11,12,17,18$, exponentially decay to the identity matrix as $t\to\infty$.
	\end{remark}
		
	 {The mixed problem concerning the eigenfunction $N^{(3)}(x,t,k)=(\omega\,\,\omega^2\,\, 1)M^{(3)}(x,t,k)$ is given by the following conditions:}
     {\begin{rhp-Dbar}\label{rhp-Dbar2}
    	Find a $1 \times 3$-row-vector valued function $N^{(3)}(x, t, k)$ with the following properties:
    		\begin{enumerate}
    			\item $N^{(3)}(x, t, \cdot)$ is continuous for $k\in \mathbb{C}\setminus\Gamma^{(3)}$ with continuous boundary values $N^{(3)}_{\pm}$.
    		\item The jump relation
    		\begin{equation*}
    			N^{(3)}_+(x, t, k) = N^{(3)}_-(x, t, k) v^{(3)}(x, t, k), \quad k \in \Gamma^{(3)},
    		\end{equation*}
    		where $v^{(3)}(k)=v^{(3)}_j(k)$ for $k\in \Gamma^{(3)}_j$, $j=1,\cdots,18$, are given by equations \eqref{v3_1} and \eqref{v3_2}, and the contour of  $\Gamma^{(3)}$ is given by Fig. \ref{figGamma3}.
    		\item $N^{(3)}(x, t, k)\rightarrow (\omega \,\, \omega^2 \,\, 1)$ as $k\to \infty$ in $\mathbb{C}\setminus\Gamma^{(3)}$.
    		\item As $k\to 0$, $
    		N^{(3)}(x, t, k) =\mathcal{O}(1)$.
    		\item For $k\in\mathbb{C}\setminus\Gamma^{(3)}$, there exists the $\bar{\partial}$ problem $\bar{\partial}N^{(3)}( k)=N^{(3)}(k)\bar{\partial}R^{(3)}(k) $.
    		\end{enumerate}
    \end{rhp-Dbar}}
		
		Similar to the operations after the first transformation, here we also perform the following factorisation
		\begin{equation}\label{M3factorisation}
			M^{(3)}(x,t,k)=E^{R}(x,t,k)M^{LC}(x,t,k),
		\end{equation}
		where $E^{R}(x,t,k)$ is the solution to the following pure $\bar{\partial}$ problem \ref{DbarER}, and  {$N^{LC}(x,t,k)$$=$$(\omega\,\,\omega^2\,\, 1)$ $M^{LC}(x,t,k)$} is the solution to the following pure RH problem \ref{rhpMLC}.
		 {\begin{rhp}\label{rhpMLC}
			Find a $1 \times 3$-row-vector valued function $N^{LC}(x, t, k)$ with the following properties:
			\begin{enumerate}
				\item  $N^{LC}(x, t, \cdot) : \mathbb{C} \setminus \Gamma^{(3)} \rightarrow \mathbb{C}^{1 \times 3}$ is analytic, that is, $\bar{\partial}N^{LC}(k)={O}_{1\times3}$.
				\item The jump relation
				\begin{equation*}
					N^{LC}_+(x, t, k) = N^{LC}_-(x, t, k) v^{(3)}(x, t, k), \quad k \in \Gamma^{(3)},
				\end{equation*}
				where $v^{(3)}(k)=v^{(3)}_j(k)$ for $k\in \Gamma^{(3)}_j$, $j=1,\cdots,18$, are given by equations \eqref{v3_1} and \eqref{v3_2}.
				\item $N^{LC}(x, t, k)\rightarrow (\omega \,\, \omega^2 \,\, 1)$ as $k\to \infty$ in $\mathbb{C}\setminus\Gamma^{(3)}$.
				\item As $k\to 0$, $
				N^{LC}(x, t, k) =\mathcal{O}(1)$.
			\end{enumerate}
		\end{rhp}}
		
		\begin{Dbar}\label{DbarER}
			Find a $3 \times 3$-matrix valued function $E^{R}(x, t, k)$ with the following properties:
			\begin{enumerate}
				\item $E^{R}$ is continuous in the complex plane $\mathbb{C}$.
				\item $E^{R}(x, t, k)\to I$ as $k\to \infty$ in $\mathbb{C}$.
				\item $\bar{\partial}E^{R}(k)=E^{R}(k)W(k)$,
				where
				\begin{equation*}
					W(k)=M^{LC}(k)\bar{\partial}R^{(3)}(k)\left(M^{LC}(k) \right)^{-1},\quad k\in V_j,\quad j=1,\cdots,24,
				\end{equation*}
				and $\bar{\partial}R^{(3)}(k)$ is given by equation \eqref{DbarR3}.
			\end{enumerate}
		\end{Dbar}
		
		\begin{lem}\label{estV56}
			The jump matrix $v^{(3)}(x,t,k)\to I$ and $\partial_xv^{(3)}(x,t,k)\to O$  uniformly for $\zeta\in\mathcal{I}$ and $k\in \Gamma^{(3)}$ as $t\to\infty$, except near the three critical points $k_0$, $\omega k_0$ and $\omega^2k_0$. Moreover, the jump matrices $v_j^{(3)}$, $j=5,6,$ satisfy
	\begin{align*}
		    &\left\|(1+\left|\,\cdot \,\right| ) \partial_x^l(v^{(3)}_5-I)\right\| _{(L^1\cap L^\infty)\left( \Gamma^{(3)}_5\right) }=\mathcal{O}({\rm{e}}^{-ct})   ,\\	&\left\|(1+\left|\,\cdot \,\right| ) \partial_x^l(v^{(3)}_6-I)\right\| _{(L^1\cap L^\infty)\left( \Gamma^{(3)}_6\right) }=\mathcal{O}({\rm{e}}^{-ct})  ,
			\end{align*}
	uniformly for $l=0,1$, and the same property applies to $v^{(3)}_j$, $j=11,12,17,18$.
		\end{lem}
		\begin{proof}
			First, consider the jump matrix $v_1^{(3)}$. According to Lemma \ref{lem_G_bounded}, Proposition \ref{prop_delta}, and Lemma \ref{R_bounded}, we know that $\delta_j^{\pm1}$, $j=1,2,3$, are positive and bounded, and both $R^{(3)}$ and $(R^{(3)})^{-1}$ are uniformly bounded functions. Since ${\rm{Re}}\, \Phi_{21}\leq0$ and ${\rm{Re}}\, \Phi_{32}\geq c>0$ for $k\in\Gamma^{(3)}_1$, $v_1^{(3)}\to I$ and $\partial_xv_1^{(3)}\to O$ as $t\to \infty$. It is worth noting that, due to ${\rm{Re}}\, \Phi_{21}(x,t,k_0)=0$, the element in the $(2,1)$ position of $v^{(3)}_1$ does not converge uniformly to 0 when $k$ is near $k_0$. The remaining three jump matrices $v_2^{(3)}$, $v_3^{(3)}$, and $v_4^{(3)}$ near $k_0$ can be similarly proven.
			
			Next, based on Proposition \ref{R_exist}, Lemma \ref{R_bounded} and Remark \ref{rem_v56}, we consider the case of $k\in \left( \frac{k_0(\sqrt{3}-1)}{2}+\frac{\ri k_0}{3},\right.$ $\left.\frac{k_0(\sqrt{3}-1)}{2}+{\ri k_0} \right)\subseteq\Gamma^{(3)}_5$.  Denote $k=\frac{k_0(\sqrt{3}-1)}{2}+\ri u k_0$, $1/3<u<1$, and by calculation we obtain
			\begin{equation*}
				{\rm{Re}}\,\Phi_{21}(\zeta,k)=3(\sqrt{3}-1)uk_0^2>(\sqrt{3}-1)k_0^2,
			\end{equation*}
			which leads to for $\zeta\in \mathcal{I}$,
	 {\begin{align*}
		\left\|\left( v^{(3)}_5-I\right) _{12}\right\|_{L^\infty\left( \Gamma^{(3)}_5\right) } \leq&C\sup_{k\in \Gamma^{(3)}_5}\left|\frac{\delta_2(k)\delta_3(k)}{\delta_1^2(k)}[R_5(k)-R_2(k)]{\rm{e}}^{-t\Phi_{21}}(\zeta,k) \right| \\
	    \leq&C\sup_{k\in \Gamma^{(3)}_5}\left| \frac{\delta_2(k)\delta_3(k)}{\delta_1^2(k)}\right| \cdot \left|f_5(k)+[r_1({\rm{Re}}\,k)-f_5(k)]\mathcal{K}_5(\arg_0k)\right.\\ &\left.-f_2(k)-[r_1({\rm{Re}}\,k)-f_2(k)]\mathcal{K}_2(\arg_0(k-k_0))\right| \cdot  {\rm{e}}^{-t\,{\rm{Re}}\,\Phi_{21}(\zeta,k)}  \\
	    \leq &C\sup_{k\in \Gamma^{(3)}_5}\left| \frac{\delta_2(k)\delta_3(k)}{\delta_1^2(k)}\right| \cdot \left| f_5(k)-f_2(k) \right| \cdot  {\rm{e}}^{-t\,{\rm{Re}}\,\Phi_{21}(\zeta,k)}\\
	    \leq & C{\rm{e}}^{-ctk_0^2},
  		\end{align*}}
		and
		\begin{align*}
			\left\|\left( v^{(3)}_5-I\right) _{12}\right\|_{L^1\left( \Gamma^{(3)}_5\right) } \leq C\int_{\frac{1}{3}}^{1}{\rm{e}}^{-(\sqrt{3}-1)tk_0^2}\rd u \leq  C{\rm{e}}^{-ctk_0^2}.
		\end{align*}
		As well as $\left\| \left( v^{(3)}_5-I\right) _{13}\right\|_{\left( L^1\cap L^\infty\right) \left( \Gamma^{(3)}_5\right) }=\mathcal{O}({\rm{e}}^{-ct})$ can also be proved similarly. In addition, due to the presence of the exponential term ${\rm{e}}^{-t\,{\rm{Re}}\,\Phi_{21}(\zeta,k)}$, it is possible to obtain $\left| \partial_x \left( v^{(3)}_5-I\right) _{12}\right|=\mathcal{O}({\rm{e}}^{-ct})$ and ${\rm{e}}^{-t\,{\rm{Re}}\,\Phi_{21}(\zeta,k)}$, it is possible to obtain $\left| \partial_x \left( v^{(3)}_5-I\right) _{13}\right|=\mathcal{O}({\rm{e}}^{-ct})$.
		\end{proof}
		 {\section{Treatment RH problem \ref{rhpMLC} $N^{LC}$}\label{sec_4}}
		
		\subsection{Local parametrix at $\omega^jk_0$}
		\ \ \ \
		In the previous section, we introduced that the jump matrix $v^{(3)}$ of the RH problem \ref{rhpMLC} decays exponentially to the unit matrix $I$, except near the three critical points $k_0$, $\omega k_0$ and $\omega^2k_0$. When obtaining the long-time asymptotics of $M^{LC}$, we only need to consider the neighborhoods of three points.
		
		Let $\epsilon=k_0/2$, denote $B_\epsilon(k_0)=\left\lbrace k\in\mathbb{C}|\left| k-k_0\right|<\epsilon  \right\rbrace $ as the open disk centered at $k_0$ with radius $\epsilon$ and $X:=k_0+\cup_{j=1}^4X_j=k_0+\cup_{j=1}^4s\,{\rm{e}}^{\frac{(2j-1)\pi \ri}{4}}$, $0\leq s<\infty$. Given $B_\epsilon=\cup_{j=0}^2B_\epsilon(\omega^j k_0)$,  $X^{\omega^j\epsilon}=\omega^jX\cap B_\epsilon(\omega^j k_0)$ for $j=0,1,2$ and $X_\epsilon=\cup_{j=0}^2 X^{\omega^j\epsilon}$ (see Fig. \ref{figGammak_0}). Perform the following scaling transformation on the variable $k$, and introduce the new variable $z=z(\zeta,k)$:
		\begin{equation*}
			z=3^{1/4}\sqrt{2 t}(k-k_0),\quad k\in B_\epsilon(k_0).
		\end{equation*}
		
		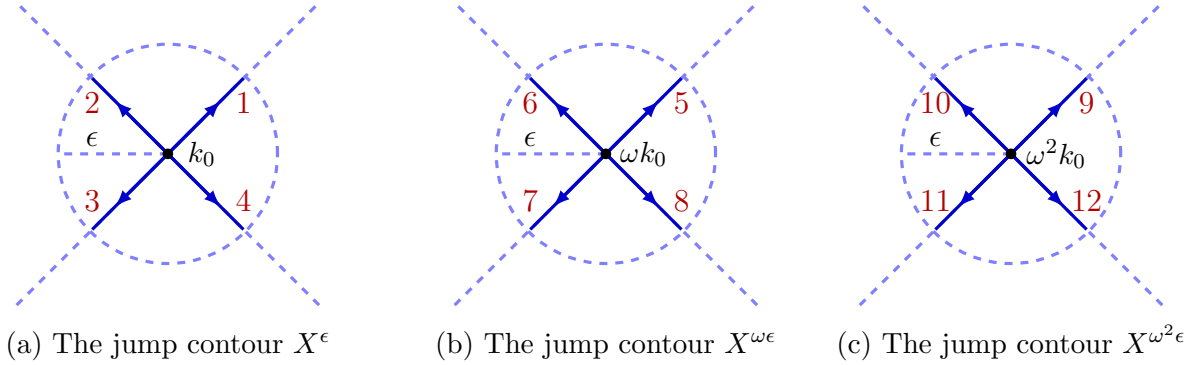
\begin{figure}[htbp]
			\centering
			\begin{subfigure}[b]{0.35\textwidth}
				\centering
				\begin{tikzpicture}
					\draw[dashed,very thick,white!50!blue] (-2,-2) -- (2,2);
					\draw[dashed,very thick,white!50!blue] (2,-2) -- (-2,2);
					
					\draw[dashed,very thick,white!50!blue] (0,0) -- (-1.45,0);
					
					\draw [very thick,black!20!blue](-1,-1) -- (1,1);
					\draw [very thick,black!20!blue](-1,1) -- (1,-1);
					
					\draw[very thick, black!20!blue, -latex]  (0,0) -- (0.7,0.7);
					\draw[very thick, black!20!blue, -latex]  (0,0) -- (-0.7,0.7);
					\draw[very thick, black!20!blue, -latex]  (0,0) -- (0.7,-0.7);
					\draw[very thick, black!20!blue, -latex]  (0,0) -- (-0.7,-0.7);
					
					\draw [dashed,very thick,white!50!blue] (0,0) circle (1.45);
					\fill (0,0) circle (2pt);
					
					\node at (0.45,0) {\small$k_0$};
					\node at (-1,0.2) {$\epsilon$};
					
					\node[red!70!black,above] at (1,0.4) {$1$};
					\node[red!70!black,above] at (-1,0.4) {$2$};
					\node[red!70!black,above] at (-1,-0.9) {$3$};
					\node[red!70!black,above] at (1,-0.9) {$4$};
					
				\end{tikzpicture}
				\caption{The jump contour $X^\epsilon$}
			\end{subfigure}
			\begin{subfigure}[b]{0.3\textwidth}
				\centering
				\begin{tikzpicture}
					\draw[dashed,very thick,white!50!blue] (-2,-2) -- (2,2);
					\draw[dashed,very thick,white!50!blue] (2,-2) -- (-2,2);
					
					\draw[dashed,very thick,white!50!blue] (0,0) -- (-1.45,0);
					
					\draw [very thick,black!20!blue](-1,-1) -- (1,1);
					\draw [very thick,black!20!blue](-1,1) -- (1,-1);
					
					\draw[very thick, black!20!blue, -latex]  (0,0) -- (0.7,0.7);
					\draw[very thick, black!20!blue, -latex]  (0,0) -- (-0.7,0.7);
					\draw[very thick, black!20!blue, -latex]  (0,0) -- (0.7,-0.7);
					\draw[very thick, black!20!blue, -latex]  (0,0) -- (-0.7,-0.7);
					
					\draw [dashed,very thick,white!50!blue] (0,0) circle (1.45);
					\fill (0,0) circle (2pt);
					
					\node at (0.5,0) {\small$\omega k_0$};
					\node at (-1,0.2) {$\epsilon$};
					
					\node[red!70!black,above] at (1,0.4) {$5$};
					\node[red!70!black,above] at (-1,0.4) {$6$};
					\node[red!70!black,above] at (-1,-0.9) {$7$};
					\node[red!70!black,above] at (1,-0.9) {$8$};
					
				\end{tikzpicture}
				\caption{The jump contour $X^{\omega\epsilon}$}
			\end{subfigure}
			\begin{subfigure}[b]{0.3\textwidth}
				\centering
				\begin{tikzpicture}
					\draw[dashed,very thick,white!50!blue] (-2,-2) -- (2,2);
					\draw[dashed,very thick,white!50!blue] (2,-2) -- (-2,2);
					
					\draw[dashed,very thick,white!50!blue] (0,0) -- (-1.45,0);
					
					\draw [very thick,black!20!blue](-1,-1) -- (1,1);
					\draw [very thick,black!20!blue](-1,1) -- (1,-1);
					
					\draw[very thick, black!20!blue, -latex]  (0,0) -- (0.7,0.7);
					\draw[very thick, black!20!blue, -latex]  (0,0) -- (-0.7,0.7);
					\draw[very thick, black!20!blue, -latex]  (0,0) -- (0.7,-0.7);
					\draw[very thick, black!20!blue, -latex]  (0,0) -- (-0.7,-0.7);
					
					\draw [dashed,very thick,white!50!blue] (0,0) circle (1.45);
					\fill (0,0) circle (2pt);
					
					\node at (0.6,0) {\small$\omega^2k_0$};
					\node at (-1,0.2) {$\epsilon$};
					
					\node[red!70!black,above] at (1,0.4) {$9$};
					\node[red!70!black,above] at (-1,0.4) {$10$};
					\node[red!70!black,above] at (-1,-0.9) {$11$};
					\node[red!70!black,above] at (1,-0.9) {$12$};
					
				\end{tikzpicture}
				\caption{The jump contour $X^{\omega^2\epsilon}$}
			\end{subfigure}
			\caption{The jump contours $X_\epsilon$ in the complex $k$-plane.}
			\label{figGammak_0}
		\end{figure}
		
		Under the above transformation, the following expression can be written as
		\begin{align*}
			t \Phi_{21}(\zeta,k)=&t\Phi_{21}(\zeta,k_0)+\sqrt{3}\ri t(k-k_0)^2\\
			=&t\Phi_{21}(\zeta,k_0)+\frac{\ri z^2}{2}.
		\end{align*}
		Combining the above expression and Proposition \ref{R_exist}, the jump matrices $v_j^{(3)}$ $(j=1,\cdots,4)$ for $k\in B_\epsilon(k_0)$ can be written in the following form
		\begin{align*}
			&v^{(3)}_1(z)=\begin{pmatrix}
				1 & 0 & 0\\
				\frac{\left( 2\sqrt{3}t\right)^{\ri\nu} {\rm{e}}^{-2\ri\nu \ln_0(z)}{\rm{e}}^{-2\chi{(\zeta,k_0)}}}{\delta_2(\zeta,k_0)\delta_3(\zeta,k_0)}\rho_1^*(k_0)
				{\rm{e}}^{t \Phi_{21}(\zeta,k_0)+\frac{\ri z^2}{2}} & 1 & -\frac{\delta_{1}\delta_2}{\delta_3^2}r_2^*(0){\rm{e}}^{-t \Phi_{32}(\zeta,z)}\\
				0 & 0 & 1
			\end{pmatrix},\\
			&v_2^{(3)}(z)=\begin{pmatrix}
				1 & \frac{\delta_2(\zeta,k_0)\delta_3(\zeta,k_0)}{\left( 2\sqrt{3}t\right)^{\ri\nu}{\rm{e}}^{-2\ri\nu \ln_0(z)}{\rm{e}}^{-2\chi{(\zeta,k_0)}}} r_1(k_0){\rm{e}}^{-t \Phi_{21}(\zeta,k_0)-\frac{\ri z^2}{2}} & -\frac{\delta_2^2}{\delta_1\delta_3}\alpha_2(k){\rm{e}}^{-t \Phi_{31}(\zeta,z)} \\
				0 & 1 & 0 \\
				0 & 0 & 1
			\end{pmatrix},\\
			&v^{(3)}_3(z)=\begin{pmatrix}
				1 & 0 & 0 \\
				-\frac{\left( 2\sqrt{3}t\right)^{\ri\nu}{\rm{e}}^{-2\ri\nu \ln_0(z)}{\rm{e}}^{-2\chi{(\zeta,k_0)}}}{\delta_2(\zeta,k_0)\delta_3(\zeta,k_0)} r_1^*(k_0) {\rm{e}}^{t \Phi_{21}(\zeta,k_0)+\frac{\ri z^2}{2}} & 1 & \frac{\delta_1\delta_2}{\delta_3^2}\alpha_3^*(k){\rm{e}}^{-t \Phi_{32}(\zeta,z)} \\
				0 & 0 & 1
			\end{pmatrix},\\
			&v^{(3)}_4(z)=\begin{pmatrix}
				1 & -\frac{\delta_2(\zeta,k_0)\delta_3(\zeta,k_0)}{\left( 2\sqrt{3}t\right)^{\ri\nu}{\rm{e}}^{-2\ri\nu \ln_0(z)}{\rm{e}}^{-2\chi{(\zeta,k_0)}}} \rho_1(k_0){\rm{e}}^{-t \Phi_{21}(\zeta,k_0)-\frac{\ri z^2}{2}} & \frac{\delta_2^2}{\delta_{1-}\delta_3}r_2(0){\rm{e}}^{-t \Phi_{31}(\zeta,z)}\\
				0 & 1 & 0\\
				0 & 0 & 1
			\end{pmatrix}.
		\end{align*}
		
		Next, denote $T_0(\zeta)=\left( 2\sqrt{3}t\right)^{-\ri\nu}{\rm{e}}^{2\chi{(\zeta,k_0)}}\delta_2(\zeta,k_0)\delta_3(\zeta,k_0)$, use this condition to define the local transformation:
		\begin{equation*}\label{transT}M^{(4)}(x,t,z)=\left\lbrace
			\begin{aligned}
				&M^{LC}(x,t,z(k))T(\zeta),\quad &&k\in B_\epsilon(k_0),\\
				&M^{LC}(x,t,z(k))\mathcal{A}^{-1}T(\zeta)\mathcal{A},\quad &&k\in B_\epsilon(\omega k_0),\\
				&M^{LC}(x,t,z(k))\mathcal{A}^{-2}T(\zeta)\mathcal{A}^2,\quad &&k\in B_\epsilon(\omega^2k_0),\\
				&M^{LC}(x,t,z(k)), && k\in \mathbb{C}\setminus B_\epsilon,
			\end{aligned}\right.
		\end{equation*}
		where
		\begin{equation*}
			T(\zeta)=
				\begin{pmatrix}
					T_0(\zeta)^{\frac{1}{2}}{\rm{e}}^{-\frac{t}{2} \Phi_{21}(\zeta,k_0)}  & 0 & 0\\
					0 & T_0(\zeta)^{-\frac{1}{2}}{\rm{e}}^{\frac{t}{2} \Phi_{21}(\zeta,k_0)} & 0\\
					0 & 0 & 1\\
				\end{pmatrix}.
		\end{equation*}

		Thus, $M^{(4)}$ is not analytic on $X_\epsilon$  as shown in Fig. \ref{figGammak_0}, and satisfies the jump relations $v^{(4)}_j=T^{-1}v^{(3)}_jT$,  $j=1,2,\cdots,4$, for  $k\in X^\epsilon$ and
		\begin{align*}
			v^{(4)}_n(k)=&\mathcal{A}^{-1}v^{(4)}_{n-4}(\omega^2 k)\mathcal{A},\quad && k\in X^{\omega\epsilon}, \,\,\,n=5,6,7,8,\\
			v^{(4)}_m(k)=&\mathcal{A}^{-2}v^{(4)}_{m-4}(\omega k)\mathcal{A}^2,&& k\in X^{\omega^2\epsilon}, \,m=9,10,11,12,
		\end{align*}
		where
		\begin{align*}
			v^{(4)}_1=&\begin{pmatrix}
				1 & 0 & 0\\
				\rho_1^*(k_0){\rm{e}}^{-2\ri\nu \ln_0(z)}
				{\rm{e}}^{\frac{\ri z^2}{2}} & 1 & -\frac{\delta_{1}\delta_2}{\delta_3^2}T_0^{\frac{1}{2}}r_2^*(0){\rm{e}}^{-\frac{t}{2} \Phi_{21}(\zeta,k_0)}{\rm{e}}^{-t \Phi_{32}(\zeta,z)}\\
				0 & 0 & 1
			\end{pmatrix},\\
				v^{(4)}_2=&\begin{pmatrix}
					1 & r_1(k_0){\rm{e}}^{2\ri\nu \ln_0(z)} {\rm{e}}^{-\frac{\ri z^2}{2}} & -\frac{\delta_2^2}{\delta_1\delta_3}T_0^{-\frac{1}{2}}\alpha_2(k){\rm{e}}^{\frac{t}{2} \Phi_{21}(\zeta,k_0)}{\rm{e}}^{-t \Phi_{31}(\zeta,z)} \\
					0 & 1 & 0 \\
					0 & 0 & 1
				\end{pmatrix},\\
				v^{(4)}_3=&\begin{pmatrix}
					1 & 0 & 0 \\
					-r_1^*(k_0){\rm{e}}^{-2\ri\nu \ln_0(z)}  {\rm{e}}^{\frac{\ri z^2}{2}} & 1 & \frac{\delta_1\delta_2}{\delta_3^2}T_0^{\frac{1}{2}}\alpha_3^*(k){\rm{e}}^{-\frac{t}{2} \Phi_{21}(\zeta,k_0)}{\rm{e}}^{-t \Phi_{32}(\zeta,z)} \\
					0 & 0 & 1
				\end{pmatrix},\\
				v^{(4)}_4=&\begin{pmatrix}
					1 & -\rho_1(k_0){\rm{e}}^{2\ri\nu \ln_0(z)} {\rm{e}}^{-\frac{\ri z^2}{2}} & \frac{\delta_2^2}{\delta_1\delta_3}T_0^{-\frac{1}{2}}r_2(0){\rm{e}}^{\frac{t}{2} \Phi_{21}(\zeta,k_0)}{\rm{e}}^{-t \Phi_{31}(\zeta,z)}\\
					0 & 1 & 0\\
					0 & 0 & 1
				\end{pmatrix}.
		\end{align*}
		
		\begin{lem}\label{T_bounded}
			The matrix function $T(x,t)$ is uniformly bounded, meaning it satisfies
			\begin{equation*}
				\sup_{\zeta\in \mathcal{I}}\sup_{t>1}\left|\partial_x^l\, T(x,t)^{\pm1} \right| <C,\quad l=0,1,
			\end{equation*}
			where $C$ is a real positive constant. Moreover, the following can be obtained for $\zeta\in \mathcal{I}$ and $t>1$
			\begin{align*}
				&\left| T_0(\zeta)\right| ={\rm{e}}^{2\pi\nu},\\
				&\left|\partial_xT_0(\zeta) \right| \leq \frac{C\ln t }{t}.
			\end{align*}
		\end{lem}
		\begin{proof}
		   From Proposition \ref{prop_delta}, we can know that
		   \begin{equation*}
		   	\left| \delta_2(\zeta,k_0)\delta_3(\zeta,k_0)\right|=\left|\delta_1(\zeta,\omega^2k_0)\delta_1(\zeta,\omega k_0) \right| =1,
		   \end{equation*}
		   and
		   \begin{equation*}
		   	{\rm{Re}}\,\chi(\zeta,k_0)=\frac{1}{2\pi }\int_{k_0}^{\infty}\pi{\rm{d}}\ln(1-\left|r_1(s)\right|^2 )=-\frac{1}{2}\ln (1-\left|r_1(k_0) \right|^2)=\pi\nu.
		   \end{equation*}
		   Thus,
		   \begin{equation*}
		   	\left| T_0(\zeta)\right| =\left| \left( 2\sqrt{3}t\right)^{-\ri\nu}\right| \left| {\rm{e}}^{2\chi{(\zeta,k_0)}}\right| \left| \delta_2(\zeta,k_0)\delta_3(\zeta,k_0)\right| ={\rm{e}}^{2\pi\nu}.
		   \end{equation*}
		   Therefore, it is also evident that $T(\zeta)$ is uniformly bounded.
		   Furthermore, using Proposition \ref{prop_delta} and relation $\partial_x=\partial_\zeta/t$, we have
		   \begin{align*}
		   	\left| \partial_xT_0(x,t) \right|&= 	\left| T_0(x,t)\partial_x\ln  T_0(x,t)\right|={\rm{e}}^{2\pi\nu}	\left| \partial_x\ln  T_0(x,t)\right|\\
		   	&={\rm{e}}^{2\pi\nu}	\left| \partial_x\left[ -\ri\nu \ln(2\sqrt{3}t)+2\chi{(\zeta,k_0)}+\ln(\delta_2(\zeta,k_0)\delta_3(\zeta,k_0))\right] \right|\\
		   	&\leq C\left(\left|\ln t\, \partial_x\nu \right| +\left|\partial_x \chi{(\zeta,k_0)}\right| +\left|\partial_x \ln(\delta_2(\zeta,k_0)\delta_3(\zeta,k_0))\right| \right) \\
		   	&\leq \frac{C\ln t}{t}.
		   \end{align*}
		\end{proof}
		
		Define $q:=r_1(k_0)$, where $q$ is given in Appendix of Ref. \cite{lenells-wang-good}. Additionally, since ${\rm{Re}}\,{\Phi_{31}}>0$ and ${\rm{Re}}\,{\Phi_{32}}>0$ for $k\in B_\epsilon(k_0)$, ${\rm{e}}^{-t \Phi_{31}}$ and ${\rm{e}}^{-t \Phi_{32}}$ uniformly decay exponentially to zero as $t\to\infty$. This indicates that $v^{(4)}_j$ tends to the jump matrix $v^X_j$ for large $t$, where $v^X_j(x,t,z)$ for $z\in X_j$, $j=1,\cdots,4$ are as follows
		\begin{align*}
			&v^{X}_1=\begin{pmatrix}
				1 & 0 & 0\\
				\frac{q^*}{1-\left| q\right| ^2}z^{-2\ri\nu(q)}
				{\rm{e}}^{\frac{\ri z^2}{2}} & 1 & 0\\
				0 & 0 & 1
			\end{pmatrix},\quad
			&&v^{X}_2=\begin{pmatrix}
				1 & qz^{2\ri\nu(q)} {\rm{e}}^{-\frac{\ri z^2}{2}} & 0 \\
				0 & 1 & 0 \\
				0 & 0 & 1
			\end{pmatrix},\\
			&v^{X}_3=\begin{pmatrix}
				1 & 0 & 0 \\
				-q^*z^{-2\ri\nu(q)} {\rm{e}}^{\frac{\ri z^2}{2}} & 1 & 0 \\
				0 & 0 & 1
			\end{pmatrix},
			&&v^{X}_4=\begin{pmatrix}
				1 & -\frac{q}{1-\left| q\right| ^2}z^{2\ri\nu(q)} {\rm{e}}^{-\frac{\ri z^2}{2}} &0\\
				0 & 1 & 0\\
				0 & 0 & 1
			\end{pmatrix}.
		\end{align*}
		More importantly, $v^X$ is the jump matrix corresponding to the eigenfunction $m^X$, and the behavior of $m^X$ as $z\to\infty$ is given in Appendix RH problem A.1 of Ref. \cite{lenells-wang-good}, with the specific form as follows
		\begin{equation}\label{mX_expand}
			m^X(q,z)=I+\frac{m_1^X(q)}{z}+\mathcal{O}\left(\frac{1}{z^2} \right),\quad z\to\infty,
		\end{equation}
		and
		\begin{equation*}
			m_1^X(q)=\begin{pmatrix}
				0 & \beta_{12} & 0\\
				\beta_{21} & 0 & 0\\
				0 & 0 & 0\\
			\end{pmatrix},
		\end{equation*}
		where $\beta_{12}$ and $\beta_{21}$ are defined by
		\begin{equation*}
			\beta_{12}=\frac{\sqrt{2\pi}{\rm{e}}^{-\pi \ri/4}{\rm{e}}^{-5\pi \nu/2}}{q^*\Gamma(-\ri\nu)},\quad \beta_{21}=\frac{\sqrt{2\pi}{\rm{e}}^{\pi \ri/4}{\rm{e}}^{3\pi \nu/2}}{q\Gamma(\ri\nu)},
		\end{equation*}
		and $\Gamma$ is the Gamma function.
		
		In order to use the solution of the model RH problem for $m^X(q,z)$ to provide the solution for RH problem \ref{rhpMLC} concerning $M^{LC}(x,t,z(k))$, the following definitions are made
		\begin{equation}\label{TMk_01}
			M^{k_0}(x,t,z(k))=T(\zeta)m^X(q(\zeta),z(k))T(\zeta)^{-1},\quad k\in B_\epsilon(k_0),
		\end{equation}
		and it satisfies
		\begin{equation}\label{TMk_02}
			M^{k_0}(x,t,k)=\mathcal{A} M^{k_0}(x,t,\omega k)\mathcal{A}^{-1}=\mathcal{A}^2 M^{k_0}(x,t,\omega^2 k)\mathcal{A}^{-2},\quad k\in B_\epsilon(k_0).
		\end{equation}
		More importantly, $M^{ k_0}\to I$ on $\partial B_\epsilon(\omega^jk_0)$ for $j=0,1,2$, as $t\to \infty$.
		
		For $k\in X^\epsilon$, it can be obtained that
		\begin{equation*}
			v^{k_0}(\zeta,k)=T(\zeta)v^X(q(\zeta),z(k))T(\zeta)^{-1},
		\end{equation*}
		and
			\begin{equation*}
			v^{(3)}-v^{k_0}=Tv^{(4)}T^{-1}-Tv^X T^{-1}=T\left(v^{(4)}-v^X \right) T^{-1}.
		\end{equation*}
		Next, we present a lemma and prove that $M^{k_0}$ is a good asymptotic approximation to $M^{(3)}$ in $B_\epsilon(k_0)$ for large $t$.
		\begin{lem}\label{lem_v3vk0}
			The jump matrices $v^{(3)}$ and $v^{k_0}$ satisfy the following relation for $\zeta\in\mathcal{I}$, $t>1$ and $l=0,1$,
			\begin{equation*}
					\left\|\partial_x^l\left( v^{(3)}-v^{k_0} \right)\right\|_{L^1\cap L^\infty(X^\epsilon)}   \leq \mathcal{O}({\rm{e}}^{-ct}),
			\end{equation*}
			where $c$ is a real positive constant and $\left\|\partial_x^l\left( M^{k_0}(x,t,\cdot)^{-1}-I \right)\right\|_{L^\infty(\partial B_\epsilon(k_0))} =\mathcal{O}\left(t^{-\frac{1}{2}}\right) $. More importantly, as $t\to \infty$
			\begin{equation}\label{Mk_0}
				\frac{1}{2\pi \ri}\int_{\partial  B_\epsilon(k_0)}\left( \left( M^{k_0}(x,t,k)\right) ^{-1}-I \right)\rd k=-\frac{T(\zeta)m^X_1(q(\zeta))T(\zeta)^{-1}}{(2\sqrt{3}t)^{1/2}}+\mathcal{O}\left( \frac{1}{t}\right),
			\end{equation}
			which can be differentiated with respect to $x$ without increasing the error term.
		\end{lem}
		\begin{proof}
			Based on Lemma \ref{T_bounded} and the relationship between $v^{(3)}$ and $v^{k_0}$, the proof of this lemma is transformed into proving the following relationship:
			\begin{equation*}
					\left\|\partial_x^l\left( v^{(4)}_j-v^{X}_j \right)\right\|_{L^1\cap L^\infty(X_j^\epsilon)}   \leq \mathcal{O}({\rm{e}}^{-ct}),\qquad j=1,\cdots,4.
			\end{equation*}
			We will only prove the case where $j=1$, the remaining three cases can be similarly proven.
			
			For $k\in X^\epsilon$, the calculation results in $v^{(4)}_1-v^{X}_1$ having a non-zero element only at the $(2,3)$ position. Next
			\begin{align*}
				\left| (v^{(4)}_1-v^{X}_1)_{23}\right|&=\left|-\frac{\delta_{1}\delta_2}{\delta_3^2}T_0^{\frac{1}{2}}r_2^*(0){\rm{e}}^{-\frac{t}{2} \Phi_{21}(\zeta,k_0)}{\rm{e}}^{-t \Phi_{32}(\zeta,k)} \right| \\
				&= {\rm{e}}^{{\rm{Re}}\,\pi\nu}{\rm{e}}^{-\frac{t}{2} {\rm{Re}}\,\Phi_{21}(\zeta,k_0)}{\rm{e}}^{-t {\rm{Re}}\,\Phi_{32}(\zeta,k)}\\
				&\leq C {\rm{e}}^{-t {\rm{Re}}\,\Phi_{32}(\zeta,k)}, \qquad k\in X_1^\epsilon.
			\end{align*}
			Denote $k=k_0+u\, {\rm{e}}^{\frac{\pi \ri}{4}}$, $u\geq 0$, then
			\begin{equation*}
				{\rm{Re}}\,\Phi_{32}(\zeta,k)={\rm{Re}}\,\Phi_{32}(\zeta,k_0+u\, {\rm{e}}^{\frac{\pi \ri}{4}})=\frac{\sqrt{3}u^2+9k_0^2+6\sqrt{2}k_0u}{2}\geq c(u+k_0)^2.
			\end{equation*}
			Substituting yields for the real positive constant $C$,
			\begin{align*}
				\left\| \left[ v^{(4)}_1(x,t,\cdot)-v^{X}_1(x,t,\cdot)\right] _{23}\right\|_{L^1( X_1^\epsilon)} &\leq C\int_{0}^\epsilon {\rm{e}}^{-ct (u+k_0)^2}\rd u=C\int_{0}^{\frac{k_0}{2}} {\rm{e}}^{-ct (u+k_0)^2}\rd u\\
				&=C\int_{k_0}^{\frac{3k_0}{2}} {\rm{e}}^{-ct v^2}\rd v\leq \mathcal{O}({\rm{e}}^{-ct}),
				\end{align*}
			and
			\begin{align*}
				\left\| \left[ v^{(4)}_1(x,t,\cdot)-v^{X}_1(x,t,\cdot)\right] _{23}\right\|_{L^\infty( X_1^\epsilon)} &\leq C\sup_{u\geq0}{\rm{e}}^{-ct (u+k_0)^2}=\mathcal{O}({\rm{e}}^{-ct}).
			\end{align*}
			
			Next, we estimate $\partial_x(v^{(4)}_1-v^{X}_1)_{23}$. Firstly, note that $r^*_2(0)=1$, hence
			\begin{align*}
			\left| \partial_x(v^{(4)}_1-v^{X}_1)_{23}\right| =&\left|\partial_x\left[ -\frac{\delta_{1}\delta_2}{\delta_3^2}T_0^{\frac{1}{2}}r_2^*(0){\rm{e}}^{-\frac{t}{2} \Phi_{21}(\zeta,k_0)}{\rm{e}}^{-t \Phi_{32}(\zeta,k)} \right] \right| \\
			 \leq&\left| \partial_x\left(\frac{\delta_{1}\delta_2} {\delta_3^2} \right)T_0^{\frac{1}{2}}{\rm{e}}^{-\frac{t}{2} \Phi_{21}(\zeta,k_0)}{\rm{e}}^{-t \Phi_{32}(\zeta,k)} \right| +
			 \left| \frac{\delta_{1}\delta_2}{\delta_3^2}\,\partial_x\left(T_0^{\frac{1}{2}}\right) {\rm{e}}^{-\frac{t}{2} \Phi_{21}(\zeta,k_0)}{\rm{e}}^{-t \Phi_{32}(\zeta,k)} \right| \\
			 &+\left| \frac{\delta_{1}\delta_2}{\delta_3^2}\,T_0^{\frac{1}{2}}\partial_x\left({\rm{e}}^{-\frac{t}{2} \Phi_{21}(\zeta,k_0)}\right) {\rm{e}}^{-t \Phi_{32}(\zeta,k)} \right| +
			 \left| \frac{\delta_{1}\delta_2}{\delta_3^2}\,T_0^{\frac{1}{2}} {\rm{e}}^{-\frac{t}{2} \Phi_{21}(\zeta,k_0)}\partial_x\left({\rm{e}}^{-t \Phi_{32}(\zeta,k)}\right)  \right| \\
			 :=& a_1+a_2+a_3+a_4.
			\end{align*}
			Similar to the calculation of $(v^{(4)}_1-v^{X}_1)_{23}$, the above shows that
			\begin{equation*}
			   a_j\leq C {\rm{e}}^{-t {\rm{Re}}\,\Phi_{32}(\zeta,k)}\leq C {\rm{e}}^{-ct (u+k_0)^2},
			\end{equation*}
			and
			\begin{equation*}
					\left\| \partial_x(v^{(4)}_1-v^{X}_1)_{23}\right\|_{L^1\cap L^\infty( X_1^\epsilon)} \leq \mathcal{O}({\rm{e}}^{-ct}).
			\end{equation*}
			
			Since $\left| z\right|=3^{1/3}\sqrt{2t}\left| k-k_0\right|  $, we can conclude that $z$ approaches infinity as $t\to\infty$ for $k\in \partial B_\epsilon(k_0)$. Hence, equation \eqref{mX_expand} can be written as
			\begin{equation*}
				m^X(q(\zeta),z(k))=I+\frac{m_1^X(q(\zeta))}{3^{1/3}\sqrt{2t}\left(k-k_0\right) }+\mathcal{O}\left(\frac{1}{t} \right),\quad t\to\infty,
			\end{equation*}
			uniformly for $\zeta\in \mathcal{I}$, $k\in \partial B_\epsilon(k_0)$ and this asymptotic formula can be differentiated with
			respect to $x$ without increasing the error term. Combining the relationship \eqref{TMk_01} between $M^{k_0}$ and $m^X$, we obtain
			\begin{equation*}
				M^{k_0}(\zeta,k)^{-1}-I=-\frac{T(\zeta)m_1^X(q(\zeta))T(\zeta)^{-1}}{3^{1/3}\sqrt{2t}\left(k-k_0\right) }+\mathcal{O}\left(\frac{1}{t} \right),\quad t\to\infty,
			\end{equation*}
			uniformly for $\zeta\in \mathcal{I}$, $k\in \partial B_\epsilon(k_0)$ and also can be differentiated with
			respect to $x$ without increasing the error term. In other words,
			$\left\|\partial_x^l\left( M^{k_0}(x,t,\cdot)^{-1}-I \right)\right\|_{L^\infty(\partial B_\epsilon(k_0))} =\mathcal{O}\left(t^{-\frac{1}{2}}\right) $. And equation \eqref{Mk_0} can be obtained by the above equation and Cauchy's formula.
			
		\end{proof}
		\subsection{A small-norm RH problem}
		\ \ \ \
		Based on the definition of $M^{k_0}$ in equations \eqref{TMk_01} and \eqref{TMk_02} on $B_\epsilon$, we can define
		\begin{equation*}
			\hat{M}=\left\lbrace
			\begin{aligned}
				&M^{LC}\mathcal{A}^{-j}\left(M^{k_0} \right)^{-1}\mathcal{A}^j,\quad && k\in B_\epsilon(\omega^jk_0),\,j=0,1,2,\\
				& M^{LC},  && {\rm{elsewhere}}.
			\end{aligned}
			\right.
		\end{equation*}
		Let $\hat{\Gamma}=\Gamma^{(3)}\cup \partial B_\epsilon$ as shown in Fig. \ref{fighatgamma}. The eigenfunction $\hat{M}$ corresponds to the jump relation $\hat{M}_+(x,t,k)=\hat{M}_-(x,t,k)\hat{v}(x,t,k)$, where
		\begin{equation*}
			\hat{v}=\left\lbrace
			\begin{aligned}
				&v^{(3)},\quad && k\in \hat{\Gamma}\setminus \bar{B_\epsilon},\\
				&\left( M^{k_0}\right)^{-1},  && k\in \partial B_\epsilon,\\
				& M_-^{k_0} v^{(3)}\left( M^{k_0}_+\right)^{-1},  && k\in X_\epsilon.
			\end{aligned}
			\right.
		\end{equation*}
		\begin{figure}[htbp]
			\centering
			\begin{tikzpicture}[scale=1.2]
				\draw [dashed,very thick,white!50!blue](4,0) -- (-4,0);
				
				\fill (2,0) circle (1.5pt);
				
				\draw [very thick,black!20!blue](2,0) circle (1cm);

				\draw [very thick,black!20!blue](2/2.732,2*1.732/2.732) -- (2-1.414/2,1.414/2);
				\draw [very thick,black!20!blue](2+1.414/2,-1.414/2) -- (3.6,-1.6);
				
				\draw [very thick,black!20!blue](2,0) -- (2-1.414/2,1.414/2);
				\draw [very thick,black!20!blue](2,0) -- (2+1.414/2,-1.414/2);
				
				\draw [very thick,black!20!blue, latex-](0.9,1.1) -- (2-1.414/2,1.414/2);
				\draw [very thick,black!20!blue, -latex](2+1.414/2,-1.414/2) -- (3.3,-1.3);
				
				\draw [very thick,black!20!blue](2,0) -- (2-1.414/2,-1.414/2);
				\draw [very thick,black!20!blue](2,0) -- (2+1.414/2,1.414/2);
				
				\draw [very thick,black!20!blue,-latex](2,0) -- (1.5,-0.5);
				\draw [very thick,black!20!blue,-latex](2,0) -- (2.5,0.5);

				\draw [very thick,black!20!blue, latex-](0.9,-1.1) -- (2-1.414/2,-1.414/2);
				\draw [very thick,black!20!blue, -latex](2+1.414/2,1.414/2) -- (3.3,1.3);

				\draw [very thick,black!20!blue](2/2.732,2*1.732/2.732) -- (2/2.732,-2*1.732/2.732);
				
				\draw [very thick,black!20!blue](2/2.732,-2*1.732/2.732) -- (2-1.414/2,-1.414/2);
				\draw [very thick,black!20!blue](2+1.414/2,1.414/2) -- (3.6,1.6);
				
				\draw [very thick,black!20!blue,-latex](2,0) -- (1.5,0.5);
				\draw [very thick,black!20!blue,-latex](2,0) -- (2.5,-0.5);

				\draw [very thick,black!20!blue, -latex](2/2.732,0) -- (2/2.732,-0.7);
				\draw [very thick,black!20!blue, latex-](2/2.732,0.7) -- (2/2.732,0);
				
				\draw [very thick,black!20!blue, latex-](1.9,1) -- (2,1);
				
				\fill (2/2.732,0) circle (1pt);

				\coordinate (center) at (0,0);
				\begin{scope}[rotate around={120:(center)}]
					\draw [dashed,very thick,white!50!blue](4,0) -- (-4,0);
					
					\fill (2,0) circle (1.5pt);
					
					\draw [very thick,black!20!blue](2,0) circle (1cm);

					\draw [very thick,black!20!blue](2/2.732,2*1.732/2.732) -- (2-1.414/2,1.414/2);
					\draw [very thick,black!20!blue](2+1.414/2,-1.414/2) -- (3.6,-1.6);
					
				\draw [very thick,black!20!blue](2,0) -- (2-1.414/2,1.414/2);
					\draw [very thick,black!20!blue](2,0) -- (2+1.414/2,-1.414/2);
					
					\draw [very thick,black!20!blue, latex-](0.9,1.1) -- (2-1.414/2,1.414/2);
					\draw [very thick,black!20!blue, -latex](2+1.414/2,-1.414/2) -- (3.3,-1.3);
					
					\draw [very thick,black!20!blue](2,0) -- (2-1.414/2,-1.414/2);
					\draw [very thick,black!20!blue](2,0) -- (2+1.414/2,1.414/2);
					
					\draw [very thick,black!20!blue,-latex](2,0) -- (1.5,-0.5);
					\draw [very thick,black!20!blue,-latex](2,0) -- (2.5,0.5);

					\draw [very thick,black!20!blue, latex-](0.9,-1.1) -- (2-1.414/2,-1.414/2);
					\draw [very thick,black!20!blue, -latex](2+1.414/2,1.414/2) -- (3.3,1.3);

					\draw [very thick,black!20!blue](2/2.732,2*1.732/2.732) -- (2/2.732,-2*1.732/2.732);
					
					\draw [very thick,black!20!blue](2/2.732,-2*1.732/2.732) -- (2-1.414/2,-1.414/2);
					\draw [very thick,black!20!blue](2+1.414/2,1.414/2) -- (3.6,1.6);
					
					\draw [very thick,black!20!blue,-latex](2,0) -- (1.5,0.5);
					\draw [very thick,black!20!blue,-latex](2,0) -- (2.5,-0.5);

					\draw [very thick,black!20!blue, -latex](2/2.732,0) -- (2/2.732,-0.7);
					\draw [very thick,black!20!blue, latex-](2/2.732,0.7) -- (2/2.732,0);
					
					\draw [very thick,black!20!blue, latex-](1.9,1) -- (2,1);
					
					\fill (2/2.732,0) circle (1pt);

				\end{scope}
				
				\coordinate (center) at (0,0);
				\begin{scope}[rotate around={240:(center)}]
					\draw [dashed,very thick,white!50!blue](4,0) -- (-4,0);
					
					\fill (2,0) circle (1.5pt);
					
					\draw [very thick,black!20!blue](2,0) circle (1cm);

					\draw [very thick,black!20!blue](2/2.732,2*1.732/2.732) -- (2-1.414/2,1.414/2);
					\draw [very thick,black!20!blue](2+1.414/2,-1.414/2) -- (3.6,-1.6);
					
					\draw [very thick,black!20!blue](2,0) -- (2-1.414/2,1.414/2);
					\draw [very thick,black!20!blue](2,0) -- (2+1.414/2,-1.414/2);
					
					\draw [very thick,black!20!blue, latex-](0.9,1.1) -- (2-1.414/2,1.414/2);
					\draw [very thick,black!20!blue, -latex](2+1.414/2,-1.414/2) -- (3.3,-1.3);
					
					\draw [very thick,black!20!blue](2,0) -- (2-1.414/2,-1.414/2);
					\draw [very thick,black!20!blue](2,0) -- (2+1.414/2,1.414/2);
					
					\draw [very thick,black!20!blue,-latex](2,0) -- (1.5,-0.5);
					\draw [very thick,black!20!blue,-latex](2,0) -- (2.5,0.5);

					\draw [very thick,black!20!blue, latex-](0.9,-1.1) -- (2-1.414/2,-1.414/2);
					\draw [very thick,black!20!blue, -latex](2+1.414/2,1.414/2) -- (3.3,1.3);

					\draw [very thick,black!20!blue](2/2.732,2*1.732/2.732) -- (2/2.732,-2*1.732/2.732);
					
					\draw [very thick,black!20!blue](2/2.732,-2*1.732/2.732) -- (2-1.414/2,-1.414/2);
					\draw [very thick,black!20!blue](2+1.414/2,1.414/2) -- (3.6,1.6);
					
					\draw [very thick,black!20!blue,-latex](2,0) -- (1.5,0.5);
					\draw [very thick,black!20!blue,-latex](2,0) -- (2.5,-0.5);

					\draw [very thick,black!20!blue, -latex](2/2.732,0) -- (2/2.732,-0.7);
					\draw [very thick,black!20!blue, latex-](2/2.732,0.7) -- (2/2.732,0);
					
					\draw [very thick,black!20!blue, latex-](1.9,1) -- (2,1);
					
					\fill (2/2.732,0) circle (1pt);

				\end{scope}
				
				\node[above] at (2,0.1) {\footnotesize$k_0$};
				\node[right] at (-0.9,1.9) {\footnotesize$\omega k_0$};
				\node[right] at (-0.9,-1.9) {\footnotesize $\omega^2 k_0$};

				\node[right] at (4,0) {${\rm{Re}\,k}$};

				\node[red!70!black,above] at (2.3,0.4) {$1$};
				\node[red!70!black,above] at (1.7,0.4) {$2$};
				\node[red!70!black,below] at (1.7,-0.4) {$3$};
				\node[red!70!black,below] at (2.3,-0.4) {$4$};
				\node[black,above] at (2,1) {\footnotesize $\partial B_\epsilon(k_0)$};
				
				\node[red!70!black,above] at (2.3,0.4) {$1$};
				\node[red!70!black,above] at (1.7,0.4) {$2$};
				\node[red!70!black,below] at (1.7,-0.4) {$3$};
				\node[red!70!black,below] at (2.3,-0.4) {$4$};
				\node[black,above] at (2,1) {\footnotesize $\partial B_\epsilon(k_0)$};
				
				\node[red!70!black,above] at (-1.55,1.3*1.45) {$5$};
				\node[red!70!black,above] at (-1.4,1) {$6$};
				\node[red!70!black,above] at (-0.5,1.1) {$7$};
				\node[red!70!black,right] at (-0.85,2.25) {$8$};
				\node[black,left] at (-1.85,1.2) {\footnotesize $\partial B_\epsilon(\omega k_0)$};
				
				\node[red!70!black,above] at (-1.55,-1.3*1.45) {\small$12$};
				\node[red!70!black,right] at (-1.2,-1.1) {\small$11$};
				\node[red!70!black,below] at (-0.5,-1.1) {\small$10$};
				\node[red!70!black,below] at (-1.05,-2.1) {$9$};
				\node[black,right] at (-0.1,-2.3) {\footnotesize $\partial B_\epsilon(\omega^2 k_0)$};

				\fill (0,0) circle (1pt);
				
			\end{tikzpicture}
			\caption{The jump contour $\hat{\Gamma}$ in the complex $k$-plane.}
			\label{fighatgamma}
		\end{figure}
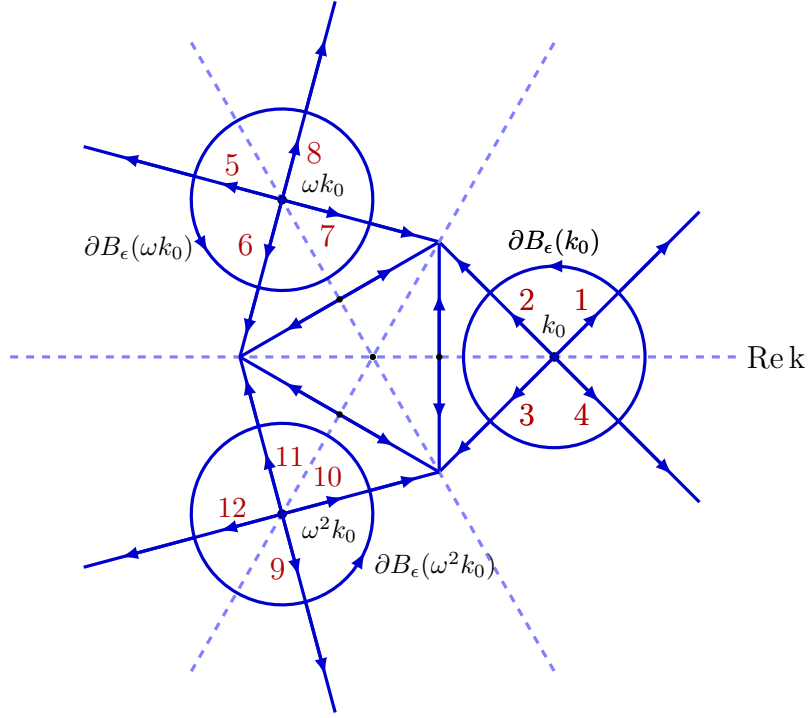
		
		\begin{lem}\label{lem_hatw}
			Let $\hat{w}=\hat{v}-I$ and establish the following estimates for $\zeta\in\mathcal{I}$ and $t>1$:
			\begin{align}
				&\left\|\left( 1+\left|\,\cdot\,\right|  \right) \partial_x^l \hat{w}(x,t,\cdot) \right\|_{L^1\cap L^\infty\left( \hat{\Gamma}\setminus \overline{B_\epsilon}\right)  }\leq \mathcal{O}({\rm{e}}^{-ct}),\label{hatw1}\\
				&\left\|\left( 1+\left|\,\cdot\,\right|  \right) \partial_x^l \hat{w} (x,t,\cdot) \right\|_{L^1\cap L^\infty(X_\epsilon) }\leq \mathcal{O}({\rm{e}}^{-ct}),\label{hatw2}\\
				&\left\|\left( 1+\left|\,\cdot\,\right|  \right) \partial_x^l \hat{w} (x,t,\cdot) \right\|_{L^1\cap L^\infty(\partial B_\epsilon(k_0))}\leq \frac{c}{\sqrt{t}},\label{hatw3}
			\end{align}
			with $l=0,1$.
		\end{lem}
		\begin{proof}
			Due to the presence of symmetry, we only need to prove the case near $k_0$. For $k\in \Gamma^{(3)}_1\setminus\hat\Gamma_1$, the matrix $\hat{w}=v^{(3)}_1-I$, and according to the equation \eqref{v3_1}, we can deduce that ${\rm{Re}}\,\Phi_{21}<0$ and ${\rm{Re}}\,\Phi_{32}>0$. In other words, $\hat{w}$ decays exponentially to the zero matrix uniformly for $k\in \Gamma^{(3)}_1\setminus\hat\Gamma_1$. Combining Lemma \ref{estV56} and using the symmetry \eqref{v3_2} satisfied by the jump matrix, it is clear that we can prove equation \eqref{hatw1}.
			
			For $k\in X_\epsilon$,
			\begin{equation*}
			\hat{w}=M_-^{k_0} v^{(3)}\left( M^{k_0}_+\right)^{-1}-I=M_-^{k_0} \left( v^{(3)}-v^{k_0}\right) \left( M^{k_0}_+\right)^{-1}.
			\end{equation*}
			On one hand, combining Lemma \ref{lem_v3vk0}, and on the other hand, $\partial_\zeta^l M_{\pm}^{k_0}$ and its inverse are uniformly bounded for $k\in X_\epsilon$ and $l=0,1$. Hence, equation \eqref{hatw2} is proved.
			
			The proof of the last equation \eqref{hatw3} can be directly obtained from Lemma \ref{lem_v3vk0}.
		\end{proof}
		
		According to the Beals-Coifman theory \cite{Beals-Coifman-1984}, for a given function $f$ defined on $\hat{\Gamma}$, define the Cauchy operator $\mathcal{C}$ on the contour $\hat{\Gamma}$ as
		\begin{equation*}
			(\mathcal{C}f)(k):=\frac{1}{2\pi \ri}\int_{\hat{\Gamma}}\frac{f(y)}{y-k}\rd y,\quad k\in\mathbb{C}\setminus \hat{\Gamma}.
		\end{equation*}
			If $ f \in \dot{E}^3(\hat{\Gamma})$, then $\mathcal{C}f\in \dot{E}^3(\mathbb{C}\setminus\hat{\Gamma})$, $ f^\pm $ exists almost everywhere on $\hat{\Gamma}$, and $f^\pm \in \dot{L}^3(\hat{\Gamma})$ (the definitions of the symbols and the specific form of the theorem can be found in \cite{lenells-2018}).
			
			Lemma \ref{lem_hatw} shows that
			\begin{equation*}
				\left\|\left( 1+\left|\,\cdot\,\right|  \right) \partial_x^l \hat{w} \ \right\|_{L^1\cap L^\infty(\hat{\Gamma})}\leq \frac{c}{\sqrt{t}},\quad t>1,\, \zeta\in\mathcal{I},\,l=0,1,
			\end{equation*}
			and then the computation follows
			\begin{equation*}
				\left\|\left( 1+\left|\,\cdot\,\right|  \right) \partial_x^l \hat{w} \ \right\|_{L^p(\hat{\Gamma})}\leq \left\|\left( 1+\left|\,\cdot\,\right|  \right) \partial_x^l \hat{w} \ \right\|_{L^1(\hat{\Gamma})}^{1/p}\left\|\left( 1+\left|\,\cdot\,\right|  \right) \partial_x^l \hat{w} \ \right\|_{L^\infty(\hat{\Gamma})}^{(p-1)/p}
				=\frac{c}{\sqrt{t}},
			\end{equation*}
			for each $1\leq p\leq \infty$. The above equation can also be reduced to $\hat{w}\in \left( \dot{L}^3\cap L^\infty\right) ( \hat{\Gamma}) $.
			
			 The trivial decomposition of the jump matrix $\hat{v}$ is considered
		\begin{equation*}
			\hat{v}:=(b_-)^{-1}b_+=I\hat{v},
		\end{equation*} and define $\hat{w}=\hat{w}_++\hat{w}_-$, $\hat{w}_+=b_+-I=\hat{v}-I$, $\hat{w}_-=I-b_-=O$. Then define $\mathcal{C}_{\hat{w}}:\, \dot{L}^3(\hat{\Gamma})+ L^\infty (\hat{\Gamma})\to \dot{L}^3(\hat{\Gamma}) $, and the operator $\mathcal{C}_{\hat{w}} f:=\mathcal{C}_+(f\hat{w}_-)+\mathcal{C}_-(f\hat{w}_+)=\mathcal{C}_-(f\hat{w}_+)$,
		where the Cauchy projection operator
		\begin{equation*}
			(\mathcal{C}_\pm f)(k)=\lim_{{\scriptsize {\begin{aligned}
							&k'\rightarrow k\in\hat{\Gamma}\\[-1ex]
							&k'\in \pm side\, of~ \hat{\Gamma}
			\end{aligned}}}}\frac{1}{2\pi \ri}\int_{\hat{\Gamma}}\frac{f(y)}{y-k'} {\rm{d}} y.
		\end{equation*}
		 It can also be proven that if $f\in \dot{E}^3(\hat{\Gamma}) $, then $\mathcal{C}_{\hat{w}} f$ is the bounded linear operator on $\dot{E}^3(\hat{\Gamma})$. Let $\mu\in I+\dot{E}^3(\hat{\Gamma})$ be the solution of the  equation
		\begin{equation*}
			\mu=I+\mathcal{C}_{\hat{w}} \mu.
		\end{equation*}
		If $I-\mathcal{C}_{\hat{w}}$ is invertible, we can get
		\begin{equation*}
			\mu=I+\left( 1-\mathcal{C}_{\hat{w}}\right)^{-1} \mathcal{C}_{\hat{w}}I\in I+\dot{L}^3(\hat{\Gamma}),
		\end{equation*}
		 and the RH problem for the eigenfunction $\hat{M}$ has the unique solution of the form
		\begin{equation}\label{hatMBC}
		\hat{M}(k)=I+\frac{1}{2 \pi \ri}\int_{\hat{\Gamma}}\frac{\mu(y)\hat{w}(y)}{y-k}{\rm{d}}y,
		\end{equation}
		and the specific proof can be directly obtained from Ref. \cite{lenells-wang-good}.
		\begin{lem}\label{lem_nu}
			For $1<p<\infty$, the following estimate holds for $t\to\infty$
			\begin{equation*}
				\left\|\partial_x^l\left(\mu-I \right)  \right\|_{L^p(\hat{\Gamma})}\leq \frac{c}{\sqrt{t}} ,\quad \zeta\in\mathcal{I},\,l=0,1.
			\end{equation*}
		\end{lem}
		\begin{proof}
			The Neumann series theorem tells us that
			\begin{align*}
				\left\|\left(\mu-I \right)  \right\|_{L^p(\hat{\Gamma})}&\leq \left\|\left( 1-\mathcal{C}_{\hat{w}}\right)^{-1} \mathcal{C}_{\hat{w}}I\right\|_{L^p(\hat{\Gamma})}
				\leq \sum_{j=1}^{\infty} \left\|\mathcal{C}_{\hat{w}} \right\|_{L^p(\hat{\Gamma})\to L^p(\hat{\Gamma})}^{j-1}\left\|\mathcal{C}_{\hat{w}}I \right\|_{L^p(\hat{\Gamma})}\\
				&\leq \sum_{j=1}^{\infty} \left\|\mathcal{C}_- \right\|_{L^p(\hat{\Gamma})\to L^p(\hat{\Gamma})}^{j}\left\|{\hat{w}} \right\|_{L^\infty(\hat{\Gamma})}^{j-1}\left\|{\hat{w}} \right\|_{L^p(\hat{\Gamma})}\\
				&\leq \frac{\left\|\mathcal{C}_- \right\|_{L^p(\hat{\Gamma})\to L^p(\hat{\Gamma})}\left\|{\hat{w}} \right\|_{L^p(\hat{\Gamma})}}{1-\left\|\mathcal{C}_- \right\|_{L^p(\hat{\Gamma})\to L^p(\hat{\Gamma})}\left\|{\hat{w}} \right\|_{L^\infty(\hat{\Gamma})}}
				\leq \frac{c}{\sqrt{t}}.
			\end{align*}
			
			For the case when $l=1$,
			\begin{align*}
				\left\|\partial_x \left(\mu-I \right)  \right\|_{L^p(\hat{\Gamma})}&= \left\|\partial_x \left(\sum_{j=1}^{\infty}\left( \mathcal{C}_{\hat{w}}\right)^jI \right)  \right\|_{L^p(\hat{\Gamma})} \\
				&\leq \sum_{j=2}^{\infty}(j-1)\left\|\mathcal{C}_{\hat{w}} \right\|_{L^p(\hat{\Gamma})\to L^p(\hat{\Gamma})}^{j-2}\left\|\partial_x\mathcal{C}_{\hat{w}} \right\|_{L^p(\hat{\Gamma})\to L^p(\hat{\Gamma})}\left\|\mathcal{C}_{\hat{w}}I \right\|_{L^p(\hat{\Gamma})}\\
				&\quad +\sum_{j=1}^{\infty}\left\|\mathcal{C}_{\hat{w}} \right\|_{L^p(\hat{\Gamma})\to L^p(\hat{\Gamma})}^{j-1}\left\|\partial_x\mathcal{C}_{\hat{w}}I \right\|_{L^p(\hat{\Gamma})}\\
				&\leq  c\frac{\left\|\partial_x \hat{w} \right\|_{L^\infty(\hat{\Gamma})}\left\|\hat{w} \right\|_{L^p(\hat{\Gamma})} + \left\|\partial_x \hat{w} \right\|_{L^p(\hat{\Gamma})}}{1-\left\|\mathcal{C}_- \right\|_{L^p(\hat{\Gamma})\to L^p(\hat{\Gamma})}\left\|{\hat{w}} \right\|_{L^\infty(\hat{\Gamma})}}\leq \frac{c}{\sqrt{t}}.	
			\end{align*}
		\end{proof}

	\subsection{Asymptotics of $M^{LC}$}\label{sub443}
	\ \ \ \
	Calculate the following nontangential limits according to the requirements of the reconstruction formula \eqref{reconst} and equation \eqref{hatMBC}:
		\begin{align*}
			\hat{M}_1(x,t)&=\lim_{k\to\infty}k(\hat{M}(x,t,k)-I)=\lim_{k\to\infty}\frac{k}{2\pi \ri}\int_{\hat{\Gamma}}\frac{\mu(y)\hat{w}(y)}{y-k}\rd y=-\frac{1}{2\pi \ri}\int_{\hat{\Gamma}}{\mu(k)\hat{w}(k)}\rd k\\
			&=-\frac{1}{2\pi \ri}\int_{\partial B_\epsilon}{\hat{w}(k)}\rd k-\frac{1}{2\pi \ri}\int_{\hat{\Gamma}\setminus\partial B_\epsilon}{\hat{w}(k)}\rd k-\frac{1}{2\pi \ri}\int_{\hat{\Gamma}}{(\mu(k)-I)\hat{w}(k)}\rd k\\
			&:=Q_1(x,t)+Q_2(x,t)+Q_3(x,t).
		\end{align*}
		By Lemma \ref{lem_hatw} and Lemma \ref{lem_nu}, we have 
		\begin{equation*}
		\partial_x^lQ_2(x,t)\leq \mathcal{O}({\rm{e}}^{-ct}), \quad \partial_x^lQ_3(x,t)\leq\frac{c}{t},\quad l=0,1.
		\end{equation*}
		Then the above equation can be rewritten as
		\begin{equation}\label{hatM-I}
			\hat{M}_1(x,t)=-\frac{1}{2\pi \ri}\int_{\partial B_\epsilon}{\hat{w}(k)}\rd k+\mathcal{O}(t^{-1}).
		\end{equation}
		
		Based on the equation \eqref{Mk_0}, the function $Q(x,t)$ is defined as follows
		\begin{align*}
				Q(x,t)&=-\frac{1}{2\pi i}\int_{\partial B_\epsilon(k_0)}{\hat{w}(k)}\rd k=-\frac{1}{2\pi i}\int_{\partial  B_\epsilon(k_0)}\left( M^{k_0}(x,t,k)^{-1}-I \right)\rd k\\
				&=\frac{T(\zeta)m^X_1(q(\zeta))T(\zeta)^{-1}}{(2\sqrt{3}t)^{1/2}}+\mathcal{O}\left( \frac{1}{t}\right),  \qquad  t\to\infty.
		\end{align*}
		In addition, $\hat{M}$ also satisfies the symmetry
		\begin{equation*}
			\hat{M}(x,t,k)=\mathcal{A}\hat{M}(x,t,\omega k)\mathcal{A}^{-1},\qquad k\in\mathbb{C}\setminus\hat{\Gamma}.
		\end{equation*}
		That is to say, both $\hat{v}$ and $\hat{w}$ satisfy the aforementioned symmetry, and thus by using equation \eqref{hatM-I}, we can obtain that the following expression holds uniformly for all 
		$\zeta\in\mathcal{I}$ as $t\to\infty$
		\begin{align*}
			\hat{M}_1(x,t)&=-\frac{1}{2\pi \ri}\int_{\partial B_\epsilon}{\hat{w}(k)}\rd k+\mathcal{O}(t^{-1})\\
			&=-\frac{1}{2\pi \ri}\left(\int_{\partial B_\epsilon(k_0)}+\int_{\partial B_\epsilon(\omega k_0)}+\int_{\partial B_\epsilon(\omega^2k_0)} \right) {\hat{w}(k)}\rd k+\mathcal{O}(t^{-1})\\
			&=-\frac{1}{2\pi \ri}\left(\int_{\partial B_\epsilon(k_0)}+\int_{\omega\,\partial B_\epsilon( k_0)}+\int_{\omega^2\,\partial B_\epsilon(k_0)} \right) {\hat{w}(k)}\rd k+\mathcal{O}(t^{-1})\\
			&=Q(x,t)+\omega\mathcal{A}^{-1}Q(x,t)\mathcal{A}+\omega^2\mathcal{A}^{-2}Q(x,t)\mathcal{A}^2+\mathcal{O}(t^{-1})\\
			&=\frac{\sum_{j=0}^{2}\omega^j\mathcal{A}^{-j}T(\zeta)m^X_1(q(\zeta))T(\zeta)^{-1}\mathcal{A}^j}{(2\sqrt{3}t)^{1/2}}+\mathcal{O}\left( \frac{1}{t}\right),
		\end{align*}
		and similarly
		 {\begin{align*}
			\partial_x\left[(\omega,\omega^2,1) \hat{M}_1(x,t)\right]_3 =&\partial_x\left[(\omega,\omega^2,1) \frac{\sum_{j=0}^{2}\omega^j\mathcal{A}^{-j}T(\zeta)m^X_1(q(\zeta))T(\zeta)^{-1}\mathcal{A}^j}{(2\sqrt{3}t)^{1/2}}\right]_3+\mathcal{O}\left( \frac{1}{t}\right)\\
			=&-\frac{3^{5/4}k_0\sqrt{\nu}}{\sqrt{2t}}\sin\left(\frac{19\pi}{12}+\nu\ln \left(6\sqrt{3}tk_0^2 \right) -\sqrt{3}tk_0^2-\arg q-\arg\Gamma\left(\ri\nu \right)  \right.\\
			&+ \left.\frac{1}{\pi} \int_{k_0}^{\infty}\ln\left(\frac{\left|s-k_0 \right| }{\left| s-\omega k_0\right| } \right)\rd \ln\left(1-\left|r_1(s) \right|^2  \right)  \right)+\mathcal{O}\left( \frac{1}{t}\right).
		\end{align*}}
		
		\section{The $\bar\partial$ problems \ref{DbarEG} and \ref{DbarER} }\label{sec_5}
		\ \ \ \
		Reviewing the several transformations we have made to RH problem \ref{rhp_M} earlier, we can obtain
		\begin{equation*}\left\lbrace
			\begin{aligned}
				M^{LC}(k)&=\hat M(k),\\
				M^{RH}(k)&=E^R(k)\hat{M}(k)\left(R^{(3)}(k) \right)^{-1}\Delta^{-1}(k),
			\end{aligned}\right.\quad k\in \mathbb{C}\setminus B_\epsilon.
		\end{equation*}
		In addition, based on the conclusions from the previous sections, the $L^\infty$ norms of $M^{RH}$ and $M^{LC}$ are bounded, and the matrix functions $Y(k)$ and $W(k)$ composed of them also meet different forms in various regions of Fig. \ref{figGamma} and Fig. \ref{figGamma3}. We will also need to utilize the specific forms of these matrices in the proof process below, which will simplify our proof.
		Next, we will estimate the two $\bar{\partial}$ problems \ref{DbarEG} and \ref{DbarER} separately.

		\subsection{The $\bar\partial$ problem \ref{DbarER}: $E^{R}$}
		\ \ \ \
		For the second $\bar{\partial}$ problem \ref{DbarER}, we consider the following problem from equation \eqref{M3factorisation},
		\begin{equation*}
			E^R(x,t,k)=M^{(3)}(x,t,k)M^{LC}(x,t,k)^{-1},
		\end{equation*}
		and it is equivalent to
		\begin{equation}\label{ERexpress}
			E^R(x,t,k) = I - \frac{1}{\pi} \iint_{\mathbb{C}} \frac{\bar{\partial} E^R(k)}{s - k} \rd A(s) =I - \frac{1}{\pi} \iint_{\mathbb{C}} \frac{E^R(s) W(s)}{s - k} \rd A(s),
		\end{equation}
		where $\rd A(s)$ is Lebesgue measure on the plane.
		Define the solid Cauchy operator ${\rm{S}}_1$ with
		\begin{equation*}
			{\rm{S}}_1\left[ f\right] (k)=-\frac{1}{\pi} \iint_{\mathbb{C}} \frac{f(s) W(s)}{s - k} \rd A(s),
		\end{equation*}
		and equation \eqref{ERexpress} can be written in
		\begin{equation}\label{ERsolve}
			(\textbf{1}-{\rm{S}}_1)\left[ E^R(k)\right]=I.
		\end{equation}
		
		The following lemma will illustrate and prove that for sufficiently large $t$, the operator ${\rm{S}}_1$ is of small-norm, ensuring that the inverse of $\left( \textbf{1}-{\rm{S}}_1\right) $ exists and can be expressed in the form of a Neumann series.
		Since the construction of $R^{(3)}$ in the equation \eqref{R3} is more complex, proofs for several additional regions are required here.

		\begin{lem}\label{lem_S1_smallnorm}
			 {For $r_1(k)\in H^{3,4}(0,\infty)$, there exists a constant $C$ such that for $t>0$ and $k\in \mathbb{C}\setminus B_\epsilon$, the operator ${\rm{S}}_1$ satisfies the following  estimates:}
			\begin{align}
				&\left\|{\rm{S}}_1(x,t,\cdot) \right\|_{L^\infty\to L^\infty}\leq C t^{-1/4},\label{S1}\\
				&\left\|\partial_x{\rm{S}}_1[I] \right\|_{L^\infty\to L^\infty}\leq C t^{-1/4}.\label{parxS1}
			\end{align}
		\end{lem}
		
		\begin{proof}
			Based on the distribution of regions $V_j$ $(j=1,2,\cdots,6)$ shown in Fig. \ref{figGamma3}, here we provide proofs for regions $V_1$, $V_2$ and $V_5$. The remaining three regions can be proven by using the conjugate symmetry.
			
			When the proof is restricted to $V_1$, let $f\in L^\infty(V_1)$, $k=k_0+\alpha+\ri\beta$ and  $s=u+\ri v$ $(u>k_0,v>0)$, hence we have
			\begin{align*}
				{\rm{Re}}\,(t\Phi_{21}(\zeta,s))=-2\sqrt{3}tv (u-k_0).
			\end{align*}
			Then according to the third condition of $\bar{\partial}$ problem \ref{DbarER} and Proposition \ref{prop_delta}, we obtain $\left\|\left( M^{LC}\right) ^{\pm1} \right\|_{L^\infty}=\left\|\hat{M}^{\pm1} \right\|_{L^\infty} \leq C$ for $t>1$, $k\in \mathbb{C}\setminus B_\epsilon$ and
			\begin{align*}
				\left| {\rm{S}}_1[f](k)\right| &\leq \iint_{V_1}\frac{\left|f(s)\hat M(s)\bar{\partial}R_1^{(3)}(s)
				\hat M(s) ^{-1} \right| }{\left|s-k \right| }\rd A(s)\\
				&\leq \left\| f \right\|_{L^\infty(V_1)}\left\|\hat M \right\|_{L^\infty(V_1)}\left\|\hat M^{-1} \right\|_{L^\infty(V_1)} \left\| \frac{\delta_{1+}^2}{\delta_2\delta_3} \right\|_{L^\infty(V_1)}\iint_{V_1}\frac{\left|\bar{\partial}R_1(s){\rm{e}}^{t \Phi_{21}} \right|  }{\left|s-k \right| }\rd A(s)\\
				&\leq  \iint_{V_1}\frac{\left( c \left|s-k_0 \right| ^{-1/2}+c \left|\rho'_1({\rm{Re}}\,s) \right|\right)   {\rm{e}}^{-2\sqrt{3}tv (u-k_0)}  }{\left|s-k \right| }\rd A(s)\\
				&:=I_1+I_2.
			\end{align*}
			And $I_1$, $I_2$ respectively satisfy the following estimates. Similar to the Lemma \ref{lem_S1_smallnorm}, for $p>2$ and ${1}/{p}+{1}/{q}=1$,
			\begin{align*}
				I_1 &= \iint_{V_1}\frac{ c \left|s-k_0 \right| ^{-1/2} {\rm{e}}^{-2\sqrt{3}tv (u-k_0)}  }{\left|s-k \right| }\rd A(s)\\
				&\leq C\int_{0}^{\infty}{\rm{e}}^{-2\sqrt{3}tv^2}\left\|(s-k_0)^{-1/2} \right\|_{L_u^p\left( v+k_0,\infty\right) } \left\|(s-k)^{-1} \right\|_{L_u^q\left( v+k_0,\infty\right) }\rd v\\
				&\leq C \int_{0}^{\infty}{\rm{e}}^{-2\sqrt{3}tv^2}v^{1/p-1/2} \left| v-\beta\right|^{1/q-1}\rd v\\
				&\leq Ct^{-1/4},
			\end{align*}
			and
			\begin{align*}
				I_2 &= \iint_{V_1}\frac{ c \left|\rho'_1({\rm{Re}}\,s) \right| {\rm{e}}^{-2\sqrt{3}tv (u-k_0)}  }{\left|s-k \right| }\rd A(s)\\
				&\leq C\int_{0}^{\infty}{\rm{e}}^{-2\sqrt{3}tv^2}   \int_{v+k_0}^{\infty}\frac{\left|\rho'_1({\rm{Re}}\,s) \right|}{\left|s-k \right|}\rd u\rd v\\
				&\leq C\left\|\rho_1' \right\|_{L^2_u(0,\infty)} \int_{0}^{\infty}{\rm{e}}^{-2\sqrt{3}tv^2} \left\|(s-k)^{-1} \right\|_{L_u^2\left(v+k_0,\infty\right) }\rd v\\
				&\leq  Ct^{-1/4}.
			\end{align*}
			Thus, when $k\in V_1$, there exists $\left\|{\rm{S}}_1 \right\|_{L^\infty(V_1)\to L^\infty(V_1)}\leq C t^{-1/4}$.

			For the case $k\in V_2$, we provide a brief proof. Using similar notation as in the previous case, but satisfying $u\in\left(\frac{(\sqrt{3}-1)k_0}{2},k_0 \right) $ and $v\in\left(0,\frac{(3-\sqrt{3})k_0}{2} \right) $ in region $V_2$.
			First, we calculate
			\begin{align*}
				{\rm{Re}}\,(-t\Phi_{31}(\zeta,s))=\frac{3t}{2}(v^2-2uk_0-v^2)+\sqrt{3}tv(u-k_0).
			\end{align*}
			Then
			\begin{align*}
				\left| {\rm{S}}_1[f](k)\right| &\leq \iint_{V_2}\frac{\left|f(s)\hat M(s)\bar{\partial}R_2^{(3)}(s)\hat M^{-1}(s) \right| }{\left|s-k \right| }\rd A(s)\\
				&\leq C \iint_{V_2}\frac{\left|\bar{\partial}R_2(s){\rm{e}}^{-t \Phi_{21}} \right|+ \left|\bar{\partial}\alpha_2(s){\rm{e}}^{-t \Phi_{31}} \right| }{\left|s-k \right| }\rd A(s)\\
				&\leq  \iint_{V_2}\frac{\left(c \left|s-k_0 \right| ^{-1/2}+c \left|r'_1({\rm{Re}}\,s) \right|\right)   \left(   {\rm{e}}^{2\sqrt{3}tv (u-k_0)} +{\rm{e}}^{\frac{3t}{2}(v^2-2uk_0-u^2)+\sqrt{3}tv(u-k_0)}\right)   }{\left|s-k \right| }\rd A(s)\\
				&\leq \int_{0}^{\frac{(3-\sqrt{3})k_0}{2} }\left(  {\rm{e}}^{-2\sqrt{3}tv^2}+{\rm{e}}^{\frac{3-2\sqrt{3}}{2}tv^2}  \right)  \rd v \int_{\frac{(\sqrt{3}-1)k_0}{2}}^{k_0-v} \frac{\left(  c \left|s-k_0 \right| ^{-1/2}+c \left|r'_1({\rm{Re}}\,s) \right|\right)   }{\left|s-k \right| } \rd u\\
				&\leq Ct^{-1/4}.
			\end{align*}
			
			Finally, we complete the proof for region $V_5$. In this region, the constraints $u\in\left(0,\frac{(\sqrt{3}-1)k_0}{2} \right) $ and $v\in\left(0,\frac{(3-\sqrt{3})k_0}{2} \right) $ are satisfied. And
			\begin{align*}
				\left| {\rm{S}}_1[f](k)\right| &\leq \iint_{V_5}\frac{\left|f(s)\hat M(s)\bar{\partial}R_5^{(3)}(s)\hat M^{-1}(s) \right| }{\left|s-k \right| }\rd A(s)\\
				&\leq C \iint_{V_5}\frac{\left|\bar{\partial}R_5(s){\rm{e}}^{-t \Phi_{21}} \right|+ \left|\bar{\partial}\alpha_5(s){\rm{e}}^{-t \Phi_{31}} \right| }{\left|s-k \right| }\rd A(s)\\
				&\leq  \iint_{V_5}\frac{\left( c \left|s \right| ^{-1/2}+c \left|r'_1({\rm{Re}}\,s) \right|\right)  \left(  {\rm{e}}^{2\sqrt{3}tv (u-k_0)} +{\rm{e}}^{\frac{3t}{2}(v^2-2uk_0-u^2)+\sqrt{3}tv(u-k_0)}\right)  }{\left|s-k \right| }\rd A(s)\\
				&\leq \int_{0}^{\frac{(3-\sqrt{3})k_0}{2} }\left(  {\rm{e}}^{(3-2\sqrt{3})k_0tv}+{\rm{e}}^{2tv(v-\sqrt{3}k_0)}  \right) \rd v \int_{\frac{\sqrt{3}}{3}v}^{\frac{(\sqrt{3}-1)k_0}{2}} \frac{\left( c \left|s \right| ^{-1/2}+c \left|r'_1({\rm{Re}}\,s) \right|\right)  }{\left|s-k \right| } \rd u\\
				&\leq \int_{0}^{\frac{(3-\sqrt{3})k_0}{2} }\left(  {\rm{e}}^{(3-2\sqrt{3})k_0tv}+{\rm{e}}^{-\sqrt{3}tv^2}  \right) \rd v \int_{\frac{\sqrt{3}}{3}v}^{\frac{(\sqrt{3}-1)k_0}{2}} \frac{\left( c \left|s \right| ^{-1/2}+c \left|r'_1({\rm{Re}}\,s) \right|\right)  }{\left|s-k \right| } \rd u\\
				&\leq Ct^{-1/4}.
			\end{align*}
			
		 {For $k\in V_1$, based on $r_1(k)\in H^{3,4}(0,\infty)$,} we have $\left\|  \frac{\rd\rho_1(k_0)}{\rd{k_0}}\right\|_{L^2}< \infty$, $\left\|  \frac{\rd\nu(k_0)}{\rd{k_0}}\right\|_{L^2}< \infty$. Combining Proposition \ref{prop_delta}, we can calculate and obtain the following expression
			\begin{align*}
				\left\| \partial_x\left( \bar \partial R_1\right)\right\|_{L^2_u(v+k_0,\infty)}  &=\left\| \frac{\ri{\rm{e}}^{\ri\phi_1}}{\rho}\sin(2\phi_1)\partial_xf_1(s)\right\|_{L^2_u(v+k_0,\infty)}  \leq \left\| \frac{{\rm{e}}^{2\ri\phi_1}}{\rho^2}\sin^2(2\phi_1)\right\|^{1/2}_{L^\infty_u(v+k_0,\infty)}\\
				&\quad\left\| \partial_x\left(\rho_1^*(k_0)(s-k_0)^{-2\ri\nu}{\rm{e}}^{-2\chi{(k_0)}}\delta_2^{-1}(k_0)\delta_3^{-1}(\zeta,k_0)\frac{\delta_2(s)\delta_3(s)}{\delta_{1+}^2(s)} \right)\right\|_{L^2_u(v+k_0,\infty)}  \\
				&\leq \frac{C}{t} \left\| \partial_{k_0}\rho_1^*(k_0)\right\|_{L^2_u(v+k_0,\infty)}\left\| {\rm{e}}^{2[\chi{(s)}-\chi{(k_0)}]}\frac{\delta_2(s)\delta_3(s)}{\delta_2(k_0)\delta_3(k_0)} \right\|_{L^\infty_u(v+k_0,\infty)}  \\
				&\quad + C \left\| \partial_{x}\left({\rm{e}}^{2[\chi{(s)}-\chi{(k_0)}]}\frac{\delta_2(s)\delta_3(s)}{\delta_2(k_0)\delta_3(k_0)} \right)\right\|_{L^\infty_u(v+k_0,\infty)}\left\| \rho_1^*(k_0)\right\|_{L^2_u(v+k_0,\infty)}\\
				&=\mathcal{O}\left(\frac{1}{t} \right).
			\end{align*}
			Then, the calculation is carried out as follows by utilizing the distribution and particularities of the elements in the $W(x,t,k)$ matrix as $t\to\infty$
			\begin{align*}
				\left\|\partial_x {\rm{S}}_1[I]\right\|_{ L^\infty} &\leq C \iint_{V_1}\frac{\left|\partial_xW(x,t,s) \right| }{\left|s-k \right| }\rd A(s)\\
				&\leq C \iint_{V_1}\frac{\max\limits_{i,j=1,2,3}\left|\partial_xW_{ij}(x,t,s) \right| }{\left|s-k \right| }\rd A(s)\\
				&\leq C  \left\| \frac{\delta_{1+}^2}{\delta_2\delta_3} \right\|_{L^\infty(V_1)}\iint_{V_1}\frac{\left|\partial_x\left( \bar{\partial}R_1(s){\rm{e}}^{t \Phi_{21}}\right)  \right|  }{\left|s-k \right| }\rd A(s)\\
				& \quad +C  \left\|\partial_x \frac{\delta_{1+}^2}{\delta_2\delta_3} \right\|_{L^\infty(V_1)}\iint_{V_1}\frac{\left| \bar{\partial}R_1(s){\rm{e}}^{t \Phi_{21}} \right|  }{\left|s-k \right| }\rd A(s)\\
				& \leq C\iint_{V_1}\frac{\left|\partial_x\left( \bar{\partial}R_1(s)\right) {\rm{e}}^{t \Phi_{21}} \right|+\left|s\bar{\partial}R_1(s){\rm{e}}^{t \Phi_{21}}\right|   }{\left|s-k \right| }\rd A(s)\\
				&\quad +\frac{C}{t}\iint_{V_1}\frac{\left| \bar{\partial}R_1(s){\rm{e}}^{t \Phi_{21}} \right|  }{\left|s-k \right| }\rd A(s)\\
				&:=I_3+I_4+I_5,
			\end{align*}
			where
			\begin{align*}
				I_3& \leq\frac{C}{t}\int_{0}^{\infty}\rd v\left( \int_{v+k_0}^{\infty}\frac{\left| {\rm{e}}^{2t \Phi_{21}} \right|}{\left|s-k \right|^2 }\rd u\right) ^{\frac{1}{2}}\\
				&\leq \frac{C}{t}\int_{0}^{\infty}{\rm{e}}^{-4\sqrt{3}tv^2}\left\|(s-k)^{-1} \right\|_{L_u^2\left(v+k_0,\infty\right) }\rd v\\
				&\leq Ct^{-5/4},
			\end{align*}
			\begin{align*}
				I_4&= C\iint_{V_1}\frac{\left|s\bar{\partial}R_1(s){\rm{e}}^{t \Phi_{21}}\right|   }{\left|s-k \right| }\rd A(s)\\
				&\leq C\iint_{V_1} \left| \bar{\partial}R_1(s){\rm{e}}^{t \Phi_{21}}\right|   \rd A(s)+C\iint_{V_1}\frac{\left|k\right|  \left| \bar{\partial}R_1(s){\rm{e}}^{t \Phi_{21}}\right|   }{\left|s-k \right| }\rd A(s)\\
				&\leq Ct^{-3/4}+C\max\left\lbrace \iint_{V_1} \left| \bar{\partial}R_1(s){\rm{e}}^{t \Phi_{21}}\right| \rd A(s),\iint_{V_1}\frac{  \left| \bar{\partial}R_1(s){\rm{e}}^{t \Phi_{21}}\right|   }{\left|s-k \right| }\rd A(s)\right\rbrace \\
				&\leq Ct^{-1/4},
			\end{align*}
			and
			\begin{align*}
				I_5= \frac{C}{t}\iint_{V_1}\frac{\left| \bar{\partial}R_1(s){\rm{e}}^{t \Phi_{21}} \right|  }{\left|s-k \right| }\rd A(s)\leq Ct^{-5/4}.
			\end{align*}
		And in the above calculation, we used $\iint_{V_1} \left| \bar{\partial}R_1(s){\rm{e}}^{t \Phi_{21}}\right|   \rd A(s)\leq Ct^{-3/4}$ proved in Lemma \ref{lem_R_error}. This step can also be obtained in the same way within this lemma. However, since the main calculation result needs to be reflected in Lemma \ref{lem_R_error}, we directly provide the result here.
		
			Overall, we have completed the proof of equation \eqref{parxS1} in region $V_1$.
			And in these two regions $V_2$ and $V_5$, equation \eqref{parxS1} can also be obtained through similar calculations as above.

			Up to now, we have completed the proofs of equations \eqref{S1} and \eqref{parxS1} for the three regions. The proofs for the remaining regions can also be directly obtained based on the above analysis.

		\end{proof}

		Using the conclusions of the above lemma, we can obtain that equation \eqref{ERsolve} is solvable, and
		\begin{equation*}
			E^R=I+\mathcal{O}\left( t^{-\frac{1}{4}}\right).
		\end{equation*}
		Before providing the estimate for the $\bar{\partial}$ problem \ref{DbarER}, it is necessary to prove the following lemma.
		
	\begin{lem}\label{lem_parxER}
		 {For $r_1(k)\in H^{3,4}(0,\infty)$, there exists a constant $C$ such that for $t>0$ and $k\in \mathbb{C}\setminus B_\epsilon$, the following  estimate holds }
			\begin{equation*}
				\left\|\partial_x E^R(x,t,\cdot)\right\|_{L^\infty}\leq Ct^{-1/4} .
			\end{equation*}
		\end{lem}
		
		\begin{proof}
	Based on equation \eqref{ERsolve} and the conclusion of Lemma \ref{lem_S1_smallnorm}, we can obtain
			\begin{align*}
				E^R=(\textbf{1}-{\rm{S}}_1)^{-1}I=I+\sum_{n=0}^{\infty}{\rm{S}}_1^n[I].
			\end{align*}
			For $n\geq1$, thanks to equation \eqref{parxS1} in Lemma \ref{lem_S1_smallnorm}, the following equation holds for all positive integers $n$ and sufficiently large $t$:
			\begin{align*}
				\partial_x\left(  {\rm{S}}_1^n[I]\right)  ={\rm{S}}_{1x}\left(  {\rm{S}}_1^{n-1}[I] \right)  +{\rm{S}}_1\left(  \partial_x\left(  {\rm{S}}_1^{n-1}[I] \right)  \right)\leq n\left(\frac{C}{t^{1/4}} \right)^n.
			\end{align*}
			
		In summary as $t \to \infty$,
			\begin{align*}
				\sum_{n=0}^{\infty}\partial_x\left(  {\rm{S}}_1^n[I]\right)(x,t,k)\leq \sum_{n=0}^{\infty}n\left(\frac{C}{t^{1/4}} \right)^n=Ct^{-1/4} ,
			\end{align*}
			converges uniformly in the region $\zeta\in\mathcal{I}$, and holds
			\begin{align*}
				\left\|\partial_x E^R(x,t,\cdot)\right\|_{L^\infty}=\left\|\partial_x \left(\sum_{n=0}^{\infty}{\rm{S}}_1^n[I] (x,t,\cdot)\right) \right\|_{L^\infty}=\left\| \sum_{n=0}^{\infty}\partial_x\left(  {\rm{S}}_1^n[I]\right)(x,t,\cdot)\right\|_{L^\infty}\leq Ct^{-1/4}.
			\end{align*}

		\end{proof}

		Further expanding equation \eqref{ERexpress} in the region $V_j$ $(j=1,2,\cdots,6)$ into
		\begin{equation*}
			E^R(k)=I+\frac{E^R_1}{k}+\mathcal{O}\left( \frac{1}{k^2}\right),
		\end{equation*}
		where
		\begin{equation*}
			E^R_1=\frac{1}{\pi}\iint_{V_j}E^R(s)W(s)\rd A(s).
		\end{equation*}

	\begin{lem}\label{lem_R_error}
			For sufficiently large $t$ and $\zeta\in \mathcal{I}$, there exists a positive constant $C$ such that the following two estimates hold:
				\begin{align}
					&\left|E^R_1 \right| \leq Ct^{-3/4},\label{E1R}\\
					&\left|\partial_xE^R_1 \right| \leq Ct^{-3/4}.\label{parxE1R}
			\end{align}
		\end{lem}
		\begin{proof}
			The proof can be referred to Lemma \ref{lem_G_error} and Lemma \ref{lem_S1_smallnorm}. Here, we only provide the proof for the case of $V_1$. Let $k=\alpha+\ri\beta$ and  $s=u+\ri v$ $(u>k_0,v>0)$, then the calculation yields
			\begin{align*}
				\left| E^R_1\right| &\leq \frac{1}{\pi}\iint_{V_1} \left|E^R(s) \hat M(s)\bar{\partial}R_1^{(3)}(s)\hat M^{-1}(s) \right|\rd A(s) \\
				&\leq \frac{1}{\pi}\left\| E^R \right\|_{L^\infty(V_1)}\left\|\hat M \right\|_{L^\infty(V_1)}\left\| \hat M^{-1} \right\|_{L^\infty(V_1)}\iint_{V_1}{\left|\bar{\partial}R_1(k)\frac{\delta_{1+}^2}{\delta_2\delta_3}{\rm{e}}^{t \Phi_{21}} \right| }\rd A(s)\\
				&\leq C \iint_{V_1} \left|s-k_0 \right| ^{-1/2}
				{\rm{e}}^{-2\sqrt{3}tv (u-k_0)}\rd A(s)+C \iint_{V_1}\left|\rho'_1({\rm{Re}}\,s) \right|
				{\rm{e}}^{-2\sqrt{3}tv (u-k_0)}\rd A(s)\\
				&:=I_6+I_7.
			\end{align*}
			Using the H\"older's inequality with $2<p<4$ and ${1}/{p}+{1}/{q}=1$, we have
			\begin{align*}
				\left| I_6\right|&\leq   C\int_{0}^{\infty}v^{1/p-1/2}\left( \int_{v+k_0}^{\infty}{\rm{e}}^{-2\sqrt{3}qtv (u-k_0)}\rd u\right)^{1/q} \rd v\\
				&\leq C t^{-1/q} \int_{0}^{\infty}v^{2/p-3/2}{\rm{e}}^{-2\sqrt{3}tv^2} \rd v\\
				&\leq Ct^{-3/4}.
			\end{align*}
			And using the Cauchy-Schwarz inequality
			\begin{align*}
				\left| I_7\right|&\leq   C\int_{0}^{\infty}\left\|\rho'_1 \right\|_{L_u^2\left(v+k_0,\infty\right) }\left( \int_{v+k_0}^{\infty}{\rm{e}}^{-4\sqrt{3}tv (u-k_0)}\rd u\right)^{1/2} \rd v\\
				&\leq C t^{-1/2} \int_{0}^{\infty}\frac{{\rm{e}}^{-2\sqrt{3}tv^2}}{\sqrt{v}} \rd v\\
				&\leq Ct^{-3/4}.
			\end{align*}
			
		When proving the second formula, Lemma \ref{lem_S1_smallnorm} and \ref{lem_parxER} are required. The detailed process is as follows
			\begin{align*}
				\left|\partial_x E^R_1\right| &\leq \frac{1}{\pi}\iint_{V_1} \left|\partial_x  \left( E^R(s)W(s) \right) \right|\rd A(s) \\
				&\leq \frac{1}{\pi}\left\| \partial_xE^R \right\|_{L^\infty(V_1)}\iint_{V_1}{\left|W(s) \right| }\rd A(s)+\frac{1}{\pi}\left\|E^R \right\|_{L^\infty(V_1)}\iint_{V_1}{\left|\partial_xW(s) \right| }\rd A(s)\\
				&\leq C  t^{-1/4} t^{-3/4}+C \iint_{V_1}\max\limits_{i,j=1,2,3}\left|\partial_xW_{ij}(x,t,s) \right| \rd A(s) \\
				&\leq C t^{-1}+  \left\| \frac{\delta_{1+}^2}{\delta_2\delta_3} \right\|_{L^\infty(V_1)}\iint_{V_1}\left|\partial_x\left( \bar{\partial}R_1(s){\rm{e}}^{t \Phi_{21}}\right)  \right|  \rd A(s) \\
				&\quad +  \left\|\partial_x \frac{\delta_{1+}^2}{\delta_2\delta_3} \right\|_{L^\infty(V_1)}\iint_{V_1}\left| \bar{\partial}R_1(s){\rm{e}}^{t \Phi_{21}} \right|  \rd A(s) \\
				& \leq  C t^{-1}+C\iint_{V_1}\left|\partial_x\left( \bar{\partial}R_1(s)\right) {\rm{e}}^{t \Phi_{21}} \right|+\left|s\bar{\partial}R_1(s){\rm{e}}^{t \Phi_{21}}\right|   \rd A(s)\\
				&\quad +\frac{C}{t}\iint_{V_1}\left| \bar{\partial}R_1(s){\rm{e}}^{t \Phi_{21}} \right|  \rd A(s)\\
				& \leq  C t^{-1}+ Ct^{-7/4}+C\iint_{V_1}\left|s\bar{\partial}R_1(s){\rm{e}}^{t \Phi_{21}}\right|   \rd A(s),
			\end{align*}
				where
			\begin{align*}
				\iint_{V_1}&\left|s\bar{\partial}R_1(s){\rm{e}}^{t \Phi_{21}}\right|  \rd A(s)\\
				&\leq C \iint_{V_1} \left| s\right| \left|s-k_0 \right| ^{-1/2}
				{\rm{e}}^{-2\sqrt{3}tv (u-k_0)}\rd A(s)+C \iint_{V_1}\left| s\right|\left|\rho'_1(u) \right|
				{\rm{e}}^{-2\sqrt{3}tv (u-k_0)}\rd A(s)\\
				&:=I_8+I_9.
			\end{align*}
		 The calculation is similar to the one in the theorem mentioned above:
			\begin{align*}
				I_8&= C\int_{0}^{\infty}\rd v\int_{v+k_0}^{\infty}\frac{\sqrt{u^2+v^2}}{\sqrt[4]{(u-k_0)^2+v^2}}{\rm{e}}^{-2\sqrt{3}tv (u-k_0)}\rd u\\
				&=C\int_{0}^{\infty}\rd v\int_{v}^{\infty}\frac{\sqrt{(w+k_0)^2+v^2}}{\sqrt[4]{w^2+v^2}}{\rm{e}}^{-2\sqrt{3}tvw}\rd w\\
				&\leq C\int_{0}^{\infty}\rd v\int_{v}^{\infty}\sqrt{w}{\rm{e}}^{-2\sqrt{3}tvw}\rd w+C\int_{0}^{\infty}\rd v\int_{v}^{\infty}\frac{{\rm{e}}^{-2\sqrt{3}tvw}}{\sqrt{w}}\rd w\\
				&\leq C t^{-3/2}\int_{0}^{\infty}v^{-3/2}\rd v\int_{2\sqrt{3}tv^2}^{\infty}\sqrt{s}e^{-s}\rd s+C \int_{0}^{\infty}\rd v\int_{v}^{\infty}{\rm{e}}^{-2\sqrt{3}tvw}\rd \sqrt{w}\\
				&\leq C t^{-1}\int_{0}^{\infty}v^{-1/2}e^{-2\sqrt{3}tv^2}\rd v+C\int_{0}^{\infty}\sqrt{v}e^{-2\sqrt{3}tv^2}\rd v+Ct\int_{0}^{\infty}v\rd v\int_{v}^{\infty}\sqrt{w}{\rm{e}}^{-2\sqrt{3}tvw}\rd w\\
				&\leq Ct^{-5/4}\int_{0}^{\infty}y^{-1/2}e^{-2\sqrt{3}y^2}\rd y+Ct^{-3/4}\int_{0}^{\infty}\sqrt{y}{\rm{e}}^{-2\sqrt{3}y^2}\rd y\\
				&\quad+Ct^{-1/2}\int_{0}^{\infty}v^{-1/2}\rd v\int_{2\sqrt{3}tv^2}^{\infty}\sqrt{s}{\rm{e}}^{-s}\rd s\\
				&\leq Ct^{-3/4},
			\end{align*}
			and
			\begin{align*}
				I_9&= C \iint_{V_1}\left| s\right|\left|\rho'_1(u) \right|
				{\rm{e}}^{-2\sqrt{3}tv (u-k_0)}\rd A(s)\\
				&\leq C \int_{0}^{\infty}\rd v\int_{v+k_0}^{\infty}u\left|\rho'_1(u) \right|
				{\rm{e}}^{-2\sqrt{3}tv (u-k_0)}\rd u\\
				&\leq C \left\|u\rho'_1(u) \right\|_{L^2(0,\infty)}\int_{0}^{\infty}\left( \int_{v+k_0}^{\infty}
				{\rm{e}}^{-4\sqrt{3}tv (u-k_0)}\rd u\right) ^{1/2}\rd v\\
				&\leq C t^{-1/2}\int_{0}^{\infty}v^{-1/2}{\rm{e}}^{-4\sqrt{3}tv^2}\rd v\\
				&\leq C t^{-3/4}.
			\end{align*}
			Combining all the above estimates, we can get for $k\in V_1$
			\begin{align*}
				\left|\partial_x E^R_1\right|\leq \frac{C}{t^{3/4}},
			\end{align*}
			and finally obtain equations \eqref{E1R} and \eqref{parxE1R} in the region $V_1$.
			The proofs for other regions only require some modifications of the integration intervals in the above process, which will not be elaborated here to save space. The proof of this lemma is now complete.

		\end{proof}
		
		\subsection{The $\bar\partial$ problem \ref{DbarEG}: $E^{G}$}
		\ \ \ \
		For the $\bar{\partial}$ problem \ref{DbarEG}, the treatment process is similar to that of the previous subsection. Moreover, since the region of $G^{(1)}$ in the construction of transformation \eqref{trans1} satisfies certain symmetries, it is only necessary to prove the case for the $D_1$ region in Fig. \ref{figGamma} during the proof process.

		From equation \eqref{M1factorisation}, we obtain the following matrix-valued function
		\begin{equation*}
			E^G(x,t,k)=M^{(1)}(x,t,k)M^{RH}(x,t,k)^{-1},
		\end{equation*}
		and $\bar\partial$ problem \ref{DbarEG} is equivalent to the  integral equation
		\begin{equation}\label{EGexpress}
			E^G(x,t,k) =  I - \frac{1}{\pi} \iint_{\mathbb{C}} \frac{E^G(s) Y(s)}{s - k} \rd A(s).
		\end{equation}
	The above equation can be written in operator notation as follows
		\begin{equation}\label{EGsolve}
			(\textbf{1}-{\rm{S}}_2)\left[ E^G(k)\right]=I,
		\end{equation}
		where
		\begin{equation*}
			{\rm{S}}_2\left[ f\right] (k)=-\frac{1}{\pi} \iint_{\mathbb{C}} \frac{f(s) Y(s)}{s - k} \rd A(s).
		\end{equation*}
		
		\begin{lem}\label{lem_S2_smallnorm}
			 {For $r_2(k)\in H^{3,4}(-\infty,0)$,} there exists a constant $C$ such that for $t>0$, the operator ${\rm{S}}_2$ satisfies the following estimates for $k\in\mathbb{C}\setminus B_\epsilon$
			\begin{equation}\label{parxS2}
				\begin{aligned}
					&\left\|{\rm{S}}_2(x,t,\cdot) \right\|_{L^\infty\to L^\infty}\leq C t^{-1/4},\\
					&\left\| \partial_x{\rm{S}}_2[I] \right\|_{L^\infty}\leq C t^{-1/4}.
				\end{aligned}
			\end{equation}
		\end{lem}
		
		\begin{proof}
			We will only prove the case within region $D_3$, with the proofs for the other regions being similar.
			Let $f\in L^\infty(D_3)$, $k=\alpha+\ri\beta$ and  $s=u+\ri v$ $(u<0,v>0)$,
			Then according to the third condition of $\bar{\partial}$ problem \ref{DbarEG}, we obtain
			\begin{align*}
				\left| {\rm{S}}_2[f](k)\right| &\leq \iint_{D_3}\frac{\left|f(s)\hat{M}(s)\bar{\partial}G^{(1)}_3(s)\hat{M}^{-1}(s) \right| }{\left|s-k \right| }\rd A(s)\\
				&\leq \left\| f \right\|_{L^\infty(D_3)}\left\|\hat{M} \right\|_{L^\infty(D_3)}\left\| \hat{M}^{-1} \right\|_{L^\infty(D_3)}\iint_{D_3}\frac{\left|\bar{\partial}G_3(s){\rm{e}}^{-t\Phi_{21}} \right| }{\left|s-k \right| }\rd A(s)\\
				&\leq  \iint_{D_3}\frac{\left( c \left|s \right| ^{-1/2}+c \left|r'_2({\rm{Re}}\,s) \right|\right)   {\rm{e}}^{2\sqrt{3}tv (u-k_0)}  }{\left|s-k \right| }\rd A(s)\\
				&:=I'_1+I'_2.
			\end{align*}
			Below, we estimate $I'_1$ and $I'_2$ separately. For $p>2$ and ${1}/{p}+{1}/{q}=1$
			\begin{align*}
				I'_1 &= \iint_{D_3}\frac{ c \left|s \right| ^{-1/2} {\rm{e}}^{2\sqrt{3}tv (u-k_0)}  }{\left|s-k \right| }\rd A(s)\\
				&\leq C\int_{0}^{\infty}{\rm{e}}^{-2tv^2}   \int_{-\infty}^{-\frac{\sqrt{3}}{3}v}\frac{1}{\left|s \right| ^{1/2}\left|s-k \right|}\rd u\rd v\\
				&\leq C\int_{0}^{\infty}{\rm{e}}^{-2tv^2}\left\|s^{-1/2} \right\|_{L_u^p\left( -\infty,-\frac{\sqrt{3}}{3}v\right) } \left\|(s-k)^{-1} \right\|_{L_u^q\left( -\infty,-\frac{\sqrt{3}}{3}v\right) }\rd v\\
				&\leq C \int_{0}^{\infty}{\rm{e}}^{-2tv^2}v^{1/p-1/2} \left| v-\beta\right|^{1/q-1}\rd v\\
				&\leq Ct^{-1/4},
			\end{align*}
			and
			\begin{align*}
				I'_2 &= \iint_{D_3}\frac{ c \left|r'_2({\rm{Re}}\,s) \right| {\rm{e}}^{2\sqrt{3}tv (u-k_0)}  }{\left|s-k \right| }\rd A(s)\\
				&\leq C\int_{0}^{\infty}{\rm{e}}^{-2tv^2}   \int_{-\infty}^{-\frac{\sqrt{3}}{3}v}\frac{\left|r'_2({\rm{Re}}\,s) \right|}{\left|s-k \right|}\rd u\rd v\\
				&\leq C\left\|r_2' \right\|_{L_u^2(-\infty,0)} \int_{0}^{\infty}{\rm{e}}^{-2tv^2} \left\|(s-k)^{-1} \right\|_{L_u^2\left( -\infty,-\frac{\sqrt{3}}{3}v\right) }\rd v\\
				&\leq  Ct^{-1/4}.
			\end{align*}
			Thus the proof of $\left\|{\rm{S}}_2 \right\|_{L^\infty(D_3)\to L^\infty(D_3)}\leq C t^{-1/4}$ is completed.
			According to Proposition \ref{prop_G_exist}, both $G_3$ and $\bar \partial G_3$ are independent of the variables $x$ and $t$, which means
				\begin{equation*}
					\partial_x\left( \bar \partial G_3(k)\right)=0.
				\end{equation*}
			Before proceeding with the next step of the proof, we first state several conclusions:
			\begin{align*}
				&\left\| \hat{M}^{\pm1}\right\|_{L^\infty}<\infty, \qquad\quad \left\| \Delta^{\pm1}\right\|_{L^\infty}<\infty,\qquad\,
				\left\|\left(  E^R\right) ^{\pm1}\right\|_{L^\infty}<\infty, \qquad\quad \,\left\|\left(  R^{(3)}\right) ^{\pm1}\right\|_{L^\infty}<\infty,\\
				&\left\| \partial_x\hat{M}^{\pm1} \right\|_{L^\infty}\leq Ct^{-1/2},\, \left\| \partial_x\Delta^{\pm1}\right\|_{L^\infty}\leq Ct^{-1},
				\,\left\| \partial_x\left( E^R\right) ^{\pm1}\right\|_{L^\infty}\leq Ct^{-1/4},\, \left\| \partial_x\left(  R^{(3)}\right) ^{\pm1}\right\|_{L^\infty}<C{\rm{e}}^{-ct},
			\end{align*}
			which can be obtained from subsection \ref{sub443}, Proposition \ref{prop_delta}, Lemmas \ref{lem_R_error} and \ref{R_bounded}.
			Then, as $t\to\infty$, calculate  the following expression
			\begin{align*}
				\left\|\partial_x {\rm{S}}_2[I]\right\|_{ L^\infty} &\leq C \iint_{D_3}\frac{\left|\partial_x\left( \hat{M}\bar{\partial}G^{(1)}_3\hat{M}^{-1} \right)  \right| }{\left|s-k \right| }\rd A(s)\\
				&= C \iint_{D_3}\frac{ \left| \partial_x\left( E^R\hat{M}\left( R^{(3)}\right)^{-1}\Delta^{-1} \bar{\partial}G^{(1)}_3\Delta R^{(3)}\hat{M}^{-1}\left( E^R\right)^{-1} \right)\right|  }{\left|s-k \right| }\rd A(s)\\
				&\leq C \left\| \partial_x E^R\right\|_{L^\infty(D_3)} \iint_{D_3}\frac{  \left| \bar{\partial}G^{(1)}_3\right|  }{\left|s-k \right| }\rd A(s)+C \left\| \partial_x \hat{M}\right\|_{L^\infty(D_3)} \iint_{D_3}\frac{  \left| \bar{\partial}G^{(1)}_3 \right| }{\left|s-k \right| }\rd A(s)\\
				&\quad+C \left\| \partial_x R^{(3)}\right\|_{L^\infty(D_3)} \iint_{D_3}\frac{  \left| \bar{\partial}G^{(1)}_3 \right| }{\left|s-k \right| }\rd A(s)+C \left\| \partial_x\Delta\right\|_{L^\infty(D_3)} \iint_{D_3}\frac{  \left| \bar{\partial}G^{(1)}_3\right|  }{\left|s-k \right| }\rd A(s)\\
				&\quad +C\iint_{D_3}\frac{\left| s\right|   \left| \bar{\partial}G_3\right| {\rm{e}}^{2\sqrt{3}tv(u-k_0)} }{\left|s-k \right| }\rd A(s)\\
				&\leq Ct^{-1/2}+C\int_{0}^{\infty}\rd v\int_{-\infty}^{-\frac{\sqrt{3}}{3}v}\frac{\left| s\right|   \left| \bar{\partial}G_3\right| {\rm{e}}^{2\sqrt{3}tv(u-k_0)} }{\left|s-k \right| }\rd u.
			\end{align*}
			Denote
			\begin{align*}
				I'_3&:=\int_{0}^{\infty}\rd v\int_{-\infty}^{-\frac{\sqrt{3}}{3}v}\frac{\left| s\right|   \left| \bar{\partial}G_3\right| {\rm{e}}^{2\sqrt{3}tv(u-k_0)} }{\left|s-k \right| }\rd u\\
				&\leq C \int_{0}^{\infty}\rd v\int_{-\infty}^{-\frac{\sqrt{3}}{3}v}\ \left| \bar{\partial}G_3\right| {\rm{e}}^{2\sqrt{3}tv(u-k_0)} \rd u+C\int_{0}^{\infty}\rd v\int_{-\infty}^{-\frac{\sqrt{3}}{3}v}\frac{\left| k\right|   \left| \bar{\partial}G_3\right| {\rm{e}}^{2\sqrt{3}tv(u-k_0)} }{\left|s-k \right| }\rd u\\
				&\leq Ct^{-3/4}+C\max\left\lbrace \int_{0}^{\infty}\rd v\int_{-\infty}^{-\frac{\sqrt{3}}{3}v}  \left| \bar{\partial}G_3\right| {\rm{e}}^{2\sqrt{3}tv(u-k_0)} \rd u,\int_{0}^{\infty}\rd v\int_{-\infty}^{-\frac{\sqrt{3}}{3}v}\frac{  \left| \bar{\partial}G_3\right| {\rm{e}}^{2\sqrt{3}tv(u-k_0)} }{\left|s-k \right| }\rd u\right\rbrace \\
				&\leq Ct^{-1/4},
			\end{align*}
			and the specific calculation of $t^{-3/4}$ in the above formula is referred to Lemma \ref{lem_G_error}, where the process is the same. Here, we only provide the result.
		Therefore, we finally obtain $\left\|\partial_x {\rm{S}}_2[I]\right\|_{ L^\infty(D_3)}\leq Ct^{-1/4}$.
		\end{proof}
		
		For sufficiently large $t$, equation \eqref{EGsolve} is solvable, and its solution satisfies
		\begin{equation*}
			E^G=I+\mathcal{O}\left( t^{-\frac{1}{4}}\right).
		\end{equation*}
		
		\begin{lem}\label{lem_parxEG}
		 {For $r_2(k)\in H^{3,4}(-\infty,0)$,} there exists a constant $C$ such that for $t>0$, we have the following  estimate for $k\in \mathbb{C}\setminus B_\epsilon$
				\begin{equation*}
					\left\|\partial_x E^G(x,t,\cdot)\right\|_{L^\infty}\leq Ct^{-1/4} .
				\end{equation*}
		\end{lem}
		
		\begin{proof}
		According to Lemma \ref{lem_S2_smallnorm}, the operator $\textbf{1}-{\rm{S}}_2$ is invertible, which can be derived from equation \eqref{EGsolve}
				\begin{align*}
					E^G=(\textbf{1}-{\rm{S}}_2)^{-1}I=I+\sum_{n=0}^{\infty}{\rm{S}}_2^n[I].
				\end{align*}
				For $n\geq1$, thanks to equation \eqref{parxS2} in Lemma \ref{lem_S2_smallnorm}, the following equation holds
				\begin{align*}
					\partial_x\left(  {\rm{S}}_2^n[I]\right)  ={\rm{S}}_{2x}\left(  {\rm{S}}_2^{n-1}[I] \right)  +{\rm{S}}_2\left(  \partial_x\left(  {\rm{S}}_2^{n-1}[I] \right)  \right)\leq n\left({C}{t^{-1/4}} \right)^n,
			\end{align*}
		for all positive integers $n$ and sufficiently large $t$.
			
		In summary as $t \to \infty$,
				\begin{align*}
					\sum_{n=0}^{\infty}\partial_x\left(  {\rm{S}}_2^n[I]\right)(x,t,k)\leq \sum_{n=0}^{\infty}n\left(\frac{C}{t^{1/4}} \right)^n=Ct^{-1/4} ,
				\end{align*}
				converges uniformly in the region $\zeta\in\mathcal{I}$, and holds
				\begin{align*}
					\left\|\partial_x E^G(x,t,\cdot)\right\|_{L^\infty}=\left\|\partial_x \left(\sum_{n=0}^{\infty}{\rm{S}}_2^n[I] (x,t,\cdot)\right) \right\|_{L^\infty}=\left\| \sum_{n=0}^{\infty}\partial_x\left(  {\rm{S}}_2^n[I]\right)(x,t,\cdot)\right\|_{L^\infty}\leq Ct^{-1/4}.
			\end{align*}

		\end{proof}

		Further expanding equation \eqref{EGexpress} in the region $D_j$ $(j=1,2,\cdots,6)$ into
		\begin{equation*}
			E^G(k)=I+\frac{E^G_1}{k}+\mathcal{O}\left( \frac{1}{k^2}\right),
		\end{equation*}
		where
		\begin{equation*}
			E^G_1=\frac{1}{\pi}\iint_{D_j}E^G(s)Y(s)\rd A(s).
		\end{equation*}

		\begin{lem}\label{lem_G_error}
			For sufficiently large $t$ and $\zeta\in \mathcal{I}$, there exists a constant $C$ such that the following two estimates hold:
			\begin{align*}
				&\left|E^G_1 \right| \leq Ct^{-3/4},\\
				&\left|\partial_xE^G_1 \right| \leq Ct^{-3/4}.
			\end{align*}
		\end{lem}
		\begin{proof}
			For the case of the region $D_3$, let $k=\alpha+\ri\beta$ and  $s=u+\ri v$ $(u<0,v>0)$.
			 Similar to the proof of Lemma \ref{lem_S2_smallnorm}, here
			\begin{align*}
				\left| E^G_1\right| &\leq \frac{1}{\pi}\iint_{D_3} \left|E^G(s)M^{RH}(s)\bar{\partial}G^{(1)}_3(s)\left(M^{RH}(s) \right)^{-1} \right|\rd A(s) \\
				&\leq \frac{1}{\pi}\left\| E^G \right\|_{L^\infty(D_3)}\left\|M^{RH} \right\|_{L^\infty(D_3)}\left\| \left(M^{RH} \right)^{-1} \right\|_{L^\infty(D_3)}\iint_{D_3}{\left|\bar{\partial}G_3(k){\rm{e}}^{-t\Phi_{21}} \right| }\rd A(s)\\
				&\leq C \iint_{D_3} \left|s \right| ^{-1/2}
				{\rm{e}}^{2\sqrt{3}tv (u-k_0)}\rd A(s)+C \iint_{D_3}\left|r'_2({\rm{Re}}\,s) \right|
				{\rm{e}}^{2\sqrt{3}tv (u-k_0)}\rd A(s)\\
				&:=I'_4+I'_5.
			\end{align*}
			Estimate $I'_4$ using the H\"older's inequality with $2<p<4$ and $\frac{1}{p}+\frac{1}{q}=1$
			\begin{align*}
				\left| I'_4\right|&\leq   C\int_{0}^{\infty}v^{1/p-1/2}\left( \int_{-\infty}^{-\frac{\sqrt{3}}{3}v}{\rm{e}}^{2\sqrt{3}qtv (u-k_0)}\rd u\right)^{1/q} \rd v\\
				&\leq C t^{-1/q} \int_{0}^{\infty}v^{2/p-3/2}{\rm{e}}^{-2tv^2} \rd v\\
				&\leq C t^{-3/4} \int_{0}^{\infty}w^{2/p-3/2}{\rm{e}}^{-2w^2}\rd w \\
				&\leq Ct^{-3/4}.
			\end{align*}
			To estimate $I'_5$ using the Cauchy-Schwarz inequality
			\begin{align*}
				\left| I'_5\right|&\leq   C\int_{0}^{\infty}\left\|r'_2 \right\|_{L_u^2\left( -\infty,0\right)}\left( \int_{-\infty}^{-\frac{\sqrt{3}}{3}v}{\rm{e}}^{4\sqrt{3}tv (u-k_0)}\rd u\right)^{1/2} \rd v\\
				&\leq C t^{-1/2} \int_{0}^{\infty}\frac{{\rm{e}}^{-2tv^2}}{\sqrt{v}} \rd v\\
				&\leq C t^{-3/4} \int_{0}^{\infty}\frac{{\rm{e}}^{-2w^2}}{\sqrt{w}}\rd w \\
				&\leq Ct^{-3/4}.
			\end{align*}
			
		Next, based on the conclusions of Lemmas \ref{lem_S2_smallnorm} and \ref{lem_parxEG}, we will prove the second equation of this lemma.
			\begin{align*}
					\left| \partial_xE^G_1\right| &\leq \frac{1}{\pi}\iint_{D_3} \left|\partial_x\left( E^G(s)\hat{M}(s)\bar{\partial}G^{(1)}(s)\hat{M}^{-1}(s) \right)  \right|\rd A(s) \\
					&= C \iint_{D_3} \left| \partial_x\left(E^G E^R\hat{M}\left( R^{(3)}\right)^{-1}\Delta^{-1} \bar{\partial}G^{(1)}_3\Delta R^{(3)}\hat{M}^{-1}\left( E^R\right)^{-1} \right)\right|  \rd A(s)\\
					&\leq C \left\| \partial_x E^R\right\|_{L^\infty(D_3)} \iint_{D_3}  \left| \bar{\partial}G^{(1)}_3\right|  \rd A(s)+C \left\| \partial_x \hat{M}\right\|_{L^\infty(D_3)} \iint_{D_3} \left| \bar{\partial}G^{(1)}_3 \right| \rd A(s)\\
					&\quad+C \left\| \partial_x R^{(3)}\right\|_{L^\infty(D_3)} \iint_{D_3} \left| \bar{\partial}G^{(1)}_3 \right| \rd A(s)+C \left\| \partial_x\Delta\right\|_{L^\infty(D_3)} \iint_{D_3} \left| \bar{\partial}G^{(1)}_3\right|  \rd A(s)\\
					&\quad +C \left\| \partial_x E^G\right\|_{L^\infty(D_3)} \iint_{D_3} \left| \bar{\partial}G^{(1)}_3 \right| \rd A(s)+C\iint_{D_3}\left| s\right|   \left| \bar{\partial}G_3\right| {\rm{e}}^{2\sqrt{3}tv(u-k_0)} \rd A(s)\\
					&\leq Ct^{-1}+C\int_{0}^{\infty}\rd v\int_{-\infty}^{-\frac{\sqrt{3}}{3}v} \left(\sqrt[4]{u^2+v^2} +\sqrt{u^2+v^2}\left| r_2'(u)\right| \right)  {\rm{e}}^{2\sqrt{3}tv(u-k_0)} \rd u\\
					&\leq Ct^{-1}+I'_6+I'_7,
				\end{align*}
				where as $t\to\infty$
			\begin{align*}
					I'_6&:=\int_{0}^{\infty}\rd v\int_{-\infty}^{-\frac{\sqrt{3}}{3}v} \sqrt[4]{u^2+v^2}   {\rm{e}}^{2\sqrt{3}tv(u-k_0)} \rd u\\
					&\leq C\int_{0}^{\infty}\rd v\int_{-\infty}^{-\frac{\sqrt{3}}{3}v} \sqrt[4]{u^2+v^2}   {\rm{e}}^{2\sqrt{3}tvu} \rd u\\
					&\leq C\int_{0}^{\infty}\rd v\int_{\frac{\sqrt{3}}{3}v}^{\infty} \sqrt{w}  {\rm{e}}^{-2\sqrt{3}tvw} \rd w+Ct^{-1}\int_{0}^{\infty}v^{-1/2} {\rm{e}}^{-2tv^2}\rd v\\
					&\leq Ct^{-3/2} \int_{0}^{\infty}v^{-3/2}\rd v\int_{2tv^2}^{\infty} \sqrt{s} {\rm{e}}^{-s} \rd s+Ct^{-5/4}\\
					&\leq C t^{-1} \int_{0}^{\infty}v^{-1/2} {\rm{e}}^{-2tv^2} \rd v+Ct^{-5/4}\\
					&\leq Ct^{-5/4},
				\end{align*}
				and
				\begin{align*}
					I'_7&:=\int_{0}^{\infty}\rd v\int_{-\infty}^{-\frac{\sqrt{3}}{3}v} \sqrt{u^2+v^2}  \left| r_2'(u)\right| {\rm{e}}^{2\sqrt{3}tv(u-k_0)} \rd u\\
					&\leq C \int_{0}^{\infty}\rd v\int_{\frac{\sqrt{3}}{3}v}^{\infty} \sqrt{w^2+v^2}  \left| r_2'(-w)\right| {\rm{e}}^{-2\sqrt{3}tv(w+k_0)} \rd w\\
					&\leq C \int_{0}^{\infty}\rd v\int_{\frac{\sqrt{3}}{3}v}^{\infty} w \left| r_2'(-w)\right| {\rm{e}}^{-2\sqrt{3}tvw} \rd w\\
					&\leq C \left\|ur_2'(u)\right\|_{L^2(-\infty,0)}\int_{0}^{\infty}\left( \int_{\frac{\sqrt{3}}{3}v}^{\infty}  {\rm{e}}^{-4\sqrt{3}tvw} \rd w\right)^{1/2}\rd v\\
					&\leq C t^{-3/4}.
				\end{align*}
				Combining the above two equations, we have proved $\left|\partial_xE^G_1 \right|\leq Ct^{-3/4} $.
		
			So far, the proof for region $D_3$ has been completed. Similar proofs can also be carried out for the other several regions.
		\end{proof}

		\section{Long-time saymptotics of $u(x,t)$}\label{sec_u}
		\ \ \ \
			Recall the reconstruction formula \eqref{reconst}, we have
		 {\begin{align*}
			u(x,t)&=-\frac{3}{2}\frac{\partial}{\partial x}\left(\lim_{k\to\infty}k \left(N_3(x,t,k)  -1 \right) \right)\\
			&=-\frac{3}{2}(\omega,\omega^2,1)\frac{\partial}{\partial x}\left[ \lim_{k\to\infty}k \left(M(x,t,k)  -I\right) \right] _{3}.
		\end{align*}}
		Based on the several transformations, we performed on the eigenfunction $M(x,t,k)$ of the RH problem \ref{rhp_M} in the preceding text
		\begin{align*}
			M=E^{G}E^{R}\hat{M}(R^{(3)})^{-1}\Delta^{-1}\left( G^{(1)}\right) ^{-1}.
		\end{align*}
		The series of preparations made above are all aimed at obtaining the following results. Overall, through the calculations, we have obtained
		 {\begin{align*}
		\frac{\partial}{\partial_x}\lim_{k\to\infty}k \left(M(x,t,k)  -I\right)&=\partial_xE^{G}_1+\partial_xE^{R}_1+\partial_x\hat{M}_1+\partial_x\lim_{k\to\infty}k \left( \Delta^{-1}-I\right)\\
		&\quad +\partial_x\lim_{k\to\infty}k \left( \left( R^{(3)}\right)^{-1}-I\right) +\partial_x\lim_{k\to\infty}\left( \left( G^{(1)}\right) ^{-1}-I\right) \\
		&=\partial_x\hat{M}_1+\mathcal{O}(t^{-3/4}),
		\end{align*}}
		and
		 {\begin{align*}
			u(x,t)&=-\frac{3}{2}	\partial_x\left[(\omega,\omega^2,1) \hat{M}_1(x,t)\right]_3+\mathcal{O}\left(t^{-3/4}\right)\\ &=-\frac{3^{5/4}k_0\sqrt{\nu}}{\sqrt{2t}}\sin\left(\frac{19\pi}{12}+\nu\ln \left(6\sqrt{3}tk_0^2 \right) -\sqrt{3}tk_0^2-\arg q-\arg\Gamma\left(i\nu \right)  \right.\\
			&\quad + \left.\frac{1}{\pi} \int_{k_0}^{\infty}\ln\left(\frac{\left|s-k_0 \right| }{\left| s-\omega k_0\right| } \right)\rd \ln\left(1-\left|r_1(s) \right|^2  \right)  \right)+\mathcal{O}\left(t^{-3/4}\right),
		\end{align*}}
		uniformly for $\zeta\in\mathcal{I}$. The proof of Theorem \ref{theo_asymptotics} is completed.

		\bigskip
		
		\noindent{\bf Data availability}
		Data sharing not applicable to this article as no data sets were generated or analysed during the current study.
		
		\subsection*{Statements and Declarations}
		
		\noindent{\bf Competing Interests}
		The authors declare that they have no conflict of interest.

		\end{document}